\documentclass[12pt]{article}
\usepackage[english]{babel}    
\usepackage[ansinew]{inputenc} 
\usepackage[T1]{fontenc}       
\usepackage{lmodern,microtype} 
\usepackage{amsmath,amsthm,amssymb,amsfonts} 
\usepackage{dsfont,mathrsfs,ushort} 
\usepackage{titlesec,titling}  
\usepackage[nohead]{geometry}  
\usepackage{setspace,caption}  
\usepackage{enumitem,booktabs} 
\usepackage[comma]{natbib} 
\usepackage{comment}
\usepackage{pstricks,pst-all}   
\usepackage{tikz}
\usepackage{pst-plot,pst-node,pst-3dplot}
\usepackage{hyperref,breakurl}
\usepackage{graphicx}
\usepackage{caption}
\usepackage{subcaption}
\usepackage{longtable}
\usepackage[flushleft]{threeparttable}
\usepackage{algorithmic}
\setenumerate{label=\small(\roman*)}
\hypersetup{colorlinks=true,pdfnewwindow=true,pdfstartview=FitH,%
	pdftitle="Peer Effects",pdfauthor="Nail Kashaev and Natalia Lazzati and Ruli Xiao",%
	urlcolor=blue!90!red!45!black,citecolor=blue!90!red!45!black,linkcolor=red!90!black}
\DeclareCaptionFont{fancy}{\bfseries\sffamily}
\allowdisplaybreaks

\providecommand{\psreset}{\psset{%
		linewidth=0.3pt,linestyle=solid,linecolor=black,
		dotsize=2.5pt,dotsep=2.5pt,arrowsize=4pt,
		fillstyle=none,fillcolor=white,
		showpoints=false,arrows=-,linearc=0,framearc=0,
		hatchsep=2pt,hatchwidth=0.2pt,nodesep=4pt,opacity=1}
	\psset{gridcolor=black!60, subgridcolor=black!30}
}

\psreset

\usepackage{graphicx}
\graphicspath{ {figures/} }

\titleformat{\section}[block]{\centering\large\bfseries\sffamily}{\thesection.}{0.3em}{}
\titleformat{\subsection}[block]{\flushleft\bfseries}{\thesubsection.}{0.3em}{}
\titleformat{\subsection}[block]{\flushleft\bfseries\sffamily}{\thesubsection.}{0.2em}{}
\titleformat{\subsubsection}[runin]{\normalsize\itshape}{\bfseries\upshape\sffamily\thesubsubsection.}{0.2em}{}[.--\:]
\renewcommand{\thesubsubsection}{\arabic{section}.\arabic{subsection}.\alph{subsubsection}}
\titlespacing{\section}{0ex}{8ex}{4ex}
\titlespacing{\subsection}{0in}{5ex}{2ex}
\titlespacing{\subsubsection}{0mm}{2ex}{0.4em}
\providecommand{\abstitle}[1]{{\par\vspace*{2ex}\small\bfseries\sffamily #1}\hspace*{1ex}}
\renewenvironment{abstract}%
{\begin{center}\begin{minipage}{0.8\linewidth}%
			\setlength{\parindent}{0.0em}\abstitle{Abstract}\small}%
		{\end{minipage}\end{center}\vfill\clearpage}

\newtheorem{proposition}{Proposition}[section]

\DeclareMathOperator*{\argmax}{arg\,max}

\providecommand{\Real}{{\mathds{R}}}

\providecommand{\tr}{^{\prime}}

\newcommand{\Char}[1]{\mathds{1}\left(\,#1\,\right)}

\newcommand{\abs}[1]{\left\lvert#1\right\rvert}

  \theoremstyle{remark}
  \newtheorem{rem}{\protect\remarkname}
  \theoremstyle{plain}
  
  \theoremstyle{definition}
  
\theoremstyle{plain}

  \theoremstyle{plain}
  
 \theoremstyle{definition}
  \newtheorem{example}{\protect\examplename}
  \theoremstyle{definition}
  \newtheorem{assumption}{\protect\assumptionname}

  \providecommand{\assumptionname}{Assumption}

  \providecommand{\definitionname}{Definition}
  \providecommand{\lemmaname}{Lemma}
  
  \providecommand{\remarkname}{Remark}
\providecommand{\corollaryname}{Corollary}
\providecommand{\theoremname}{Theorem}
\providecommand{\examplename}{Example}

\newcounter{aux}
\newcounter{eg1}
\newcounter{eg2}

\usepackage{xcolor}

\makeatother
\begin{document}
\title{Peer Effects in Consideration and Preferences\thanks{\scriptsize This paper subsumes some results from \citet{kashaev2019peer}. We thank the editor and four anonymous referees for comments and suggestions that have greatly improved the manuscript. We also thank Brice Romuald Gueyap Kounga for excellent RA work. For their useful comments we thank Victor Aguiar, Roy Allen, Tim Conley, Steven Durlauf, Yao Luo, Mathieu Marcoux, Rory McGee, Paola Manzini, Marco Mariotti, Salvador Navarro, \'Aureo de Paula, David Rivers, and Al Slivinski. We would also like to thank the audience at Workshop in Experiments, Revealed Preferences, and Decisions, the Inter University Colloquium on Analytical Research in Economics, the CARE-Colloquium series at University of Kansas, NAWM ES 2021, the Econometrics of Games, Matching and Network at Toulouse, Bristol ESG 2022, MEG 2022, Midwest Economic Theory and International Economics Meetings 2022, 2023 PennState Alumni Conference, CESG 2023, NASM ES 2024, Advances in Demand Analysis 2024, SITE 2025 at Stanford, and the seminar presentations at El Colegio de Mexico, Cornell University, LSE, MSU, Queen Mary University of London, Rice University, UCL, UC Riverside, Washington University in St. Louis, Georgetown University, York University, UNC, Duke University, and University of Virginia. Kashaev gratefully acknowledges financial support from Social Sciences and Humanities Research Council Development and  Insight Development Grants.}}

\author{ 
	Nail Kashaev
	\and
	Natalia Lazzati
    \and
	Ruli Xiao\thanks{Kashaev: University of Western Ontario; \href{mailto:nkashaev@uwo.ca}{nkashaev@uwo.ca}. Lazzati: UC Santa Cruz; \href{mailto:nlazzati@ucsc.edu}{nlazzati@ucsc.edu}.
    Xiao: Indiana University; \href{mailto:rulixiao@iu.edu}{rulixiao@iu.edu}.
	}
	}
\date{\footnotesize{This version: \today\\}
      First version: April 13, 2019, arXiv:1904.06742}
\maketitle

\begin{abstract} 
\noindent \footnotesize{We develop a general model of discrete choice that incorporates peer effects in preferences and consideration sets. We characterize the equilibrium behavior and establish conditions under which all parts of the model can be recovered from a sequence of choices. We allow peers to affect preferences, consideration, or both. We show that these peer-effect mechanisms have different behavioral implications in the data. This allows us to recover the set and the type of connections between the agents in the network. We then use this information to recover each agent's preferences and consideration mechanisms. These nonparametric identification results allow for general forms of heterogeneity across agents and do not rely on the variation of either exogenous covariates or the set of available options (menus). We apply our results to model expansion decisions by tea chains and find evidence of limited consideration. We simulate counterfactual predictions and show how limited consideration slows market penetration and competition.}
\bigskip \bigskip


\end{abstract}

\section{Introduction}\label{sec: introduction}

\noindent Agents are known to be influenced by others when making decisions \citep{durlauf2001social}. This social influence has been shown to be important for people in areas such as health and education and for the decisions of firms such as opening a new store. It has also been argued that agents can influence each other in different ways \citep{manski2000economic}. A comprehensive social influence approach is needed to understand the mechanisms by which the interactions operate in practice and inform private and public policies. We offer a model of social influence where the choices of connected agents or peers shape the subset of options that the agent ends up considering\footnote{This possibility has been (explicitly or implicitly) discussed by other researchers in specific contexts ---e.g., the choices made by friends may help us discover a new television show \citep{godes2004using}, a new welfare program \citep{caeyers2014peer}, a new retirement plan \citep{duflo2003role}, a new restaurant \citep{qiu2018learning}, or an opportunity to protest \citep{enikolopov2020social}.} and her preferences over the alternatives. We show that these two mechanisms have different behavioral implications in the data and that variation over time in the choices made by connected agents allows us to recover all parts of the model. We illustrate our ideas with an empirical application of expansion decisions of the two largest tea chains in the high-end tea industry in China.

In our model, a fixed group of agents are linked through a given network. A linked agent or peer can affect the options an agent considers, how she ranks the alternatives, or both.\footnote{Throughout the paper, we use a behavioral definition of peers: for a given agent, her peers are defined as all other agents that directly impact the choices that the agent makes.} More precisely, at a randomly given time, an agent gets the opportunity to select a new option from a finite set of alternatives. As in standard consideration set models, the agent does not pay attention to all the available options. Instead, she first forms a consideration set and then picks an option from it. The distinctive feature of our model is that the probability that a given alternative enters the consideration set depends on the number of peers currently adopting that option.\footnote{As in most of the literature, we assume that consideration sets are (ex-ante) random. While there are alternative interpretations, one could motivate random preferences by saying that when subjects make a decision, they maximize a well-defined utility function (or preference), but this changes stochastically over time. We interpret random consideration in our model in a similar way. Our set-up allows current consideration to depend on current choices, but does not allow the dependence of current consideration on the previously considered (but not picked) alternatives ---Remark~\ref{rem: D} discusses this important limitation of the model.} We also allow the choices made by her peers to affect the preferences of the agent over alternatives. 

Our model might fit in a large number of applications. In the domain of consumer behavior, we can think of an online platform that offers video games to a set of players (agents).\footnote{\citet{lee2015interpersonal} finds that the likelihood of a player adopting a particular game increases as more of her online friends have previously adopted it.} The number of games offered by the platforms is often quite large, so agents might not be able to pay attention to all of them when making a decision. Platforms often allow agents to form social networks and make the last purchased or played game by peers visible to the agent. This might shrink the subset of games she ends up considering. Some of these games are played in groups, so the choices made by her peers can directly affect the utility the agent gets from playing a particular game. If we observe choices made by the players over time, our model can help the platform personalize each agent's reference groups to maximize profits. While well-suited for many applications, our model requires repeated choices by agents in the network. Thus, this framework does not directly fit, for instance, a situation where alternatives involve durable goods such as different brands of vehicles.

After showing equilibrium existence, we consider a long sequence of choices made by the network members from an invariant, latent network structure and time-stable preferences and consideration set formation mechanisms. We show that all primitives of the model can be uniquely recovered: the network structure, the consideration probabilities, and the distribution of choices given a consideration set (which we associate with preferences). Three aspects of our nonparametric identification results deserve special attention. First, we allow for very general forms of heterogeneity across agents. Second, we identify not only who influences whom, but also whether a peer affects consideration, preferences, or both. Third, we do not rely on changes in exogenous covariates or the set of available options (menus) to identify the model. Instead, we use variation in the choices made by peers. One can think of them as excluded covariates that affect different parts of the decision process.\footnote{We thank Francesca Molinari for pointing this out.} These excluded covariates are special as they vary \emph{endogenously} in the model and we need to \emph{identify} them in the data.

In our model, the observed choices are generated by a system of conditional choice probabilities (CCPs): each CCP specifies the frequency of choices made by a given agent conditional on the choice configuration (i.e., choices of everyone in the network) at the moment of revising her selection. The identification strategy has two parts. First, we identify the primitives of the model from the CCPs. Second, we recover the CCPs from the data.

To identify the model from the CCPs, we proceed in steps. First, we observe that changes in an agent's peers' choices lead to changes in her own CCPs. This helps us identify who the peers are, but not whether their influence comes from affecting consideration or preferences. We separate the two mechanisms using the following feature: for a given agent, the probability of selecting an alternative can be written as the probability of considering the alternative times the probability of selecting it conditional on it being considered. These two terms capture the peer effect channels by the consideration and the preference, respectively. While the first probability changes when a consideration peer switches to that alternative, the second probability remains constant since the agent is already considering the alternative. Also, while the second probability varies when a preference peer chooses something different from the alternative, the first probability is not affected by this change. In other words, an interdependence between alternatives is present in preferences, but not in consideration. We use these observations to separate the two mechanisms via a cross order effect of peers in alternatives in the CCPs.

When the network structure is recovered, choices made by different types of peers can be used as double exclusion restrictions to identify consideration mechanisms and preferences. For example, variation in choices made by peers that only affect consideration (consideration-only peers) can be used to identify changes in consideration probabilities. To recover preferences, we first show that variation in choices made by consideration-only peers can be used to mimic variation in menus. Thus, one can identify the CCPs for the cases in which a subset of alternatives has been completely removed from the menu. This artificial variation in menus generated by consideration-only peers can then be used to identify preferences.

To identify the CCPs, we assume the researcher observes a sequence of choices for the network members in a long time-series with a time-stable structure. We consider two datasets: In continuous-time datasets, the researcher observes agents' choices in real time, so one can identify the CCPs directly. Our empirical application is an example of this type of ``ideal'' dataset. In discrete-time datasets, the researcher observes the joint choices at fixed time intervals (e.g., every Monday). In this case, we use \citet{blevins2017identifying,blevins2018identification} to show that the CCPs are also uniquely identified under a mild condition.

We offer empirically relevant extensions: we study finite history dependence; we explain what to do when one of the choices (e.g., ``do nothing'') is not observed in the data.

To showcase our methodology, we study the possibility of limited attention in the expansion decisions of China's top two high-end tea chains, Nayuki and Heytea. The set of markets in which these firms can open a new store is very large, and we observe their choices near the time when they started their business. As newcomers, they could lack knowledge about current and future market conditions (e.g., the level of consumer demand and local government regulations) or the potential actions of competitors, and thus, not be able to form expectations about market profitability. We think that managers may have employed simple rules (heuristics) to first narrow down the markets they considered, and then used their limited resources to evaluate those options more carefully. This idea builds on previous theories of bounded rationality in firms offered by \citet{simon1955behavioral}.\footnote{See also the discussion of boundedly rational firm behavior in \citet{simon1945administrative}, \citet{armstrong2010behavioral}, and \citet{heidhues2018behavioral}.} We add neighborhood effects to these heuristics by assuming that, for a given market, the number of firm stores in the nearby markets affects the probability of considering that market, but not its profitability. This exclusion restriction allows us to identify the network structure ---neighboring markets that affect consideration only (i.e. consideration-only peers)--- and to recover profits and consideration mechanisms. Our results show that firms in our application limit their consideration, that ignoring this behavior can mislead the analysis of market profitability, and that limited consideration may play a key role in shaping market structure.

Finally, we relate our results with the existing literature. From a modeling perspective, our setup combines the dynamic model of social interactions of \citet{blume1993statistical,blume1995statistical} with the (single-agent) model of random consideration sets of \citet{manski1977structure} and \citet{manzini2014stochastic}. By adding peer effects in consideration sets, we use variation in the choices made by the peers of a given agent as the main tool to recover her random preferences. The literature on the identification of single-agent consideration set models has mainly relied on the variation of the set of available options (menus), e.g., \citet{aguiar2017random, ABDsatisficing16, brady2016menu, caplin2018rational, cattaneo2017random, horan2019random, kashaev2022random, lleras2017more, manzini2014stochastic}, and \citet{masatliogluRA}; or offered partial identification results, e.g., \citet {barseghyan2019heterogeneous}. (See \citealp{aguiar2023random}, for a comparison of several consideration set models in an experiment.) Other papers have relied on covariates that shift preferences or consideration sets, e.g., \citet{barseghyan2019discrete, crawford2021survey, conlon2013demand, draganska2011choice, gaynor2016free, goeree2008limited, mehta2003price, lu2014moment}, and \citet{roberts1991development}. \citet{abaluck2021consumers} show that choice probabilities in full consideration models satisfy a symmetry property analogous to Slutsky symmetry in continuous-choice models that breaks down in consideration set models when changes in characteristics perturb consideration. They use unbounded alternative-specific covariates to generate exogenous menu variation to identify the consideration probabilities. 
\citet{aguiar2020identification, allen2023latent, crawford2021survey}, and \citet{dardanoni2020inferring} use repeated choices (i.e., panel data) but do not allow for peer effects.  

There is a vast econometric literature on the identification of models of social interactions in which friends' choices affect the preferences but not the consideration set of a given person (see, for example, \citealp{blume2011identification, bramoulle2020peer, de2017econometrics} and \citealp{graham2015methods}, for comprehensive reviews of this literature). We add to this literature a second mechanism of peer effects that might be important in specific applications. 

As we mentioned earlier, we can recover from the data the set of connections between the agents in the network. In the context of linear models, a few recent papers have made progress in the same direction. Among them, \citet{blume2015linear, bonaldi2015empirical, de2018recovering, lewbel2023social}, and \citet{manresa2013estimating}. In the context of discrete choice, \citet{chambers2019behavioral} also identify the network structure. However, in this paper, peers do not affect consideration sets but directly change preferences (among other differences).

Two theoretical papers incorporate peer effects in the consideration sets: \citet{borah2018} do so in a static framework and use menu variation for identification. \citet{lazzati2018diffusion} considers a dynamic model but focuses on two alternatives that can be acquired together.

In relation to the application, a large empirical literature addresses how firms make entry and expansion decisions in oligopolies. Part of this literature studies strategic interactions in payoffs in static entry games (e.g., \citealp{Bresnahan1990,berry1992}, and \citealp{ciliberto2009}). These studies typically use cross-sectional data and address the challenge of multiple equilibria across markets. Another branch, more closely related to our work, models expansion decisions as the result of dynamic games (e.g., \citealp{aguirregabiria2007sequential, bajari2007, arcidiacono2016estimation}, and \citealp{blevins2018}). These papers use panel data, often assume a single equilibrium in the data-generating process, and add forward-looking behavior. By capturing the trade-off between the upfront cost and the potential benefits in the future, they show the incentives of firms to preempt their opponents by entering the market early. We study a third possible channel of interdependencies in firms' entry decisions via bounded rationality in managerial decision-making and introduce interaction effects in both payoffs and consideration sets. By modeling sequential revisions in a continuous-time setting, the underlying mechanisms and factors that produce the observed data in our set-up lead to only one equilibrium. Although we abstract away from firms' forward-looking behavior, we allow the history of rivals' actions to affect the firm's payoffs.\footnote{As we mentioned earlier, our application studies the expansion decisions of Nayuki and Heytea when they were relatively new in the market and of moderate size. Thus, a less-sophisticated model might explain their behavior well. We understand that a fully rational, forward-looking approach would be more appropriate to study well-established large-scale companies.} One could view marginal profit in our framework as the reduced-form value from expansion, which embeds forward-looking behaviors without explicitly modeling them. We leave for future research the modeling of limited consideration, multiplicity of equilibria, and forward-looking behavior altogether ---this new model will require a different interpretation.

The rest of the paper is organized as follows. Section~\ref{sec: model} presents the model, the main assumptions, and establishes equilibrium existence. Section~\ref{sec: model identification} studies identification. Section~\ref{sec: extensions} offers some extensions. Section~\ref{sec: model application} applies our model to expansion decisions by firms in the tea market, and Section~\ref{sec: conclusion} concludes. The Online Appendix contains the regularity conditions for identification, all the formal proofs, and extra results for identification and estimation.

\section{The Model}\label{sec: model}
\noindent This section describes the model and the main assumptions in the paper. It also establishes the existence and uniqueness of equilibrium. 

\subsection{Network, Consideration Sets, and Preferences}\label{sec: model 1}
\noindent \textbf{Network and Choice Configuration} There is a finite set of agents $\mathcal{A}=\left\{ 1,2,\dots,A\right\}$, $A\geq 2$, and a finite set (menu) of alternatives $\mathcal{Y}=\left\{ 0,1,2,\dots,Y\right\}$, $Y\geq 1$, from which the agents might choose. Alternative $0$ is called the default alternative. We refer to $\mathbf{y}=\left( y_{a}\right)_{a\in \mathcal{A}}\in\mathcal{Y}^{A}$ as a choice configuration.\footnote{The model easily extends to settings where menus are agent-specific if, for every pair of agents, there is a one-to-one mapping between their choice sets.}  

Agents are connected through a fixed given network and interact with others in different ways. Specifically, the choices made by the peers of a given agent can affect the set of alternatives the agent ends up considering, her preferences over the alternatives, or both. Thus, the network comprises two sets of edges in $\mathcal{A}$, $\Gamma=(\Gamma_{C}, \Gamma_{R})$, where $\Gamma_{C}$ and $\Gamma_{R}$ are the sets of consideration and preference edges, respectively. Each edge identifies two connected agents and the direction of the connection. Hence, the reference group of Agent $a$ consists of reference groups for consideration, $\mathcal{NC}_a$, and for preferences, $\mathcal{NR}_a$. Formally, for each $a\in \mathcal{A}$ 
\[
\mathcal{NC}_{a}=\left\{ a'\in \mathcal{A}: \exists\text{\ edge from $a$ to $a'$ in $\Gamma_{C}$}\right\},  \mathcal{NR}_{a}=\left\{ a'\in \mathcal{A}: \exists\text{ edge from $a$ to $a'$ in $\Gamma_{R}$}\right\}.
\]
The full reference group is $\mathcal{N}_a = \mathcal{NC}_a \bigcup \mathcal{NR}_a$. We follow the convention and assume that $a\not\in\mathcal{N}_a$. Since we allow for the possibility that some peers affect both consideration and preferences, $\mathcal{NCR}_a = \mathcal{NC}_a \bigcap \mathcal{NR}_a$ can be nonempty. 

We will use a simple example throughout the paper to illustrate the main concepts, assumptions, and identification results.

\setcounter{eg1}{\value{example}}
\begin{example}
There are four agents and three alternatives. That is, $\mathcal{A}=\left\{ 1,2,3,4\right\}$ and $\mathcal{Y}=\left\{ 0,1,2\right\}$. The reference groups for consideration are:
\[
\mathcal{NC}_{1}=\left\{2, 3\right\},\quad \mathcal{NC}_{2}=\left\{1\right\},\quad \mathcal{NC}_{3}=\left\{2\right\}, \quad\mathcal{NC}_{4}=\emptyset.
\]
This means that, for instance, Agents 2 and 3 affect the set of alternatives that Agent 1 ends up considering. We assume that only Agents 1 and 3 affect each other's preferences:
\[
\mathcal{NR}_{1}=\left\{3\right\},\quad \mathcal{NR}_{2}=\emptyset,\quad \mathcal{NR}_{3}=\left\{1\right\}, \quad \mathcal{NR}_{4}=\emptyset.
\]
Agent 4 has no peers and Agent 3 affects both preferences and consideration of Agent 1.  
\hfill $\square$

\end{example}

\noindent\textbf{Choice Revision Process} We model the revision process as a standard continuous-time Markov process on the space of choice configurations. We assume agents are endowed with independent Poisson ``alarm clocks'' with rates $\left( \lambda_{a}\right)_{a\in \mathcal{A}}$. At randomly given moments (exponentially distributed with mean $1/\lambda_{a}$) the alarm of Agent $a$ goes off.\footnote{That is, each Agent $a$ is endowed with a collection of random variables $\left\{\tau_{n}^{a}\right\}_{n=1}^{\infty}$ such that $\tau_{n}^{a}-\tau_{n-1}^{a}$ is exponentially distributed with mean $1/\lambda_{a}$. These differences are independent across agents and time. All the identification results in Section~\ref{PIP} are still valid if agents have perfectly synchronized clocks.} When this happens, the agent forms a consideration set and then selects an alternative from it. Since the probability of any two alarm clocks going off simultaneously is zero, the probability that two agents make choices simultaneously is also zero. This fact simplifies identification and rules out multiple equilibria in the data-generating process.

\bigskip

\noindent \textbf{Peer Effects in Consideration Sets} The probability that Agent $a$ pays attention to and, thereby, includes a particular alternative in her consideration set depends on the choice configuration at the moment of revising decisions. We indicate by $\operatorname{Q}_{a}\left( v \mid \mathbf{y},\mathcal{NC}_{a}\right)$ the ex-ante probability that Agent $a$ considers alternative $v$ given a choice configuration $\mathbf{y}$ and her consideration reference group $\mathcal{NC}_{a}$. We assume that each alternative is added to the consideration set independently from other alternatives.
\begin{assumption}[Independent Consideration]\label{ass: MM}
    For each $a\in\mathcal{A}$ and $\mathbf{y}\in \mathcal{Y}^{A}$, the probability of facing consideration set $\mathcal{C}$, which is a subset of menu $\mathcal{Y}$, is
\begin{equation*}
\operatorname{C}_{a}\left( \mathcal{C}  \mid \mathbf{y},\mathcal{NC}_{a},\mathcal{Y}\right)=\prod\nolimits_{v\in \mathcal{C}}\operatorname{Q}_{a}\left( v \mid\mathbf{y},\mathcal{NC}_{a}\right) \prod\nolimits_{v\in \mathcal{Y}\setminus\mathcal{C}}\left( 1-\operatorname{Q}_{a}\left(
v \mid\mathbf{y},\mathcal{NC}_{a}\right) \right).
\end{equation*}
\end{assumption}
Since the consideration set cannot be empty, we assume that the default alternative is always considered. That is, $\operatorname{Q}_{a}\left( 0 \mid\mathbf{y},\mathcal{NC}_{a}\right)=1$ for all $a\in\mathcal{A}$ and $\mathbf{y}\in \mathcal{Y}^{A}$. This restriction is the only one imposed on the default alternative, and it is satisfied in many applications (including the one we present in Section~\ref{sec: model application}). We include $\mathcal{Y}$ as a determinant of $\operatorname{C}_{a}$ to simplify the notation when we later identify counterfactual choice distributions for different menus.

\setcounter{aux}{\value{example}}
\setcounter{example}{\value{eg1}}
\begin{example}[continued]
    Recall that the consideration peers of Agent 1 are $\mathcal{NC}_{1} = \left\{2,3\right\}$. Thus, the probability that Agent 1 considers at the moment of choosing, say, set $\left\{0,1\right\}$  is 
    \begin{equation*}
\operatorname{C}_{1}\left( \left\{0,1\right\}  \mid \mathbf{y},\left\{2,3\right\},\left\{0,1,2\right\}\right)=\operatorname{Q}_{1}\left( 1 \mid\mathbf{y},\left\{2,3\right\}\right) \left[ 1-\operatorname{Q}_{1}\left(
2 \mid\mathbf{y},\left\{2,3\right\}\right) \right].
\end{equation*}
    
\noindent Assume that the attention to option $v$ given by Agent 1 is modeled as follows: $v$ is considered if and only if $c_{1,v}\left(\mathbf{y},\left\{2,3\right\}\right) \geq \varepsilon_{1,v}$, where $c_{1,v}(\mathbf{y},\mathcal{NC}_1)$ measures the mean attention of Agent $1$ to $v$ as a function of her consideration peers and $\varepsilon_{1,v}$ is a random attention shock independent of $\mathbf{y}$. Then, the probability of considering $v$ is
    $
    \operatorname{Q}_{1}\left( v \mid\mathbf{y},\left\{2,3\right\}\right) = 
    \operatorname{F}_{\varepsilon_{1,v}}\left(c_{1,v}\left(\mathbf{y},\left\{2,3\right\}\right)\right),
    $
    where $\operatorname{F}_{\varepsilon_{1,v}}$ is the cumulative distribution function (c.d.f.) of $\varepsilon_{1,v}$.
    \hfill $\square$
\end{example}

\noindent \textbf{Peer Effects in Preferences} 
After the consideration set is formed, the agent selects an alternative according to some choice rule. The choices made by the agent from any consideration set may be random from the researcher's perspective if there are latent preference shocks across choice instances, as in the example below. Choice rules summarize the decision process after the consideration set is formed. Since, in practice, the underlying preferences (utility function parameters) can be identified and estimated from the choice rule, we focus on the choice rule and leave its association with preferences to be flexible. 

We incorporate the peer effect in preferences by allowing the choice rule of agents to depend on the configuration of choices and her preference peers. Formally, given consideration set $\mathcal{C}$, the choice rule $\operatorname{R}_{a}\left( \cdot  \mid \mathbf{y},\mathcal{NR}_{a},\mathcal{C}\right)$ is a discrete probability distribution supported on $\mathcal{C}$. That is, $\operatorname{R}_{a}\left(v  \mid \mathbf{y},\mathcal{NR}_{a},\mathcal{C}\right)\geq0$ for all $v$, and 
$\sum\nolimits_{{v}\in\mathcal{C}}\operatorname{R}_{a}\left(v  \mid \mathbf{y},\mathcal{NR}_{a},\mathcal{C}\right) = 1$.

\setcounter{aux}{\value{example}}
\setcounter{example}{\value{eg1}}
\begin{example}[continued]
Recall the preference peer of Agent 1 is $\mathcal{NR}_{1} = \left\{3\right\}$. Hence, her probability of selecting alternative 1 in consideration set $\left\{0,1\right\}$ is
$
    \operatorname{R}_{1}\left(1  \mid \mathbf{y},\left\{3\right\},\left\{0,1\right\}\right).
$ Assume the utility Agent $1$ gets from $v$ in set $\mathcal{C}$ is given by $u_{1,v,\mathcal{C}}\left(\mathbf{y},\mathcal{NR}_{1}\right)+\xi_{1,v,\mathcal{C}}$, where $u_{1,v,\mathcal{C}}$ captures the mean utility from the alternative in the given consideration set. The vector of agent- and consideration-set-specific taste shocks $\xi_{a,\mathcal{C}}=(\xi_{a,v,\mathcal{C}})_{v\in\mathcal{C}}$ is continuously distributed\footnote{Continuity is only needed to handle cases where two alternatives give the same utility.} with the conditional c.d.f. $\operatorname{F}_{a,\xi}(\cdot\mid \mathbf{y},\mathcal{NR}_a, \mathcal{C})$. Then, for $v\in\mathcal{C}$,
\[
\operatorname{R}_{1}\left(v  \mid \mathbf{y},\left\{3\right\},\mathcal{C}\right)=\int\Char{v=\argmax_{v'\in\mathcal{C}}\{u_{1,v',\mathcal{C}}\left(\mathbf{y},\left\{3\right\}\right)+\xi_{1, v',\mathcal{C}}\}}d\operatorname{F}_{1,\xi}(\xi_{1,\mathcal{C}}\mid  \mathbf{y},\left\{3\right\},\mathcal{C}),
\]
where $\Char{\cdot}$ is the indicator function. If the $\xi_{a,v,\mathcal{C}}$s are independent and identically distributed (i.i.d.) shocks, distributed according to the standard Type I extreme value distribution, then the above expressions simplify to
\[ \quad \quad \quad \quad \quad \quad  \quad \quad
    \operatorname{R}_{1}\left(v  \mid \mathbf{y},\left\{3\right\},\mathcal{C}\right)=\dfrac{\exp\left(u_{1,v,\mathcal{C}}\left(\mathbf{y},\left\{3\right\}\right)\right)}{\sum_{v'\in\mathcal{C}}\exp\left(u_{1,v',\mathcal{C}}\left(\mathbf{y},\left\{3\right\}\right)\right)}.\quad \quad  \quad \quad \quad \quad  \quad \quad  \quad \quad  \quad  \hfill \square \]
  
\end{example}
Note that in the above example, consideration sets and choices of preference peers can directly affect mean utilities and the distribution of the latent taste shocks. Since the choices of consideration-only peers cannot affect the distribution of utilities conditional on consideration sets, in contrast to \citet{barseghyan2019heterogeneous} and \citet{lu2014moment}, our model imposes some restrictions on the dependence between random preferences and consideration.\footnote{\citet{aguiar2020identification} allow for a similar form of dependence of mean utilities on consideration.} 

We extend the model in Section~\ref{sec: extensionslong} to allow the dependence of $\operatorname{Q}_a$ and $\operatorname{R}_a$ on the current or past choices made by Agent $a$ (e.g., a Markov process with memory). This allows the possibility that an agent considers for sure (i.e., with probability 1) her current option in the next revision. We write the model in a stricter way here only to simplify the exposition.

By combining preferences and consideration sets, the conditional choice (ex-ante) probability (CCP) that Agent $a$ selects (at the time of choosing) alternative $v$ given $\mathbf{y}$ is
\begin{equation*}
\operatorname{P}_{a}\left( v \mid \mathbf{y}\right) =\sum\nolimits_{ 
\mathcal{C}\subseteq \mathcal{Y}}\operatorname{R}_{a}\left( v \mid \mathbf{y},\mathcal{NR}_{a},\mathcal{C}\right)\prod\nolimits_{v'\in \mathcal{C}}\operatorname{Q}_{a}\left( v' \mid\mathbf{y},\mathcal{NC}_{a}\right) \prod\nolimits_{v'\in \mathcal{Y}\setminus\mathcal{C}}\left( 1-\operatorname{Q}_{a}\left(
v' \mid\mathbf{y},\mathcal{NC}_{a}\right) \right).
\end{equation*}
\noindent We aim to identify $\mathcal{NR}_a$, $\mathcal{NC}_a$, $\operatorname{R}_a$, and $\operatorname{Q}_a$ from a sequence of choices over time.

\begin{rem}
    Our identification arguments \emph{only} use variation in the choices made by connected agents. That is, we do not use exogenous variation in observable characteristics or menus. Thus, to simplify the exposition, we do not include observable covariates in the model. We can interpret our setting as if we were conditioning on them. We show in our application that covariates can be easily incorporated into the model for estimation purposes and could also offer extra sources of identification. In particular, alternative-specific covariates can serve as further exclusion restrictions for the consideration or preferences ---e.g., product-specific advertisement might only affect attention to a specific product \citep{goeree2008limited}. 
\end{rem}
\begin{rem}
The dynamic interaction process we model assumes that each agent best responds to the observed choices made by others and does not anticipate their actions in the future or how her choice could affect them. Allowing for these possibilities would require a different interpretation of our model. For instance, an agent could select an option so that others incorporate the alternative in their consideration sets. We believe this setup could lead to a new model of endogenous social norms or rules within a group of people.
\end{rem}
\begin{rem}\label{rem: D}
    Our framework does not allow the dependence of current consideration on the previously considered (but not picked) alternatives. Allowing this invalidates our exclusion restrictions. Specifically, the introduction of the dependence of the consideration on past consideration leads to a hidden state Markov process, where the hidden state is the past consideration set. Since the past consideration sets and the past choices are correlated, and the past choices depend on preference peers, we obtain that the past consideration sets are correlated with choices made by preference peers. As a result, the choices made by \emph{all} peers affect consideration. We leave this important and hard problem for future research.
\end{rem}

\subsection{Main Assumptions}
\noindent Our results for equilibrium existence and identification build on four main assumptions. We have already discussed Assumption~\ref{ass: MM}. We introduce next the other three main restrictions. Let $\operatorname{NC}_{a}^{v}\left( \mathbf{y}\right)$ and $\operatorname{NR}_{a}^{v}\left( \mathbf{y}\right)$ be the number of agents in the consideration and preference groups of Agent $a$ who select option $v$ in choice configuration $\mathbf{y}$, respectively. Formally, 
\begin{align*}
\operatorname{NC}_{a}^{v}\left(\mathbf{y}\right) =\abs{\left\{a'\in\mathcal{NC}_a\::\: y_{a'}=v\right\}}
\text{ and } \operatorname{NR}_{a}^{v}\left(\mathbf{y}\right) =\abs{\left\{a'\in\mathcal{NR}_a\::\: y_{a'}=v\right\}},
\end{align*}
where $\abs{A}$ is the cardinality of $A$. Let $\operatorname{NR}^{\mathcal{S}}_{a}\left(\mathbf{y}\right)=\left(\operatorname{NR}_{a}^{v}\left(\mathbf{y}\right) \right)_{v\in\mathcal{S}\setminus\{0\}}$ and $\operatorname{NC}^{\mathcal{S}}_{a}\left(\mathbf{y}\right)=\left(\operatorname{NC}_{a}^{v}\left(\mathbf{y}\right) \right)_{v\in\mathcal{S}\setminus\{0\}}$ for any $\mathcal{S}\subseteq\mathcal{Y}$. The second assumption is as follows.

\begin{assumption}[Consideration]\label{ass: A1}
    For each $a\in\mathcal{A}$, $\mathbf{y}\in \mathcal{Y}^{A}$, and $v\neq 0$, we have that
    \begin{enumerate}
    \item $\operatorname{Q}_{a}\left( v \mid \mathbf{y},\mathcal{NC}_{a}\right) >0$; 

    \item $\operatorname{Q}_{a}\left( v \mid \mathbf{y},\mathcal{NC}_{a}\right) \equiv \operatorname{Q}_{a}\left( v \mid \operatorname{NC}_{a}^{v}\left(\mathbf{y}\right)\right)$; and
    
    \item Under Assumption~\ref{ass: A1}(ii), $\operatorname{Q}_{a}\left( v \mid 1\right)/ \operatorname{Q}_{a}\left( v \mid 0\right)$ differs from 1 and $\operatorname{Q}_{a}\left( v \mid 2\right)/ \operatorname{Q}_{a}\left( v \mid 1\right)$.
    \end{enumerate}
\end{assumption}

\noindent Assumption~\ref{ass: A1}(i) states that the probability of considering each option is strictly positive regardless of how many peers have selected it. This assumption captures the idea that agents can eventually pay attention to an alternative for various reasons that are outside the control of our model (e.g., watching an ad on television or receiving a coupon). We allow alternatives to be considered with probability 1,\footnote{Note that, due to its multiplicative structure, the ex-ante probability of facing a given consideration set is between $0$ and $1$ and adds up to $1$ when we sum across all consideration sets.} capturing a form of persistence in consideration sets. We allow even richer forms of evolution of consideration sets by introducing history dependence to the model in  Section~\ref{sec: extensionslong}. Assumption~\ref{ass: A1}(ii) says that the probability of considering a specific option depends on the number but not the identity of the consideration peers that currently selected it. Assumption~\ref{ass: A1}(iii) is a variability requirement stating that the choices made by consideration peers have an effect on consideration probabilities. It rules out constant or exponential probability functions of the number of peers (i.e., $\ln \operatorname{Q}_a(v\mid\cdot)$ is nonlinear) around the origin. This restriction could be imposed at any other point in the support. This assumption allows for different levels of satiation. For example, consideration may change only when the number of peers picking the option achieves a given threshold (e.g., 10 agents, 20 agents, etc.). Assumption~\ref{ass: A1}(iii) is not fully needed for all our results. For instance, nonexponential probabilities are needed to identify $\mathcal{NC}_a$, but $\mathcal{N}_a$ and $\mathcal{NR}_a$ only require consideration probabilities to vary with the choices made by the peers of Agent $a$.  
\setcounter{aux}{\value{example}}
\setcounter{example}{\value{eg1}}
\begin{example}[continued]
    Suppose that the mean attention of Agent $1$ towards alternative $v$, $c_{1,v}$, is such that 
    $c_{1,v}(\mathbf{y},\left\{{2,3}\right\})= \operatorname{NC}_{1}^{v}\left(\mathbf{y}\right)$ with $\operatorname{NC}_{1}^{v}\left(\mathbf{y}\right) = \Char{y_2 = v} +\Char{y_3 = v}$. Thus, $\operatorname{Q}_1(v|\operatorname{NC}_{1}^{v}\left(\mathbf{y}\right))=\operatorname{F}_{\varepsilon_{1,v}}\left(\Char{y_2 = v} +\Char{y_3 = v}\right)$. Assumption~\ref{ass: A1} holds if, for instance, the probability of considering an option increases with the number of consideration peers that select that alternative at a decreasing rate, e.g., $\operatorname{F}_{\varepsilon_{1,v}}(0) = 1/8$, $\operatorname{F}_{\varepsilon_{1,v}}(1) = 1/2$, and $\operatorname{F}_{\varepsilon_{1,v}}(2) = 3/4$. It also holds if the rate of change is initially increasing, as in the well-known S-shaped curve for network effects, where a phase of increasing returns is then followed by diminishing returns. It only rules out a constant rate of change at some point.
\hfill $\square$
\end{example}
\setcounter{example}{\value{aux}}

Let $\mathbf{0}_{v}^{1}$ denote a vector obtained by replacing the $v$-th component of the zero vector $\mathbf{0}$ by $1$. The third assumption restricts the preference part of the decision process.
\begin{assumption}[Preferences]\label{ass: A2}
    For each $a\in\mathcal{A}$, $\mathbf{y}\in \mathcal{Y}^{A}$, $\mathcal{C}\subseteq\mathcal{Y}$, and $v\in \mathcal{C}\setminus\{0\}$,
\begin{enumerate}
    \item $\operatorname{R}_{a}\left( v \mid \mathbf{y},\mathcal{NR}_{a},\mathcal{C}^*\right) >0$ for some $\mathcal{C}^*$ such that $\operatorname{C}_a(\mathcal{C}^*|\mathbf{y},\mathcal{NC}_a,\mathcal{Y})>0$;
    \item $\operatorname{R}_{a}\left( v \mid \mathbf{y},\mathcal{NR}_{a},\mathcal{C}\right) \equiv \operatorname{R}_{a}\left( v \mid \operatorname{NR}^{\mathcal{C}}_{a}\left(\mathbf{y}\right),\mathcal{C}\right)$; and
    
    \item $\operatorname{R}_{a}\left( v \mid \mathbf{0}_{v}^{1},\mathcal{C}\right)-\operatorname{R}_{a}\left( v \mid \mathbf{0},\mathcal{C}\right)\neq 0$, and its sign does not depend on  $\mathcal{C}$.
\end{enumerate}
\end{assumption}

\noindent Assumption~\ref{ass: A2}(i) requires each alternative to be picked with a positive probability at least in one consideration set that can be realized. Together with Assumption~\ref{ass: A1}(i), it implies that every alternative can be picked with a positive probability. This assumption allows for both random and deterministic choice rules. Assumption~\ref{ass: A2}(ii) states that the choice rule depends on the number (but not the identity) of preference peers that selected each of the alternatives in the consideration set. Assumption~\ref{ass: A2}(iii) assumes that the change in the probability of selecting a given alternative due to an additional preference peer selecting this alternative is either positive or negative for all consideration sets that contain the alternative. The effect is required to be strict only around the origin, that is, when all other peers select the default. It rules out that the different changes induced by a preference peer changing her option cancel out. Though this condition can be relaxed, as written, it is easy to interpret and is weaker than assuming either positive or negative peer effects in preferences ---as in many existing models. As in Assumption~\ref{ass: A1}(iii), the direction of the peer effect in preferences does not need to be known and can be different for different agents and alternatives. 
\setcounter{aux}{\value{example}}
\setcounter{example}{\value{eg2}}
\begin{example}[continued]
    Suppose that the mean utility of Agent $1$ from alternative $v$ given consideration set $\mathcal{C}$  is 
    $u_{1,v,\mathcal{C}}\left(\mathbf{y},\left\{3\right\}\right)=
\bar{u}_{1,v,\mathcal{C}}\left(\operatorname{NR}^v_1\left(\mathbf{y}\right)\right)
    $ with $\operatorname{NR}_{1}^{v}\left(\mathbf{y}\right) = \Char{y_3 = v}.$
    Thus, given that only Agent $3$ affects the preferences of Agent $1$, the choice rule is expressed as
\[
    \operatorname{R}_{1}\left(v  \mid \operatorname{NR}^\mathcal{C}_1\left(\mathbf{y}\right),\mathcal{C}\right)=\dfrac{\exp\left(\bar{u}_{1,v,\mathcal{C}}\left(\Char{y_3 = v}\right)\right)}{\sum_{v'\in\mathcal{C}}\exp\left(\bar{u}_{1,v',\mathcal{C}}\left(\Char{y_3 = v'}\right)\right)}.
    \]
Assumption~\ref{ass: A2} holds if $\bar{u}_{1,v,\mathcal{C}}(1) > \bar{u}_{1,v,\mathcal{C}}(0)$ or $\bar{u}_{1,v,\mathcal{C}}(1) < \bar{u}_{1,v,\mathcal{C}}(0)$ for all $\mathcal{C}$. This requires peer effects in preferences to be either positive or negative for each $v$ and all $\mathcal{C}$. 
\hfill $\square$

\end{example}
\setcounter{example}{\value{aux}}

The fourth assumption imposes some restrictions on the network of each person. 
\begin{assumption}[Network]\label{ass: A3}
    For each $a$, 
    if $\abs{\mathcal{NCR}_a} \geq 1$ then $\abs{\mathcal{NC}_a\setminus\mathcal{NR}_a} + \abs{\mathcal{NR}_a\setminus\mathcal{NC}_a} \geq 1$.
\end{assumption}
\noindent Assumption~\ref{ass: A3} is an exclusion restriction. It states that if the agent has a peer that affects both consideration and preferences, then the agent also has at least another peer that affects either only consideration or only preferences. The choice made by such a peer provides an exclusion restriction. Without Assumption~\ref{ass: A3} the scenario where $\mathcal{NC}_a=\mathcal{NR}_a$ is observationally equivalent to the one in which $\mathcal{NC}_a=\emptyset$. Note that we do not need the two sets of agents to be nonempty. If only one of the peer-effect mechanisms operates in practice, our results will allow us to state whether the interdependencies in choices are due to peer effects in preferences or in consideration.

\setcounter{aux}{\value{example}}
\setcounter{example}{\value{eg2}}
\begin{example}[continued]
    Agent 1 is the only agent for whom $\abs{\mathcal{NCR}_1} = 1 \geq 1$. 
    But we also have that 
    $\abs{\mathcal{NC}_1\setminus\mathcal{NR}_1} + \abs{\mathcal{NR}_1\setminus\mathcal{NC}_1} = 1 \geq 1$. Thus, Assumption~\ref{ass: A3} is satisfied.
    \hfill $\square$
\end{example}
\setcounter{example}{\value{aux}}

\subsection{Equilibrium Behavior}\label{sec: equilibrium and mistakes}

\noindent In this subsection, we define a notion of equilibrium in the network system, i.e., the invariant distribution in the Markov process, and establish its existence and uniqueness.

The i.i.d. Poisson alarm clocks, which determine the revision process, guarantee that each time, at most, one agent revises her selection almost surely. Thus, the transition rates between choice configurations that differ in more than one agent changing the current selection are zero. This facilitates identification as there are fewer terms to recover \citep{blevins2017identifying,blevins2018identification}. It also rules out some mechanisms for multiple equilibria in the data-generating process. Formally, the transition rate from choice configuration $\mathbf{y}$ to any different one $\mathbf{y}'$ is
\begin{equation*}
\operatorname{w}\left( \mathbf{y}' \mid \mathbf{y}\right) =\left\{ 
\begin{array}{lcc}
0 & \text{if} & \sum\nolimits_{a\in \mathcal{A}}\mathds{1}\left( y_{a}'\neq
y_{a}\right) >1 \\ 
\sum\nolimits_{a\in \mathcal{A}}\lambda_{a}\operatorname{P}_{a}\left(y'_a \mid \mathbf{y}\right) \mathds{1}\left( y_{a}'\neq y_{a}\right) & 
\text{if} & \sum\nolimits_{a\in \mathcal{A}}\mathds{1}\left( y_{a}'\neq
y_{a}\right) =1
\end{array}
\right. .  \label{T}
\end{equation*}
In the statistical literature on continuous-time Markov processes, these transition rates are the off-diagonal terms of the \textit{transition rate matrix} (also known as the \textit{infinitesimal generator matrix}). The diagonal terms are simply given by $\operatorname{w}\left( \mathbf{y}\mid \mathbf{y}\right) =-\sum\nolimits_{\mathbf{y}%
'\in \mathcal{Y}^{A}\setminus \left\{ \mathbf{y}\right\}
}\operatorname{w}\left( \mathbf{y}'\mid \mathbf{y}\right)$.

We indicate by $\mathcal{W}$ the transition rate matrix. In our model, the number of possible choice configurations is $\left(Y+1\right)^{A}$. Thus, $\mathcal{W}$ is a $\left(Y+1\right)^{A}\times \left(Y+1\right)^{A}$ matrix. To avoid ambiguity in the exposition, we let the choice configurations be ordered according to the lexicographic order. Formally, let $\iota\left( \mathbf{y}\right) \in\left\{1,2,\dots,\left(Y+1\right)^{A}\right\} $ be the position of $\mathbf{y}$ according to the lexicographic order. Then, $\mathcal{W}_{\iota \left(\mathbf{y}\right) \iota \left(\mathbf{y}'\right) }=\operatorname{w}\left(\mathbf{y}'\mid \mathbf{y}\right)$.

The system in equilibrium is characterized by an invariant distribution $\mu: \mathcal{Y}^{A} \rightarrow\left(0,1\right)$, with $\sum\nolimits_{\mathbf{y}\in \mathcal{Y}^{A}}\mu \left( \mathbf{y}\right) =1$, of the dynamic process with transition rate matrix $\mathcal{W}$. This equilibrium behavior relates to the transition rate matrix in a linear fashion $\mu \mathcal{W}=\mathbf{0}$.

The next proposition shows equilibrium existence. It also states that the same conditions guarantee the equilibrium is unique and has full support on the choice configurations. The full support feature is important for us as we rely on choice variation to recover the model. 

\begin{proposition}\label{prop: unique equilibrium}
If Assumptions~\ref{ass: MM},~\ref{ass: A1}(i), and~\ref{ass: A2}(i) hold, then there exists a unique equilibrium $\mu$ with full support.
\end{proposition}

\section{Empirical Content of the Model}\label{sec: model identification}
\noindent This section offers restrictions under which the researcher can uniquely recover the set of connections, $\mathcal{NC}_{a}$ and $\mathcal{NR}_{a}$, the consideration mechanism, $\operatorname{Q}_a$, the choice rule, $\operatorname{R}_a$, and the Poisson alarm clock, $\lambda_a$, for every Agent $a$. We separate the identification analysis in two parts. We first show how to recover the set of connections, choice rules, and consideration probabilities from the CCPs, $\operatorname{P}=\left( \operatorname{P}_{a}\right)_{a\in \mathcal{A}}$. We then show how to recover $\operatorname{P}$.

\begin{rem}
In Online Appendix B, we run simulations from our motivating example to show that we can estimate CCPs and parts of the model nonparametrically using our identification ideas step-by-step. We also show that estimation becomes harder with more agents and/or alternatives. In the application, we use a parametric estimator to circumvent these difficulties.     
\end{rem}  

\subsection{Identification of the Model From \texorpdfstring{$\operatorname{P}$}{the Conditional Choice Probabilities}}\label{PIP}

\noindent The identification strategy we offer builds on a sequence of steps. We initially recover the network structure in three stages: First, we recover the reference group of every agent. Second, we identify whether a given peer affects consideration only or preferences. Lastly, we show how to distinguish between a peer that affects preferences only (preference-only peer) and a peer that affects consideration and preferences (consideration-preference peer). We finally use this information to recover consideration probabilities and choice rules. 

Along the analysis we assume a mild ``regularity condition.'' We briefly discuss the role of this condition below and formalize it in Online Appendix A.
\bigskip

\noindent \textbf{Network} The key idea in recovering the network is that changes in the choices made by the peers of a given agent induce variation in her CCPs \emph{and} that different types of peers induce a different type of variation. To see this, note that $\operatorname{P}_{a}$ can be rewritten as%
\begin{align*}
\operatorname{P}_{a}\left( v \mid \mathbf{y}\right) &=\operatorname{Q}_{a}\left(v \mid \operatorname{NC}_{a}^{v}\left(\mathbf{y}\right)\right)\times \operatorname{D}_{a}\left(v \mid \operatorname{NR}_{a}^{\mathcal{Y}}\left(\mathbf{y}\right),\operatorname{NC}_{a}^{\mathcal{Y}\setminus\{v\}}\left(\mathbf{y}\right)   \right),
\end{align*}
\noindent where the second term in the right-hand side is given by
{\footnotesize
\begin{align*}
\operatorname{D}_{a}\left(v \mid \operatorname{NR}_{a}^{\mathcal{Y}}\left(\mathbf{y}\right),\operatorname{NC}_{a}^{\mathcal{Y}\setminus\{v\}}\left(\mathbf{y}\right)   \right)&=\sum\nolimits_{ 
\mathcal{C}\subseteq\mathcal{Y}\setminus\{v\}}\operatorname{R}_{a}\left( v \mid \operatorname{NR}^{\mathcal{C}\cup\{v\}}_{a}\left(\mathbf{y}\right),\mathcal{C}\cup\{v\}\right)\operatorname{C}_a\left(\mathcal{C}\mid \operatorname{NC}_a^{\mathcal{Y}\setminus\{v\}}(\mathbf{y}),\mathcal{Y}\setminus\{v\}\right).
\end{align*}}

\noindent In words, the observed probability that $v$ is picked equals the product of the probabilities that it is considered and that it is picked when considered. Note that consideration-only peers who select option $v$ enter only the first term. In addition, consideration and/or preference peers who select alternatives other than $v$ affect only the second term. These two observations allow us to use certain changes in $\ln{\operatorname{P}_a}$ to identify the network.

Define $\Delta_{a'}^{v}$ as a linear operator that indicates the increment of a function when the choice of Agent $a'$ in $\mathbf{y}$ changes to $v$. Formally, given $f:\mathcal{Y}^{A}\to\Real$, let $\Delta_{a'}^{v}f(\mathbf{y})=f(\mathbf{y}^{v}_{a'})-f(\mathbf{y})$,
where $\mathbf{y}^{v}_{a'}$ denotes the vector in which the $a'$-th component of $\mathbf{y}$ is replaced by $v$. 

We first identify the reference group of Agent $a$ by using changes in her CCPs. Intuitively, Agent $a'$ is in the reference group of Agent $a$ if changing her choice in the choice configuration affects the decision of Agent $a$. Specifically, the implied difference in $\ln{\operatorname{P}_a}$ is given by
\begin{align}\label{eq: delta log}
&\Delta_{a'}^{v}\ln{\operatorname{P}_{a}\left( v \mid  \mathbf{0}\right)} = \Delta_{a'}^{v}\ln{\operatorname{Q}_{a}\left(v \mid \operatorname{NC}_{a}^{v}\left(\mathbf{0}\right)\right)} + \Delta_{a'}^{v}\ln{\operatorname{D}_{a}\left(v \mid \operatorname{NR}_{a}^{\mathcal{Y}}\left(\mathbf{0}\right),\operatorname{NC}_{a}^{\mathcal{Y}\setminus\{v\}}\left(\mathbf{0}\right)   \right)},
\end{align}
where the zero vector $\mathbf{0}$ denotes the configuration where everyone picks the default. Each term in Equation~\eqref{eq: delta log} relates to one mechanism of peer effects: the first term reflects (if present) the peer effect in consideration. The second term captures the peer effect in preferences. 

When peer effects in consideration and preferences are of the same sign, then, under Assumptions~\ref{ass: MM} -~\ref{ass: A2}, $\Delta_{a'}^{v}\ln{\operatorname{P}_{a}\left( v \mid \mathbf{0}\right)} \neq 0$ if and only if Agent $a'$ is in the reference group of Agent $a$ (i.e., $a' \in \mathcal{N}_a$). When the interaction effects are of different signs, the ``if'' part requires a ``regularity condition'' that we discuss in detail in Online Appendix A. This extra condition rules out the possibility that peer effects in consideration and preferences be of opposite signs \emph{and} of equal magnitude. We next state that under all our restrictions the reference groups (even if they are empty) can be recovered from the CCPs.

\begin{proposition}\label{prop: identification of peers}
If Assumptions~\ref{ass: MM} - \ref{ass: A2} hold, then $\mathcal{N}_a$ is identified for each $a\in\mathcal{A}$.
\end{proposition}

\setcounter{example}{\value{eg2}}
\begin{example}[continued] The probability that Agent 1 selects option 1 satisfies
\begin{align*}
\ln{\operatorname{P}_{1}\left( 1 \mid \mathbf{y}\right)} =
&\ln{\operatorname{Q}_{1}\left( 1 \mid \operatorname{NC}^{1}_{1} = \Char{y_2 = 1}+\Char{y_3 = 1}\right)}   \\ 
&+\ln{\operatorname{D}_{1}\left(1 \mid \operatorname{NR}_{1}^{\left\{1,2\right\}} = \left(\Char{y_3 = 1}, \Char{y_3 = 2}\right), \operatorname{NC}_{1}^{2} = \Char{y_2 = 2}+\Char{y_3 = 2}\right)}. 
\end{align*}
Take $\mathbf{y} = \mathbf{0}$. The changes in $\ln{\operatorname{P}_{1}}$ when we change the other agents choices from 0 to 1 are
\begin{align*}
&\Delta_{2}^{1}\ln{\operatorname{P}_{1}\left( 1 \mid \mathbf{0}\right)} = \ln{\operatorname{Q}_{1}\left( 1 \mid \operatorname{NC}_{1}^{1} = 1\right)} - \ln{\operatorname{Q}_{1}\left( 1 \mid \operatorname{NC}_{1}^{1} = 0\right)} \neq 0,
\\
&\Delta_{3}^{1}\ln{\operatorname{P}_{1}\left( 1 \mid \mathbf{0}\right)} = \ln{\operatorname{Q}_{1}\left( 1 \mid \operatorname{NC}_{1}^{1} = 1\right)} - \ln{\operatorname{Q}_{1}\left( 1 \mid \operatorname{NC}_{1}^{1} = 0\right)}  \\
&\qquad \qquad \qquad \quad +\ln{\operatorname{D}_{1}\left(1 \mid \operatorname{NR}_{1}^{\left\{1,2\right\}} = \left(1, 0\right), \operatorname{NC}_{1}^{2} = 0\right)} - \ln{\operatorname{D}_{1}\left(1 \mid \operatorname{NR}_{1}^{\left\{1,2\right\}} = \left(0, 0\right), \operatorname{NC}_{1}^{2} = 0\right)} \neq 0,
\\
&\Delta_{4}^{1}\ln{\operatorname{P}_{1}\left( 1 \mid \mathbf{0}\right)} = 0.
\end{align*}
\noindent Proposition~\ref{prop: identification of peers} identifies the peers of Agent 1 as those for which these changes differ from 0. Following this idea, we correctly recover $\mathcal{N}_1 = \left\{2,3\right\}.$
\hfill $\square$
\end{example}
\setcounter{example}{\value{aux}}

Next, we identify whether Agent $a'$ in $\mathcal{N}_{a}$ affects Agent $a$'s preferences or consideration only. Note that differences in $\ln{\operatorname{P}_a}$ allow us to recover the reference groups, but these differences are silent about the mechanism by which the interactions happen. To see why, note that (for instance) a nonzero $\Delta_{a'}^{v}\ln{\operatorname{P}_{a}\left( v \mid \mathbf{0}\right)}$ could be generated from the first summand in Equation~\eqref{eq: delta log} via $\operatorname{Q}_{a}$ and/or from the second summand via $\operatorname{R}_{a}$. But these two terms differ in that the second summand varies with the number of peers that select alternatives that are \emph{different} from $v$, while the first term does not. Thus, the two mechanisms can be set apart by a second shift in $\ln{\operatorname{P}_{a}\left( v \mid \mathbf{0}\right)}$. Let $a',a'' \in \mathcal{N}_{a}$ and $w \in \mathcal{Y}\setminus\{0\}$ with $w \neq v$. Since $\Delta_{a}^v$ is a linear operator, we can define a double difference as follows
\begin{align*}
\Delta_{a''}^{w}\Delta_{a'}^{v}\ln{\operatorname{P}_{a}\left( v \mid \mathbf{0}\right)} &\equiv\Delta_{a''}^{w}\left[\ln{\operatorname{P}_{a}\left( v \mid \mathbf{0}^{v}_{a'}\right)} -\ln{\operatorname{P}_{a}\left( v \mid \mathbf{0}\right)}\right]=\Delta_{a''}^{w}\ln{\operatorname{P}_{a}\left( v \mid \mathbf{0}^{v}_{a'}\right)} -\Delta_{a''}^{w}\ln{\operatorname{P}_{a}\left( v \mid \mathbf{0}\right)}\\
&=\left[\ln{\operatorname{P}_{a}\left( v \mid \left(\mathbf{0}^{v}_{a'}\right)^{w}_{a''}\right)} -\ln{\operatorname{P}_{a}\left( v \mid \mathbf{0}^{v}_{a'}\right)}\right]-\left[\ln{\operatorname{P}_{a}\left( v \mid \mathbf{0}^{w}_{a''}\right)} -\ln{\operatorname{P}_{a}\left( v \mid \mathbf{0}\right)}\right].
\end{align*}
Specifically, we have that
\begin{align*}
\Delta_{a''}^{w}\Delta_{a'}^{v}\ln{\operatorname{P}_{a}\left( v \mid \mathbf{0}\right)} = \Delta_{a''}^{w}\Delta_{a'}^{v}\ln{\operatorname{Q}_{a}\left(v \mid \operatorname{NC}_{a}^{v}\left(\mathbf{0}\right)\right)} + \Delta_{a''}^{w}\Delta_{a'}^{v} \ln{\operatorname{D}_{a}\left(v \mid \operatorname{NR}_{a}^{\mathcal{Y}}\left(\mathbf{0}\right),\operatorname{NC}_{a}^{\mathcal{Y}\setminus\{v\}}\left(\mathbf{0}\right)   \right)}.
\end{align*}
Note that $\Delta_{a''}^{w}\Delta_{a'}^{v}\ln{\operatorname{Q}_{a}\left(v \mid \operatorname{NC}_{a}^{v}\left(\mathbf{0}\right)\right)} = 0$ since $\operatorname{Q}_{a}\left(v \mid \cdot \right)$ does not depend on the number of peers who picked $w$. Also, if $a'$ is a consideration-only peer ($a'\in\mathcal{NC}_a\setminus\mathcal{NR}_a$), then 
\[
\Delta_{a'}^{v}\ln{\operatorname{D}_{a}\left(v \mid \operatorname{NR}_{a}^{\mathcal{Y}}\left(\mathbf{0}\right),\operatorname{NC}_{a}^{\mathcal{Y}\setminus\{v\}}\left(\mathbf{0}\right)   \right)}=0.
\]
As a result, $\Delta_{a''}^{w}\Delta_{a'}^{v}\ln{\operatorname{P}_{a}\left( v \mid \mathbf{0}\right)} = 0$ if Agent $a'$ is a consideration-only peer.
A key observation is that if Agent $a'$ affects Agent $a$'s preferences, then the second summand in Equation~\eqref{eq: delta log} will not disappear after switching Agent $a''$ from $0$ to $w$. In summary, under Assumptions~\ref{ass: A1}-\ref{ass: A2} and the regularity condition, $a'\in \mathcal{NR}_{a}$ if and only if
\begin{align*}
\Delta_{a''}^{w}\Delta_{a'}^{v}\ln{\operatorname{P}_{a}\left( v \mid \mathbf{0}\right)} \neq 0 \text{ for some } a''\in \mathcal{N}_a.
\end{align*}
Thus, by checking the double difference for each agent in the reference group of Agent $a$, we can divide her reference group into consideration-only peers and preference peers (who may or may not affect consideration). This identification strategy requires $Y \geq 2$ and $\abs{\mathcal{N}_a} \geq 2$.

Note that we can also separate the preference peers in two sets by the magnitude of the changes in CCPs. Specifically, for $a'\in\mathcal{NCR}_a$ and $a''\in\mathcal{NR}_a\setminus\mathcal{NC}_a$, we have that
\[
\Delta_{a'}^{v}\ln{\operatorname{P}_{a}\left( v \mid  \mathbf{0}\right)}\neq \Delta_{a''}^{v}\ln{\operatorname{P}_{a}\left( v \mid  \mathbf{0}\right)}.
\]
This allows us to separate the preference peers into two groups that we define as $\mathcal{M}'$ and $\mathcal{M}''$. Although we know that one of these sets is $\mathcal{NCR}_a$ and the other is $\mathcal{NR}_a\setminus\mathcal{NC}_a$, without further restrictions, we cannot tell which is which. We address this issue in the end.

\begin{proposition}\label{prop: identification of pref group}
Suppose Assumptions~\ref{ass: MM} - \ref{ass: A2} hold. For any $a\in\mathcal{A}$, if $Y\geq 2$ and $\abs{\mathcal{N}_a} \geq 2$, then $\mathcal{NC}_a \setminus\mathcal{NR}_a$ and $\mathcal{NR}_a = \mathcal{M}' \cup \mathcal{M}''$ are identified. Also, $\mathcal{NCR}_a \in \{\mathcal{M}',\mathcal{M}''\}$.   
\end{proposition}

\setcounter{aux}{\value{example}}
\setcounter{example}{\value{eg2}}
\begin{example}[continued] Recall that we identified that $\mathcal{N}_1 = \left\{2,3\right\}$. We now want to know what type of peers Agents 2 and 3 are. Next we display two double differences
{\small
\begin{align*}
\Delta_{3}^{2}\Delta_{2}^{1}\ln{\operatorname{P}_{1}\left( 1 \mid \mathbf{0}\right)} = & 0, \\
\Delta_{2}^{2}\Delta_{3}^{1}\ln{\operatorname{P}_{1}\left( 1 \mid \mathbf{0}\right)} = &\left[\ln{\operatorname{D}_{1}\left(1 \mid \operatorname{NR}_{1}^{\left\{1,2\right\}} = \left(1, 0\right), \operatorname{NC}_{1}^{2} = 1\right)} - \ln{\operatorname{D}_{1}\left(1 \mid \operatorname{NR}_{1}^{\left\{1,2\right\}} = \left(0, 0\right), \operatorname{NC}_{1}^{2} = 1\right)}\right] \\
-&\left[\ln{\operatorname{D}_{1}\left(1 \mid \operatorname{NR}_{1}^{\left\{1,2\right\}} = \left(1, 0\right), \operatorname{NC}_{1}^{2} = 0\right)} - \ln{\operatorname{D}_{1}\left(1 \mid \operatorname{NR}_{1}^{\left\{1,2\right\}} = \left(0, 0\right), \operatorname{NC}_{1}^{2} = 0\right)}\right] \neq 0.
\end{align*}
}%
\noindent In the first line, we first change the choice of Agent 2 from 0 to 1. Since Agent 2 is a consideration-only peer, this difference does not depend on the number of other peers selecting options different from 1. Thus, when we further change the choice of Agent 3 from 0 to 2, the result is 0, and we conclude that Agent 2 is a consideration-only peer. 

In the second line, we first change the choice of Agent 3 from 0 to 1. Since Agent 3 is a consideration-preference peer of Agent 1, this difference depends on the choices made by other peers selecting other alternatives. Thus, when we further change the choice of Agent 2 from 0 to 2, the result differs from 0, and we conclude that Agent 3 is a preference peer. 
\hfill $\square$
\end{example}
\setcounter{example}{\value{aux}}

Finally, we identify the set of consideration-preference peers (i.e., $\mathcal{NCR}_a$) from the group of peers that affect preferences. We discuss identification with and without consideration-only peers separately. By Assumption~\ref{ass: A3}, if $\mathcal{NCR}_a$ is nonempty, then there exists a peer that is either a consideration-only or preference-only peer. Assume that we have already identified two peers such that Agent $a'$ is a consideration-only peer (i.e., $a'\in\mathcal{NC}_a\setminus\mathcal{NR}_a$) and Agent $a''$ affects preferences (i.e., $a''\in\mathcal{NR}_a$). Note that
\begin{align*}
\Delta_{a''}^{v}\Delta_{a'}^{v}\ln{\operatorname{P}_{a}\left( v \mid \mathbf{0}\right)} = \Delta_{a''}^{v}\Delta_{a'}^{v}\ln{\operatorname{Q}_{a}\left(v \mid \operatorname{NC}_{a}^{v}\left(\mathbf{0}\right)\right)}.
\end{align*}

\noindent This is so because, since Agent $a'$ only affects consideration, the second term in Equation~(\ref{eq: delta log}) is zero. Thus, if Assumption~\ref{ass: A1}(iii) holds, $a''\in\mathcal{NCR}_a$ if and only if $\Delta_{a''}^{v}\Delta_{a'}^{v}\ln{\operatorname{P}_{a}\left( v \mid \mathbf{0}\right)} \neq 0$. 

Suppose next that there is no consideration-only peer. We can implement a similar idea by \emph{replicating} the consideration-only peer behavior with a consideration-preference peer and a preference-only one. Note that these two peers can be identified by Proposition~\ref{prop: identification of pref group}. Pick some Agent  $a'\in\mathcal{M}'$ and Agent $a''\in\mathcal{M}''$. We have that
\begin{align*}
\ln{\operatorname{P}_{a}\left( v \mid \mathbf{0}_{a'}^{v}\right)} - \ln{\operatorname{P}_{a}\left( v \mid \mathbf{0}_{a''}^{v}\right)}=(-1)^{\Char{a'\not\in\mathcal{NCR}_a}} (\ln{\operatorname{Q}_{a}\left(v\mid  1\right)} - \ln{\operatorname{Q}_{a}\left(v\mid  0\right)}).
\end{align*}
That is, when we switch a consideration-preference peer from the default to alternative $v$,   $\ln{\operatorname{P}_a}\left( v \mid \mathbf{0}\right)$ changes via two effects, namely, preference and consideration. Importantly, the effect via preferences coincides with the one we get by switching a preference-only peer from the default to alternative $v$. Thus, by subtracting the two effects, we can recover the effect (up to sign) of switching a consideration-only peer from the default to option $v$. 

Finally, take another Agent $a'''$ from either $\mathcal{M}'$ or $\mathcal{M}''$ and implement a double difference to the previous expression by changing the alternative of Agent $a'''$ from the default to $v$. As before, we identify whether Agent $a'''$ is a consideration-preference or preference-only peer by checking whether this double change is different from zero:\footnote{This procedure requires at least three peers in $\mathcal{N}_a$.} 
\begin{align*}
\Delta_{a'''}^{v}\left[\ln{\operatorname{P}_{a}\left( v \mid \mathbf{0}_{a'}^{v}\right)} - \ln{\operatorname{P}_{a}\left( v \mid \mathbf{0}_{a''}^{v}\right)}\right]\neq 0 \iff a'''\in\mathcal{NCR}_a.
\end{align*}
This information allows us to know whether $\mathcal{NCR}_a = \mathcal{M}'$ or $\mathcal{NCR}_a = \mathcal{M}''$.

The last proposition offers final conditions for all the parts of the network to be identified. 

\begin{proposition}\label{prop: identification of groups}
Suppose Assumptions~\ref{ass: MM} - \ref{ass: A3} hold. Suppose also that $\mathcal{NC}_a\setminus\mathcal{NR}_a$ is identified (or known) and $\abs{\mathcal{N}_a} \geq 3 - \abs{\mathcal{NC}_a \setminus\mathcal{NR}_a}$. Then,
$\mathcal{NC}_{a}$ and $\mathcal{NR}_{a}$ are identified.
\end{proposition}

\setcounter{aux}{\value{example}}
\setcounter{example}{\value{eg2}}
\begin{example}[continued] We learned earlier that Agent 2 is a consideration-only peer of Agent 1. We also learned that Agent 3 is either a preference-only peer or affects both preferences and consideration of Agent 1.  We next establish whether $\mathcal{NC}_1 = \left\{2\right\}$ or $\mathcal{NC}_1 = \left\{2, 3\right\}$. Since Agent 2 is a consideration-only peer, we have that
\[
\Delta_{2}^{1}\ln{\operatorname{P}_{1}\left( 1 \mid \mathbf{0}\right)} = \ln{\operatorname{Q}_{1}\left( 1 \mid \operatorname{NC}_{1}^{1} = 1\right)} - \ln{\operatorname{Q}_{1}\left( 1 \mid \operatorname{NC}_{1}^{1} = 0\right)}.
\]
Changing the alternative of Agent 3 from 0 to 1, we obtain that by Assumption~\ref{ass: A1}(iii)
\begin{align*}
\Delta_{3}^{1}\Delta_{2}^{1}\ln{\operatorname{P}_{1}\left( 1 \mid \mathbf{0}\right)} = &\left[\ln{\operatorname{Q}_{1}\left( 1 \mid \operatorname{NC}_{1}^{1} = 2\right)} - \ln{\operatorname{Q}_{1}\left( 1 \mid \operatorname{NC}_{1}^{1} = 1\right)}\right]  \\
-&\left[\ln{\operatorname{Q}_{1}\left( 1 \mid \operatorname{NC}_{1}^{1} = 1\right)} - \ln{\operatorname{Q}_{1}\left( 1 \mid \operatorname{NC}_{1}^{1} = 0\right)}\right] \neq 0.
\end{align*}
Thus, we identify that Agent 3 is also a consideration peer. That is, $\mathcal{NC}_1 = \left\{2, 3\right\}$. 
\hfill $\square$
\end{example}
\setcounter{example}{\value{aux}}

To sum up, the reference group of Agent $a$ is identified by checking the variation in $\ln{\operatorname{P}_a}$ as we switch other agents from the default alternative to a specific $v$. If, in doing so, we identify that the agent has two or more peers, we can recover the consideration-only peers by using the additive separability of $\ln{\operatorname{P}_a\left( v \mid \mathbf{y}\right)}$ in $\operatorname{Q}_{a}\left(v \mid \operatorname{NC}_{a}^{v}\left(\mathbf{y}\right)\right)$. Finally, if we identify at least one consideration-only peer, we can use her as a baseline to identify all other types of peers. Otherwise, we \emph{create} such a peer by mixing the behavior of a consideration-preference peer with the one of a preference-only peer and use the behavior of the \emph{constructed} peer as a baseline to complete the network identification. In this case, we need at least three peers. 

We finally remark that while no restriction on the number of options is needed to recover the reference group of a given agent, we assume $Y \geq 2$ to divide this set in consideration and preference peers. We next use the initial example to see why this requirement is needed and to state what can be done when it fails, i.e. $Y = 1$.

\setcounter{aux}{\value{example}}
\setcounter{example}{\value{eg2}}
\begin{example}[continued] Let us keep the network but assume $\mathcal{Y} = \left\{0,1\right\}$. Then
\begin{align*}
\ln{\operatorname{P}_{1}\left( 1 \mid \mathbf{y}\right)} =
\ln{\operatorname{Q}_{1}\left( 1 \mid \Char{y_2 = 1} + \Char{y_3 = 1}\right)} + \ln{\operatorname{R}_{1}\left( 1 \mid \Char{y_3 = 1},\left\{0,1\right\}\right)}. 
\end{align*}
\noindent In this set-up,
\begin{align*}
&\Delta_{2}^{1}\ln{\operatorname{P}_{1}\left( 1 \mid \mathbf{0}\right)} = \ln{\operatorname{Q}_{1}\left( 1 \mid 1\right)} - \ln{\operatorname{Q}_{1}\left( 1 \mid 0\right)},
\\
&\Delta_{3}^{1}\ln{\operatorname{P}_{1}\left( 1 \mid \mathbf{0}\right)} = \ln{\operatorname{Q}_{1}\left( 1 \mid 1\right)} - \ln{\operatorname{Q}_{1}\left( 1 \mid 0\right)} + \ln{\operatorname{R}_{1}\left( 1 \mid 1,\left\{0,1\right\}\right)} - \ln{\operatorname{R}_{1}\left( 1 \mid 0,\left\{0,1\right\}\right)},
\\
&\Delta_{4}^{1}\ln{\operatorname{P}_{1}\left( 1 \mid \mathbf{0}\right)} = 0.
\end{align*}
\noindent Thus, we can learn that $\mathcal{N}_1 = \left\{2, 3\right\}$. But now, double differences will not allow us to state whether each of these agents is a consideration and/or a preference peer. 
\hfill $\square$
\end{example}
\setcounter{example}{\value{aux}}

The above example shows that our identification strategy fails to separate the type of peers when $Y = 1$. In this case, there are other sets of assumptions that we could invoke to restore identification. Among them, partial knowledge of the network structure and knowledge of the sign of peer effects would allow identification, as we illustrate next. 
\setcounter{aux}{\value{example}}
\setcounter{example}{\value{eg2}}
\begin{example}[continued] Recall that $Y = 1$. Suppose we have partial knowledge of the network structure. In particular, suppose we know that $\mathcal{NR}_1 = \left\{3\right\}$. Since we can still learn that $\mathcal{N}_1 = \left\{2, 3\right\}$, we conclude that $\mathcal{NC}_1 \setminus \mathcal{NR}_1 = \left\{2\right\}$. Thus, to recover the complete network, we only need to learn whether Agent 3 is a preference-only peer or a consideration-preference peer of Agent 1. As before, the fact that 
\begin{align*}
\Delta_{3}^{1}\Delta_{2}^{1}\ln{\operatorname{P}_{1}\left( 1 \mid \mathbf{0}\right)} = &\left[\ln{\operatorname{Q}_{1}\left( 1 \mid \operatorname{NC}_{1}^{1} = 2\right)} - \ln{\operatorname{Q}_{1}\left( 1 \mid \operatorname{NC}_{1}^{1} = 1\right)}\right] \\
 -&\left[\ln{\operatorname{Q}_{1}\left( 1 \mid \operatorname{NC}_{1}^{1} = 1\right)} - \ln{\operatorname{Q}_{1}\left( 1 \mid \operatorname{NC}_{1}^{1} = 0\right)}\right] \neq 0
\end{align*}
allows us to conclude that Agent 3 is a consideration-preference peer.

When $Y=1$ we can also recover the network structure under sign restrictions, which we have not imposed so far. If we assume that peer effects in consideration and preferences are of opposite signs, then we could dispense with the assumption that either $\mathcal{NC}_a$ or $\mathcal{NR}_a$ is known. This situation might apply to vaccines. Arguably, a person becomes aware of a vaccine if more of her friends are getting shots. Also, if more friends get vaccinated, then the chances of getting sick reduce, and this reduces the incentives to get the vaccine. Thus, the peer effects in consideration and preferences are positive and negative, respectively.\footnote{In a different model, a similar idea has been used by \citet{agranov2021importance} to explain some data on COVID-19 vaccine uptake.}  
\hfill $\square$
\end{example}
\setcounter{example}{\value{aux}}

\bigskip
\noindent \textbf{Consideration Mechanisms and Choice Rules} We first state that if we know the network structure, and each agent has at least one consideration-only peer ---or such a peer can be constructed from consideration-preference and preference-only peers, as we do above--- then we can recover ratios of consideration probabilities. To show this claim, let $a'\in\mathcal{NC}_{a} \setminus\mathcal{NR}_{a}$. Since Agent $a'$ only affects consideration, we can shift Agent $a'$'s choice from the default to $v$ and recover some information about the peer effect in consideration. Specifically, we have  
\[
\Delta_{a'}^v\ln\operatorname{P}_a(v\mid\mathbf{0})=\ln{\operatorname{Q}_a(v\mid1)}-\ln{\operatorname{Q}_a(v\mid0)}.
\]
Thus, we can identify 
$
\operatorname{Q}_{a}\left(
v \mid  1 \right)/\operatorname{Q}_{a}\left(
v \mid  0 \right).
$
If $\mathcal{NC}_{a} \setminus\mathcal{NR}_{a}$ is empty, but $\mathcal{NR}_{a} \setminus\mathcal{NC}_{a}$ is not, we can use preference-only peers in a similar way. In particular, suppose $a'\in\mathcal{NCR}_{a} $ and $a''\in\mathcal{NR}_{a} \setminus\mathcal{NC}_{a}$. Then, $\ln\operatorname{P}_a\left(v\mid\mathbf{0}_{a'}^{v}\right)-\ln\operatorname{P}_a\left(v\mid\mathbf{0}_{a''}^{v}\right)=\ln{\operatorname{Q}_a(v\mid1)}-\ln{\operatorname{Q}_a(v\mid0)}.$ By applying the same ideas to different initial configurations, we can identify ratios of consideration probabilities as we formally state next.

\begin{proposition}\label{prop: identification of ratios of Q}
Let $\mathcal{NC}_a$ and $\mathcal{NR}_a$ be known and Assumptions~\ref{ass: MM} - \ref{ass: A3} hold. Then
\[
\operatorname{Q}_{a}\left(
v \mid n + 1 \right)/\operatorname{Q}_{a}\left(
v \mid n \right)
\]
is identified from $\operatorname{P}_a$ for each $n$ from 0 to $\abs{\mathcal{NC}_{a}} - 1$. (We use the convention that if $\abs{\mathcal{NC}_{a}}=0$, then the set ``from 0 to -1'' is empty.)
\end{proposition}
\begin{rem}
    Proposition~\ref{prop: identification of ratios of Q} is valid for a substantially more general consideration set model. For example, the assumption that each alternative is added to the consideration sets independently from other alternatives (Assumption~\ref{ass: MM}) can be dropped. Indeed, by definition, $\operatorname{P}_a(v\mid\mathbf{y})=\operatorname{Q}_a(v\mid\operatorname{NC}^v_{a}(\mathbf{y}))\operatorname{Pr}_a(v\mid\mathbf{y},v\text{ is considered}),$
    where the second term is the conditional probability that $v$ is picked conditional on being considered. Thus, variation in the choices made by $a'\in\mathcal{NC}_a\setminus\mathcal{NR}_a$ would identify $\operatorname{Q}_a$ up to scale. Note, however, that in this case, knowing $\operatorname{Q}_a$ is not enough to identify $\operatorname{C}_a$ since $\operatorname{Q}_a$ does not convey information about the probability of several items being considered simultaneously.
\end{rem}

We next show that we can also recover some counterfactual objects of interest. Adding some restrictions, these counterfactuals will allow us to recover the choice rules. Define 
\begin{equation*}
\operatorname{P}^*_{a}\left( v \mid \mathbf{y}, \mathcal{Y}\setminus\mathcal{Z}\right) =\sum\nolimits_{ 
\mathcal{C}\subseteq \mathcal{Y}\setminus\mathcal{Z}}\operatorname{R}_{a}\left( v \mid \operatorname{NR}^{\mathcal{C}}_{a}\left(\mathbf{y}\right), \mathcal{C} \right)\operatorname{C}_a\left(\mathcal{C}\mid \operatorname{NC}_a^{\mathcal{Y}\setminus\mathcal{Z}}(\mathbf{y}),\mathcal{Y}\setminus\mathcal{Z}\right)
\end{equation*}
for each $\mathcal{Z} \subseteq \mathcal{Y}\setminus\{0\}$. That is, $\operatorname{P}^{*}_{a}\left( v \mid \mathbf{y}, \mathcal{Y}\setminus\mathcal{Z}\right)$ is the counterfactual probability of selecting alternative $v$ under choice configuration $\mathbf{y}$ when we restrict the set of available options or the menu from $\mathcal{Y}$ to $\mathcal{Y}\setminus\mathcal{Z}$. It tells us what happens to the CCPs when we remove set $\mathcal{Z}$ from the original menu. Note that, by definition, $\operatorname{P}^{*}_{a}\left( v \mid \mathbf{y}, \mathcal{Y}\right) = \operatorname{P}_{a}\left( v \mid \mathbf{y}\right)$. 

To fix the ideas behind the next result, consider the setting with $\mathcal{A} = \{a,a'\}$, $\mathcal{Y}=\{0,v,v'\}$, $\mathcal{NR}_a=\emptyset$, and $\mathcal{NC}_a=\{a'\}$. Take $\mathbf{y}$ such that $y_{a'}=0$ ($y_{a}$ can be arbitrary). Recall that $\mathbf{y}^{v'}_{a'}$ denotes a configuration where the $a'$-th component of $\mathbf{y}$ is replaced by $v'$. Since
\[
\operatorname{P}^{*}_a(v \mid \mathbf{y}, \mathcal{Y}\setminus\{v'\})=\operatorname{Q}_a\left(v\mid 0\right)\operatorname{R}_a\left(v\mid 0,\{0,v\}\right),
\]
we have that
\begin{align*}
\operatorname{P}_a\left(v\mid\mathbf{y}\right)=
&\operatorname{Q}_a\left(v'\mid0\right)\operatorname{Q}_a\left(v\mid0\right)\operatorname{R}_a\left(v\mid (0,0),\{0,v,v'\}\right)+
\Big[1-\operatorname{Q}_a\left(v'\mid0\right)\Big]\operatorname{P}^{*}_a(v \mid \mathbf{y}, \mathcal{Y}\setminus\{v'\}).
\end{align*}
This is the observed probability of Agent $a$ choosing option $v$ given that her peer $a'$ previously chose the default. Moreover, by switching $a'$'s choice from the default to $v'$, we have 
\begin{align*}
\operatorname{P}_a\left(v\mid\mathbf{y}^{v'}_{a'}\right)=
&\operatorname{Q}_a\left(v'\mid1\right)\operatorname{Q}_a\left(v\mid0\right)\operatorname{R}_a\left(v\mid (0,0),\{0,v,v'\}\right)+
\left(1-\operatorname{Q}_a\left(v'\mid1\right)\right)\operatorname{P}^{*}_a(v \mid \mathbf{y}, \mathcal{Y}\setminus\{v'\}).
\end{align*}
Note that we used the fact that since Agent $a'$ only affects Agent $a$'s consideration probability, but not the preference, the variation of Agent $a'$'s choice in the choice configuration provides variation in the consideration probability but not in the choice rule. That is, $\operatorname{R}_a\left(v\mid (0,0),\{0,v,v'\}\right)$ does not vary when $a'$ switches from the default to a different alternative. Moreover, we also used the fact that
$
\operatorname{P}^{*}_a(v \mid \mathbf{y}, \mathcal{Y}\setminus\{v'\})=\operatorname{P}^{*}_a(v \mid \mathbf{y}_{a'}^{v'}, \mathcal{Y}\setminus\{v'\}),
$
which follows from $v'$ being excluded from the menu and, thus, switching to it does not change the probability of picking $v$.

Solving this system of two equations with respect to $\operatorname{P}^{*}_a(v \mid \mathbf{y}, \mathcal{Y}\setminus\{v'\})$, we obtain that
\[
\operatorname{P}^{*}_a(v \mid \mathbf{y}, \mathcal{Y}\setminus\{v'\})=\dfrac{\operatorname{P}_a\left(v\mid\mathbf{y}^{v'}_{a'}\right)-t_{v'}\operatorname{P}_a\left(v\mid\mathbf{y}\right)}{1-t_{v'}},
\]
where $t_{v'}=\operatorname{Q}_a\left(v'\mid1\right)/\operatorname{Q}_a\left(v'\mid0\right)\neq 1$ can be identified using Proposition~\ref{prop: identification of ratios of Q}. It follows that we can recover the counterfactual CCP $\operatorname{P}^{*}_a(v \mid \mathbf{y}, \mathcal{Y}\setminus\{v'\})$ for any $\mathbf{y}$ for which the alternative corresponding to one of the consideration-only peers is equal to $0$ (i.e., $y_{a'}=0$). Essentially, we just used a consideration-only peer to exclude one alternative from the menu. Applying the same argument to these new counterfactual CCPs, we can exclude two alternatives as long as we have two consideration-only peers. Again, we can use any initial $\mathbf{y}$ as long as the components that correspond to any two consideration-only peers are set to $0$. That is, we can exclude any set of nondefault alternatives if its cardinality is smaller than $\abs{\mathcal{NC}_a\setminus\mathcal{NR}_a}$. 

The next result formalizes and extends this argument.

\begin{proposition}\label{prop: counterfactual CCP}
Suppose $\mathcal{NC}_a$ and $\mathcal{NR}_a$ are known, and Assumptions~\ref{ass: MM} - \ref{ass: A2} are satisfied. Then $\operatorname{P}^{*}_{a}\left( v \mid \mathbf{y}, \mathcal{Y}\setminus\mathcal{Z}\right)$ is identified from $\operatorname{P}_a$ for every $\mathcal{Z}\subseteq\mathcal{Y}\setminus\{0\}$ such that $\abs{\mathcal{Z}}\leq \abs{\mathcal{NC}_a \setminus \mathcal{NR}_a }$ and each $\mathbf{y}$ for which at least $\abs{\mathcal{Z}}$ of its components corresponding to any peers in $\mathcal{NC}_a \setminus \mathcal{NR}_a$ are $0$.
\end{proposition}

Proposition~\ref{prop: counterfactual CCP} addresses an important counterfactual prediction: What would happen if some alternatives were removed or become unavailable? Note the identification of these counterfactual CCPs does not require knowledge of either $\operatorname{Q}_a$ or $\operatorname{R}_a$. We only use ratios of $\operatorname{Q}_a$s. It follows from these ideas that (in our setting) variation in the choices of consideration-only peers is equivalent to menu variation in the stochastic choice literature \citep{aguiar2023random}. In particular, if one has enough consideration-only peers, we can identify the counterfactual CCPs for binary menus $\operatorname{P}^{*}_a(v \mid \mathbf{y}, \{0,v\})=\operatorname{Q}_a\left(v\mid\operatorname{NC}^v_{a}\left(\mathbf{y}\right)\right)\operatorname{R}_a\left(v\mid \operatorname{NR}_a^v\left(\mathbf{y}\right),\{0,v\}\right)$. Hence, if either $\operatorname{Q}_a\left(v\mid\operatorname{NC}^v_{a}\left(\mathbf{y}\right)\right)$ or $\operatorname{R}_a\left(v\mid \operatorname{NR}_a^v\left(\mathbf{y}\right),\{0,v\}\right)$ is known, we can recover $\operatorname{Q}_a\left(v\mid\cdot\right)$ (by Proposition~\ref{prop: identification of ratios of Q}) and then $\operatorname{R}_a\left(v\mid \operatorname{NR}_a^v\left(\mathbf{y}\right),\{0,v\}\right)$ from our recent ideas. Applying the same argument to menus of size three, we can identify $\operatorname{R}_a$ for sets of size three, and so on. 

\begin{proposition} \label{prop: identification of all}
Suppose the assumptions of Proposition~\ref{prop: counterfactual CCP} are satisfied. If, in addition, we have that $\abs{\mathcal{NC}_a \setminus \mathcal{NR}_a } \geq Y-1$ and, for each $v\neq0$, either $\operatorname{Q}_{a}\left( v \mid n_1\right)$ or $\operatorname{R}_{a}\left(v \mid n_2, \{0,v\} \right)$ is known for some $n_1$ and $n_2$ in the support, then $\operatorname{Q}_a$ and $\operatorname{R}_a$ are identified from $\operatorname{P}_a$.
\end{proposition}

The assumption that either $\operatorname{Q}_{a}\left( v \mid n_1\right)$ or $\operatorname{R}_{a}\left(v \mid n_2, \{0,v\} \right)$ is known for some $n_1$ and $n_2$ in the support can be satisfied in different settings. For example, it is satisfied if the default is never picked when it is part of a binary menu with some alternative $v$ and when all preference peers pick $v$ (i.e., $\operatorname{R}_{a}\left(v \mid \abs{\mathcal{NR}_a}, \{0,v\} \right)=1$). This may happen when one decides whether to use a particular social media with the default choice being not to use any. It would be reasonable to think that if all the friends of a given person are using this social media, then the person will use it for sure. Another example when the assumption is satisfied is when the alternative is considered with probability 1 if enough (or all) consideration peers pick the alternative (i.e., $\operatorname{Q}_{a}\left(v \mid \abs{\mathcal{NC}_a} \right)=1$). In the case of online games, this is the same as to say that a player considers for sure a game when all her peers have just played it. 

\subsection{Identification of \texorpdfstring{$\operatorname{P}$}{the Conditional Choice Probabilities}}\label{sec: identification of CCPs}
\noindent This section studies identification of the CCPs, $\operatorname{P}$, and the rates of the Poisson alarm clocks from two different datasets. In Dataset 1, the researcher observes the precise moment at which an agent revises her strategy and the configuration of choices at that time. In Dataset 2, the researcher observes the configuration of choices at fixed time intervals.

Assume the researcher observes agents' choices at time intervals of length $\Delta$ and can consistently estimate $\Pr\left(\mathbf{y}^{t+\Delta }=\mathbf{y}'\mid\mathbf{y}^{t}=\mathbf{y}\right)$ for each pair $\mathbf{y}',\mathbf{y}\in{\mathcal{Y}}^{A}$. We capture these transition probabilities by a matrix $\mathcal{P}\left( \Delta \right)$.\footnote{Here again, we assume that the choice configurations are ordered according to the lexicographic order when we construct $\mathcal{P}\left(\Delta \right)$.} Let $e^{\left( \Delta \mathcal{W}\right)}$ be the matrix exponential of $\Delta \mathcal{W}$. Then $\mathcal{P}\left( \Delta \right)$ relates to transition rate matrix $\mathcal{W}$ in Section~\ref{sec: equilibrium and mistakes} by $\mathcal{P}\left( \Delta \right) =e^{\left( \Delta \mathcal{W}\right) }$.

The two datasets we consider differ regarding $\Delta$: In Dataset 1, the time interval is very small, i.e., the researcher knows $\lim_{\Delta
\rightarrow 0}\mathcal{P}\left( \Delta \right)$. This ideal dataset registers agents' choices at the exact time at which any given agent revises her choice. With the proliferation of online platforms and scanners, this kind of data is often available. In Dataset 2, the time interval is of arbitrary size, i.e., the researcher knows $\mathcal{P}\left(
\Delta \right)$. In both cases, the identification question is whether (or under what extra restrictions) we can recover $\mathcal{W}$ from the transition probabilities in $\mathcal{P}\left( \Delta \right)$ which are identified and estimated from the data directly.

\begin{proposition}[Dataset 1]\label{ID2} If Assumptions~\ref{ass: MM},~\ref{ass: A1}(i), and~\ref{ass: A2}(i) hold, then the CCPs $\operatorname{P}$ and the rates of the Poisson alarm clocks $(\lambda_a)_{a\in\mathcal{A}}$ are identified from Dataset 1.
\end{proposition}

The proof of Proposition~\ref{ID2} follows because when the time interval between the observations goes to zero, we can recover $\mathcal{W}$. At least two known cases produce the same result without assuming $\Delta \rightarrow 0$. One of them requires the length of the interval $\Delta$ to be below a threshold $\overline{\Delta }$. The issue of this approach is that the value of the threshold depends on details of the model that are unknown to the researcher. The second case requires the researcher to observe the dynamic system at two different intervals $\Delta_{1}$ and $\Delta_{2}$ that are not multiples of each other (see, for example, \citealp{blevins2017identifying} and the literature therein). The following proposition, based on Theorem~1 in \citet{blevins2018identification}, offers a third case in which the transition rate matrix can be identified from Dataset 2. 

\begin{proposition}[Dataset 2]\label{ID3}
If Assumptions~\ref{ass: MM},~\ref{ass: A1}(i), and~\ref{ass: A2}(i) hold, and $\mathcal{W}$ has distinct eigenvalues that do not differ by an integer multiple of $2\pi i/\Delta $, where $i$ denotes the imaginary unit, then $\operatorname{P}$ and $(\lambda_a)_{a\in\mathcal{A}}$ are generically identified from Dataset 2.
\end{proposition}
The restriction on eigenvalues of $\mathcal{W}$ is a regularity condition that is generically satisfied.\footnote{See \citet{blevins2017identifying} for a discussion of this assumption.} The key element in proving Proposition~\ref{ID3} is that the transition rate matrix in our model is rather parsimonious since, at any given time, only one agent revises her selection with a nonzero probability. Thus, the transition rate matrix $\mathcal{W}$ has many zeros in known locations. 

\section{Extensions}\label{sec: extensions}

\subsection{History Dependence and Own Past Choices}\label{sec: extensionslong}
\noindent We have assumed that the choices made by a given agent are only affected by the current aggregate choices made by her peers and ignore her own past choices. We next extend the model by allowing that both the consideration and preferences of a given agent depend on the history of her own choices and those of her peers. As consideration probabilities can be 1, the dependence on past choices allows nontrivial dynamics in consideration sets. For instance, they may not change much over long periods. Thus, this extension allows us to accommodate, among others, persistence in consideration sets and choices ---see the discussion after Assumption~\ref{ass: A1} in Section~\ref{sec: model}. We use these ideas in our empirical application.

There are many ways in which history can be embedded into the model. We propose here a possibility that allows us to model an interesting situation (described below) and requires minimal extra notation. Let $\{t_{k}\}_{k=1}^{+\infty}$ be an (increasing) sequence of random time periods in which the clocks of different agents went off. Let $\mathbf{y}_{t_k}$ denote the configuration of choices in the network at the $k$-th time period (at this moment the alarm clock of some agent went off). As a result, we can encode the whole history of choice configurations before moment $t$ as $h_t=(\mathbf{y}_{t_{k}})_{t_k<t}$. Next, assume that the choice rules and consideration probabilities depend not only on choices made by peers at the moment at which the choice is revised but also on the whole history of choices $h_t$. Hence, given the history of choice configurations $h_t$, the probability that alternative $v$ is picked by Agent $a$ at time $t$ would be
\begin{align*}
\operatorname{P}_{a}\left( v \mid \mathbf{y}_t, h_t\right) =&\sum\nolimits_{ 
\mathcal{C}\subseteq \mathcal{Y}}\operatorname{R}_{a}\left( v \mid \mathbf{y}_t, h_t,\mathcal{NR}_{a},\mathcal{C}\right)\\
&\prod\nolimits_{v'\in \mathcal{C}}\operatorname{Q}_{a}\left( v' \mid \mathbf{y}_t, h_t,\mathcal{NC}_{a}\right) \prod\nolimits_{v'\in \mathcal{Y}\setminus\mathcal{C}}\left( 1-\operatorname{Q}_{a}\left(
v' \mid \mathbf{y}_t, h_t,\mathcal{NC}_{a}\right) \right).
\end{align*}
None of our previous results use variation beyond the choices made by connected agents at the moment of making a decision. Hence, if we condition on the choice made by Agent $a$, $y_{at}$, and the history $h_t$ of choices, then we can establish the identification of all parts of the model from $\operatorname{P}_a$ using our previous ideas ---thus, we omit the proof of the next result. 

\begin{proposition}\label{prop: history from P}
    Suppose Assumptions~\ref{ass: MM} - \ref{ass: A3} are satisfied conditional on $y_{at}$ and the history $h_t$ for all possible $y_{at}$ and $h_t$. Also, let us extend the definition of $\operatorname{P}^*_a$ to allow for dependence on $y_{at}$ and $h_t$. Then, all propositions from Section~\ref{PIP} are still valid. 
\end{proposition}

Proposition~\ref{prop: history from P} takes as an input the CCPs that (now) depend on the histories of choices made by everyone in the network, i.e., it is implicitly assumed that $\operatorname{P}_{a}$ is identified. Since we only observe choices made by agents from one network, it would be impossible to identify the CCPs conditional on \emph{all} histories without further assumptions. To address this difficulty, we could restrict the length of the history that affects $\operatorname{P}_a$.\footnote{CCPs could also be identified even if they depend on the whole choice history if one requires the impact of the remote past to decay sufficiently fast with time (see, \citealp{hardle1997review,bierens1996topics}, and \citealp{truquet2023strong} for examples).} This can be done by assuming that there exists finite $K\geq1$ such that $\operatorname{Q}_a$ and $\operatorname{R}_a$ depend only on the first $K$ components of $h_t$. Under this restriction, we would get that for any $k>K$, $(\mathbf{y}_{t_{k'}})_{k'=1,\dots,k}$, $v$, and $a$
    \[
    \operatorname{P}_{a}\left( v \mid \mathbf{y}_{t_k}, (\mathbf{y}_{t_{k'}})_{k'=1,\dots, k-1}\right)=\operatorname{P}_{a}\left( v \mid \mathbf{y}_{t_k}, (\mathbf{y}_{t_{k'}})_{k'=k-K,\dots, k-1}\right).
    \]
Hence, $\operatorname{P}$ can be recovered from Dataset 1.\footnote{Alternatively, one can also restrict the history in terms of the length of the time period rather than the number of actions. We could assume that there exists $\bar{t}>0$ such that 
$
\operatorname{P}_{a}\left( v \mid \mathbf{y}_{t_k}, (\mathbf{y}_{t_{k'}})_{k'=1,\dots, k-1}\right)=\operatorname{P}_{a}\left( v \mid \mathbf{y}_{t_k}, (\mathbf{y}_{t_{k'}})_{\{k'\::\:k'<k,\:t_{k}-t_{k'}<\bar{t}\}}\right).
$}

To motivate the details of this extension let us go back to the example of the online platform that offers video games to a set of players. As we mentioned in the introduction, these platforms often allow agents to form social networks and make the last purchased or played by peers game visible to the agent. Based on the history of acquired games, the platforms could also share with their subscribers the last few games that were acquired, and the identity of the players that acquired them. One could argue that recently acquired games could receive further attention by the subscribers of the platform. Recent games played by a subscriber could also have a special effect on her consideration set. 

\setcounter{example}{\value{eg2}}
\begin{example}[continued] Let $a^*(t_{k-1})$ be the agent that made a choice at the $t_{k-1}$ moment. Hence, $\mathbf{y}_{a^*(t_{k-1}),t_k}$ is the choice that Agent $a^*(t_{k-1})$ made. Assume that consideration of an alternative (in addition to the previous arguments) depends on whether that alternative is the most recent choice made by a consideration peer and the alternative of the agent in the current choice configuration. That is, 
\[
\operatorname{Q}_{a}\left( v \mid \mathbf{y}_{t_{k}}, h_{t_{k}},\mathcal{NC}_{a}\right)=\operatorname{Q}_{a}\left( v \mid y_{a,t_k}, \operatorname{NC}_{a}^v(\mathbf{y}_{t_{k}}), \Char{\mathbf{y}_{a^*(t_{k-1}),t_k}=v}\Char{a^*(t_{k-1})\in\mathcal{NC}_a}\right).
\]
Thus, histories of length $K=1$ affect the CCPs. For $\mathbf{y}_{t_{k}} = \left(1, 1, 2, 0\right)$ and $\mathbf{y}_{t_{k-1}} = \left(1, 0, 2, 0\right)$, we have that $a^*(t_{k-1})=2$ and $\mathbf{y}_{a^*(t_{k-1}),t_k}=1$. Thus, the consideration probabilities of Agent 1 are 
$\operatorname{Q}_{1}\left( 1 \mid 1, 1, 1\right)$ and $\operatorname{Q}_{1}\left( 2 \mid 0, 1, 0\right)$ since option 1 (and not 2) was chosen at $t_{k-1}$.
\hfill $\square$
\end{example}
\setcounter{example}{\value{aux}}

\subsection{Nonobservable Default}\label{sec: nodefault}
\noindent In many settings, the decision to choose the default alternative is often not observed. For example, if the default is ``do nothing,'' then at any point in time that there is no change in the behavior of a given agent, we do not know whether she woke up and decided to do nothing or she did not have an opportunity to make a new decision. When this happens, even in a continuous-time data setting (Dataset 1), there is no hope to separately identify $\lambda_a$ and $\operatorname{P}_a$. Therefore, some form of normalization is required. In our empirical application, we find it convenient to assume that $\lambda_a=1$. This implies that, on average, agents have an opportunity to make a choice once per unit of time. Once $\lambda_a$ is normalized, we can identify the CCPs $\operatorname{P}_a$ from the data directly, with which we can follow the identification results for network structure, consideration probabilities, and choice rules. 

\section{Application}\label{sec: model application}

\noindent We investigate peer effects in consideration and payoffs on the expansion decisions of the two dominant tea chains in the high-end tea industry in China. We have three goals: First, we showcase our identification strategy and provide a practical estimation procedure. Second, we show that ignoring the presence of limited consideration might mislead our estimates of profitability of different markets. Third, we quantify the direct effect of limited consideration and peer effect in consideration on the dynamics of market structure. 

The tea beverage industry in China has seen a rapid expansion with its overall revenue increasing from 42.2 to 83.1 billion yuan from 2017 to 2020. This industry is divided into three segments: high-end, middle, and low-end. We study the two leading tea firms in the high-end segment ---Heytea and Nayuki--- which did not accept franchising before 2022. We acquired city-level store registration data from a commercial provider that sources records from the National Enterprise Credit Information Publicity System \citep{registrationdata_2021}. This dataset allows us to determine when the enterprise enters and exits the market, enabling us to construct the cumulative number of stores in each city. We supplement the registration data with regional information from the China City Statistical Yearbook \citep{NBSC}.\footnote{Online Appendix C.1 offers more details on the dataset we use.} To avoid any changes in demand caused by the COVID-19 pandemic, we restrict our sample until the end of 2020. By then, Nayuki had $485$ stores in $57$ cities, while Heytea had $729$ stores in $46$ cities.\footnote{The tea chain dataset we acquired for the empirical exercise draws on aspects of the city-level store registration data from a larger coffee chain dataset collected by Ping Xiao (with funding support from the Melbourne Business School Internal Competitive Grant 2022) to jointly study with Ruli Xiao the impact of digital transformation of two major coffee chains in China.}

\subsection{Empirical Model}
\noindent We first describe the model of firm expansion decisions and then introduce the specifications for consideration and payoffs. We define a market at the level of a prefecture-level city, which ranks below a province and above a county in China's administrative structure. Thus markets are geographically isolated from each other. We collect all unknown primitives by $\theta$.\bigskip

\noindent \textbf{Choice Set, Agents, and Peers} There are finite sets of firms, $\mathcal{F}$, and markets to expand to, $\mathcal{M}$. Every firm $f$ decides whether to open a store ($v=1$) in market $m$ or not ($v=0$). We call a pair $a=(f,m)\in\mathcal{F}\times\mathcal{M}$ an agent ---knowing the firm and the market identifies the agent and vice versa. Thus, $\mathcal{A}=\mathcal{F}\times\mathcal{M}$ and $\mathcal{Y}=\{0,1\}$.\footnote{We abuse notation a bit since $\mathcal{A}$ was previously defined as a set of the form $\{1,2,\dots, A\}$. To be consistent with the initial notation, we can take any one-to-one mapping $\tilde{a}:\mathcal{F}\times\mathcal{M}\to\{1,2,\dots, \abs{\mathcal{F}\times\mathcal{M}}\}$ and define an agent as $a=\tilde{a}(f,m)$.} The set of markets in which these firms can open a new store is quite large, and the data correspond to a time when these firms were relatively new in the industry. We argue that managers might circumscribe the set of markets they consider at a given time to facilitate the decision process. At the moment of deciding whether to open a store, the attention that firm $f$ pays to market $m$ depends on its own and competitor's past choices in market $m$; it also depends on past openings at ``neighboring'' markets. Formally, $\mathcal{NC}_a$ and $\mathcal{NR}_a$ are the sets of pairs of firms and markets that affect consideration and payoffs, respectively, of firm $f$ in market $m$. 

As we just described, each firm decides whether to open a new store in each market. At the end of the analysis of network identification in Section~\ref{PIP}, we modified Example 1 to show that with only one non-default option we cannot rely on double differences to recover all parts of the network structure, and thereby need to assume partial knowledge of it. We follow the literature \citep[e.g.,][]{arcidiacono2016estimation} and assume the marginal profit of firm $f$ in market $m$ from opening a new store is only affected by both its own and its competitors' openings in market $m$. Formally, $(f',m')\in \mathcal{NR}_{(f,m)}$ if and only if $m=m'$. After constructing $\mathcal{NR}_a$ in this way, we estimate the consideration network from the data. A potential concern with this assumption is the possibility of spillover effects across markets on profits. For example, the profits in different markets may be correlated through shipping cost savings from the distribution chain (see, for instance, \citealp{jia2008happens,holmes2011diffusion,zheng2016spatial,Houde2023}). We believe this is unlikely in our application as Nayuki relies on third-party logistics for shipping instead of building its own distribution centers. Moreover, its financial report states that storage and shipping accounted for approximately only $1.9$ percent of total revenues in 2020.\footnote{\url{http://www.cn156.com/cms/scm/105854.html}, Assessed August 2025.}\footnote{We had not found any reliable sources for Heytea. But, being similar in other respects, we suspect the same argument applies to Heytea.} 
Another potential source of spillover effects is information aggregation: the more stores a firm has in the area, the more information it has about the profitability of a particular market. In Online Appendix C.3, we re-estimate the model allowing the number of own stores in the nearby markets to affect firm's profitability in the focal market ---keeping the assumption that rival's stores in the nearby markets only affect consideration. We show the main results are robust to these potential spillover effects.

\bigskip 
\noindent\textbf{Observable Characteristics} Every market $m$ at every moment of time $t$ is characterized by observed market characteristics $S_{mt}$ (e.g., GDP and population density) that include a constant. Let $N_{at}$ denote the number of stores of Agent $a$ (i.e., the number of stores of firm $f$ in market $m$). Also define $S_t=(S_{mt})_{m\in\mathcal{M}}$ and $N_{t}=(N_{at})_{a\in\mathcal{A}=\mathcal{F}\times\mathcal{M}}$.

\bigskip 
\noindent\textbf{Market Consideration} In our application, there are 71 markets where at least one firm opened a store by the end of the measurement. As we already justified, we allow firms to consider only a subset of markets when making a decision. Given the numbers of stores each firm had in the market at time $t$, $N_t$, the probability that firm $f$ considers opening a new store in market $m$ at time $t$ is
\[
\operatorname{Q}_a(1\mid N_t,S_t,\mathcal{NC}_a)=\operatorname{F}_{\tilde\varepsilon}\left(\bar{\tilde{\pi}}_{at}(S_{t},N_{t};\theta)\right),
\]
where $\operatorname{F}_{\tilde\varepsilon}$ is a known c.d.f.; $\theta$ is the vector of unknown parameters; and $\bar{\tilde{\pi}}_{at}(S_{t},N_{t};\theta)$ is the mean attention index, which is known up to $\theta$. We allow the current market features of market $m$ (including the market characteristics and all firms' number of stores) to affect the attention index of Agent $a$. Moreover, we allow the market structure of Agent $a$'s neighborhood markets to affect her attention to market $m$.

\bigskip
\noindent\textbf{Payoff from a New Store} Conditional on a market being considered, the firm decides whether to open at least one new store in that market based on its marginal profit $\pi_{at}$. This marginal profit captures not just the instantaneous (one period) profitability of an extra store, but the expected profitability of the store in the long run.\footnote{We do not explicitly model forward-looking behavior to focus on the peer effect in limited consideration on firms' decisions. Incorporating forward-looking behavior can be computationally intensive and requires additional assumptions about players' expectations and beliefs, which is beyond the scope of this paper.} The probability of opening a new store in market $m$ by firm $f$ at time $t$ conditional on it being considered is
\[
\operatorname{R}_a(1\mid N_t,S_t,\mathcal{NR}_a,\{0,1\})=\operatorname{F}_{\varepsilon}\left(\bar{\pi}_{at}(S_{t}, N_{t};\theta)\right),
\]
where $\operatorname{F}_{\varepsilon}$ is a known c.d.f. and $\bar{{\pi}}_{at}(S_{t},N_{t};\theta)$ is the mean marginal profit ---known up to $\theta$. 

\bigskip 
\noindent\textbf{Model Implied CCP} Altogether, the probability that firm $f$ opens a new store in market $m$ (with $a=(f,m)$) is
\begin{equation*}
\operatorname{Pr}_a(1\mid N_t,S_t;\theta)=\operatorname{F}_{\tilde\varepsilon}\left(\bar{\tilde{\pi}}_{at}(S_{t},N_{t};\theta)\right)\operatorname{F}_{\varepsilon}\left(\bar{\pi}_{at}(S_{t}, N_{t};\theta)\right),
\end{equation*}
which completely characterizes the probability of observing a new store in a given market by a given firm conditional on the history and the market characteristics. When evaluated at the true parameter value $\theta_0$, $\operatorname{Pr}_a$ matches the CCP $\operatorname{P}_a$, i.e., $\operatorname{P}_a(1\mid N_t,S_t)=\operatorname{Pr}_a(1\mid N_t,S_t;\theta_0)$.

The vector of parameters $\theta$ contains the parameters entering $\bar{\tilde{\pi}}_{at}$ and $\bar{{\pi}}_{at}$ that relate to both the covariates and the consideration network structure $\mathcal{NC}_{a}$, $a\in\mathcal{A}$. Note that $\mathcal{NR}_a=\{(f',m')\::\: f'\neq f, m=m'\}$ is assumed to be known and, thus, is not a part of $\theta$.

\subsection{Estimation} 
\noindent \textbf{Data} The data we have consist of three objects: (i) the exact date of store openings $\{t_k\}_{k=1}^K$; (ii) the state of the market structure $\{N_{at_k}\}_{a\in\mathcal{A},k=1,\dots,K}$ sampled from a continuous time over interval $[0,t_K]$, where $N_{at_k}$ is the number of stores owned by firm $f$ in market $m$ immediately prior to $k$-th change at time $t_k$ ---the last date of measurements coincides with the last day in which any action was observed; and (iii) observable market characteristics $\{S_{m,t_k}\}_{a\in\mathcal{A},k=1,\dots,K}$.

\bigskip
\noindent \textbf{Likelihood Function} Our identification argument is constructive and can be used to estimate the model nonparametrically. However, for small and moderate-sized samples, we suggest using the parametric maximum likelihood estimator of CCPs $\operatorname{P}_a$, as it is reasonably flexible in allowing market-specific consideration network links and efficiently uses all variations across markets (see Online Appendix B for details on the nonparametric estimator). We also add the network links in the parametrization of CCPs to estimate the CCPs, these network links, and the other consideration and payoff parameters in one step ---instead of first estimating the CCPs and then applying our identification argument to estimate the rest of the model. In particular, we construct from the data a state vector $r_{t_k}=(r_{at_k})_{a\in\mathcal{A}}$, where $r_{at_k}$ indicates whether there is a change in the number of stores of firm $f$ in market $m$ at time $t_k$, i.e.,
$
r_{at_k}=\Char{N_{at_{k+1}}>N_{at_k}}.
$ The probability of observing $r_{t_k}$, given the data and model parameters $\theta$ conditional on an alarm clock going off, is 
\begin{align*}
p(r_{t_k},S_{t_k},N_{t_k};\theta) &=\prod_{a:r_{at_k}=1} \operatorname{Pr}_a(1\mid N_{t_k},S_{t_k};\theta) \times\prod_{a:r_{at_k}=0}\Big[1- \operatorname{Pr}_a(1\mid N_{t_k},S_{t_k};\theta) \Big].
\end{align*}
Hence, the probability that no new stores are opened in any market by any firm, given market characteristics and number of stores already opened (probability of picking the default), is
\begin{align*}
p_{0}(S_{t_k},N_{t_k};\theta) &=
\prod_{a\in\mathcal{A}}\Big[1-\operatorname{Pr}_a(1\mid N_{t_k},S_{t_k};\theta)\Big].
\end{align*}
Finally, given that the arrival process is exponential, the log-likelihood of observing the data given $\theta$ and normalizing $\lambda_a = 1$ (the choice of ``doing nothing'' is not observed; see Section~\ref{sec: nodefault} for details on this normalization) is
\begin{align*}
    \hat{L}(\theta)=\sum_{k=1}^{K} -(t_{k+1}-t_{k})\lambda (1-p_{0}(S_{t_k},N_{t_k};\theta))+\ln (\lambda p(r_{t_{k}},S_{t_k},N_{t_k};\theta)).
\end{align*}
The maximum likelihood estimator of $\theta$, $\hat{\theta}$, is defined as the maximizer of $\hat{L}$ over a parameter space $\Theta$.\footnote{Checking all possible network structures is not feasible in our application. We use a variation of a greedy algorithm. See Online Appendix C.2 for further details.} The estimator $\hat{\theta}$ leads to an estimator of CCPs
\[
\hat{\operatorname{P}}_a(1\mid N_t,S_t)=\operatorname{Pr}_a(1\mid N_t,S_t;\hat{\theta}).
\]
The construction of confidence sets for the parameters and estimated CCPs would require taking into account the estimation error in the estimated network. We leave this difficult problem for future research. An important feature stemming from Assumptions~\ref{ass: MM},~\ref{ass: A1}(i), and~\ref{ass: A2}(i), as shown by Proposition~\ref{prop: unique equilibrium}, is that our model has a single equilibrium or invariant distribution, i.e., the multiplicity of equilibria in the data-generating process is not an issue in our estimation.

\bigskip
\noindent\textbf{Parameterization} We assume that $\operatorname{F}_{\tilde\varepsilon}$ and $\operatorname{F}_{\varepsilon}$ are Logistic c.d.f. Given that our sample size is small relative to the number of agents ---there are $598$ different dates in which we observed firms opening a store\footnote{For some days, we observe multiple agents opening a store. As we mentioned earlier, the identification results in Section~\ref{PIP} are still valid if agents have perfectly synchronized clocks.} for $71\times2$ agents--- we use the following second-degree polynomial parameterization to flexibly approximate mean marginal profits and mean attention index:\footnote{We have estimated the model under several alternative specifications. The results are qualitatively the same and are available upon request.} 
\begin{align*}
    \bar{\pi}_{at}(S_{t}, N_{t};\theta)=&S_{mt}\tr \beta_{f}+\sum_{f'}\Big[N_{(f',m)t}\alpha_{f,f'}+N^2_{(f',m)t}\gamma_{f,f'}\Big], \\
    \bar{\tilde{\pi}}_{at}(S_{t}, N_{t};\theta)=&S_{mt}\tr \tilde\beta_{f} +\sum_{f'}\Big[N_{(f',m)t}\tilde\alpha_{f,f'}+N^2_{(f',m)t}\tilde\gamma_{f,f'}\Big]+\\
    +&\sum_{f'}\left[\tilde\delta_{f,f'}\sum_{a'':f''=f'}\delta_{m,m''}N_{a''t}+\tilde\eta_{f,f'}    \left(\sum_{a'':f''=f'}\delta_{m,m''}N_{a''t}\right)^2\right].  
\end{align*}
The parameter $\delta_{m,m''}\in\{0,1\}$ captures the consideration network structure. It is equal to $1$ if stores in market $m''$ affect consideration in market $m$ and it is $0$ otherwise. Note that estimation of the consideration network and probabilities crucially relies on variation of the number of openings by the two firms in nearby markets. For example, $\hat \delta_{m,m''}=1$ if allowing the variation of openings in market $m''$ to affect the CCPs of market $m$ increases the value of the likelihood function. That is, as in our identification strategy, the estimates add market $m''$ to the consideration network of market $m$ if changing the number of stores in market $m''$ affects firms' CCPs in market $m$. Therefore, it is the variation in openings that allows us to pin down the values of the parameters that define the different parts of the model in the maximum likelihood estimation.

The mean marginal profit has two parts: the first one captures the impact of the observable market characteristics, and the second part captures the impact of the number of stores all firms have in market $m$.\footnote{The fully structural model of marginal profits should contain information on fixed and marginal costs and prices among many other things. We specify the marginal profit function in the reduced form because of the availability of the data and to simplify the analysis.} The mean attention has an additional part capturing the peer effect in consideration from markets different from $m$, i.e., we allow firm-specific peer effects.\footnote{We use a second-degree polynomial in our approximation to capture potential nonmonotonicities. For instance, in markets with few stores, firms may act as complements, while, in more mature markets, firms may become substitutes. This would result in marginal profits that are not monotonic in competitors' stores.} 

The parameterization we use imposes two restrictions that are not needed for identification but reduce the computational burden: First, the payoff and consideration parameters are firm-specific and do not change across markets. Second, the consideration network link parameters $\delta_{m,m''}$ (that capture $\mathcal{NC}_a$) vary across markets but not across firms. 

\begin{rem}
    Since we use the total number of stores opened by each firm in every market as the determinant of consideration and expansion probabilities, formally we have a model with infinite history dependence. This, however, does not constitute any issues in our application, since the mean attention and mean marginal profits take known parametric forms. 
\end{rem}

\subsection{Estimation Results}
\noindent\textbf{Network Structure} The estimated consideration network has $266$ directed links (i.e., the adjacency matrix is not symmetric and has $266$ nonzero elements). The initial network we used in optimization uses spatial information about markets (see Online Appendix C.2) and allows up to $563$ links, so we use the likelihood value to close down almost $300$. 

\bigskip
\noindent\textbf{Consideration} We calculated the consideration probabilities for every market and each firm using the numbers of stores and covariate values observed at the beginning (i.e., when Heytea started operating) and at the end of the measurements. Figure~\ref{fig: con hist 12} shows the fraction of markets as a function of consideration probabilities for Heytea and Nayuki, respectively. Both firms display substantial limited consideration at the beginning: the averages (standard deviations) across markets of the consideration probabilities are $0.005$ ($0.006$) and $0.029$ ($0.03$) for Heytea and Nayuki, respectively. By the end of 2020, as a consequence of increases in the number of stores that the firms have in different markets, the consideration probabilities became substantially larger, with Nayuki becoming an almost full consideration firm: the averages (standard deviations) across markets of the consideration probabilities are $0.027$ ($0.04$) and $0.82$ ($0.38$) for Heytea and Nayuki, respectively.\footnote{As anecdotal evidence for our findings, the news reported that Nayuki implemented a nationwide city expansion plan. Beginning in late 2017, expansion moved beyond the province of Guangdong, rapidly extending into South China, Central China, East China, and other regions (\url{https://news.qq.com/rain/a/20230109A03FOE00?utm.com}, Accessed August 2025). In contrast, Heytea announced in 2020 its plan to continue focusing primarily on first-tier and provincial capital cities in China (\url{https://news.sina.cn/gn/2020-07-14/detail-iivhvpwx5274609.d.html}, Accessed August 2025).} We next show that incorporating limited attention is important for obtaining accurate estimates of profitability of markets and that it affects market structure ---which impacts consumer welfare. 

\begin{figure}[h!]
\centering
  \includegraphics[width=0.75\textwidth]{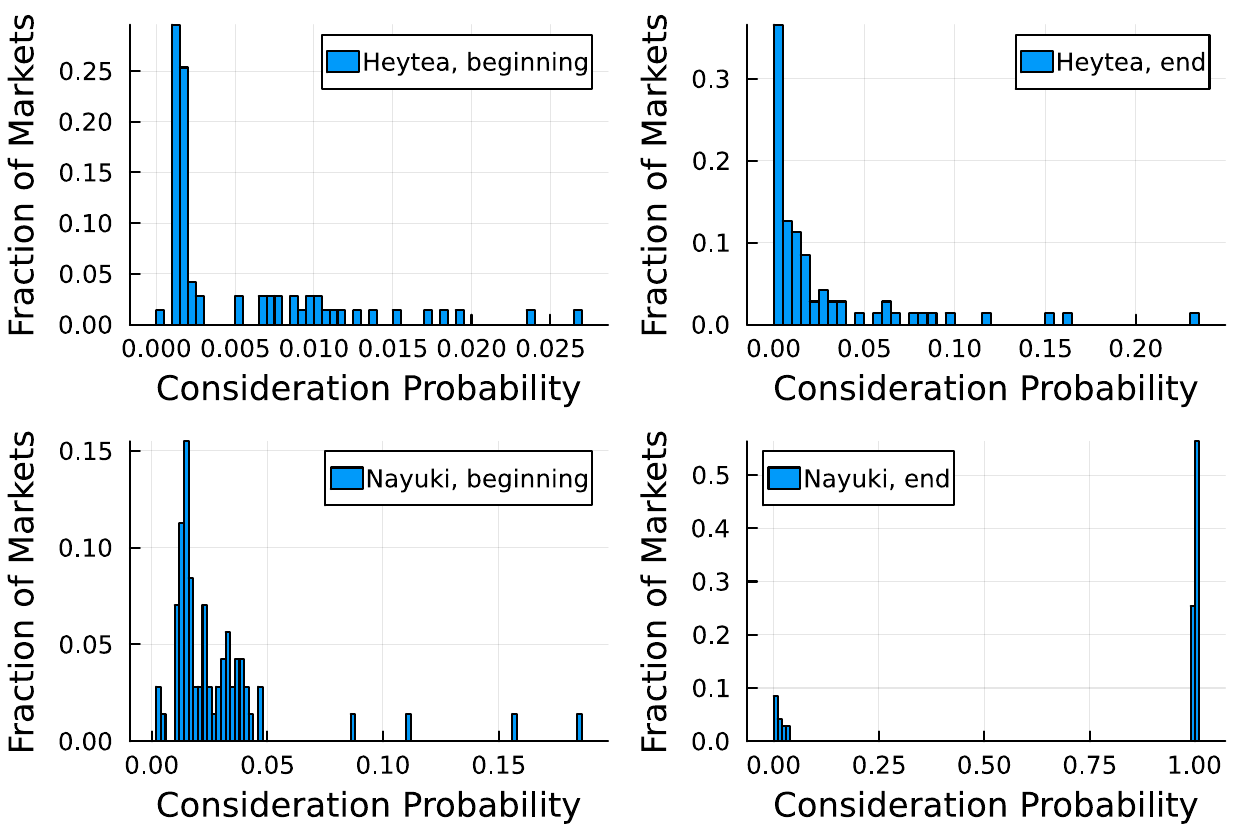}
\caption{\textbf{Normalized histogram of consideration probabilities for both firms at the data's beginning and end.} }
\label{fig: con hist 12}
\end{figure}

\bigskip
\noindent\textbf{Marginal Profits Estimates} We analyze the probabilities of opening a new store across markets (conditional on being considered). To quantify the effect of adding limited consideration to the expansion decision, we also estimated the marginal profit parameters assuming that all markets are considered. We refer to the former as limited consideration estimates and to the latter as full consideration estimates. The results of the estimation are presented in Figure~\ref{fig: exp hist 12}. The difference between the limited and full consideration expansion probabilities is striking. In the beginning, the full consideration model would substantially underestimate the expansion probabilities (all points are below the $45^{\circ}$-line). By the end of 2020, while the discrepancy almost disappeared for Nayuki, which is not surprising given that it became an almost full consideration firm, the expansion probabilities for Heytea are still heavily underestimated. Thus, ignoring limited consideration leads to completely misleading estimates about the profitability of different markets. Qualitatively, this difference is explained by the fact that the full consideration model attributes ``not-opening'' a new store to negative marginal profits instead of limited consideration.

\begin{figure}[h!]
\centering
  \includegraphics[width=0.75\textwidth]{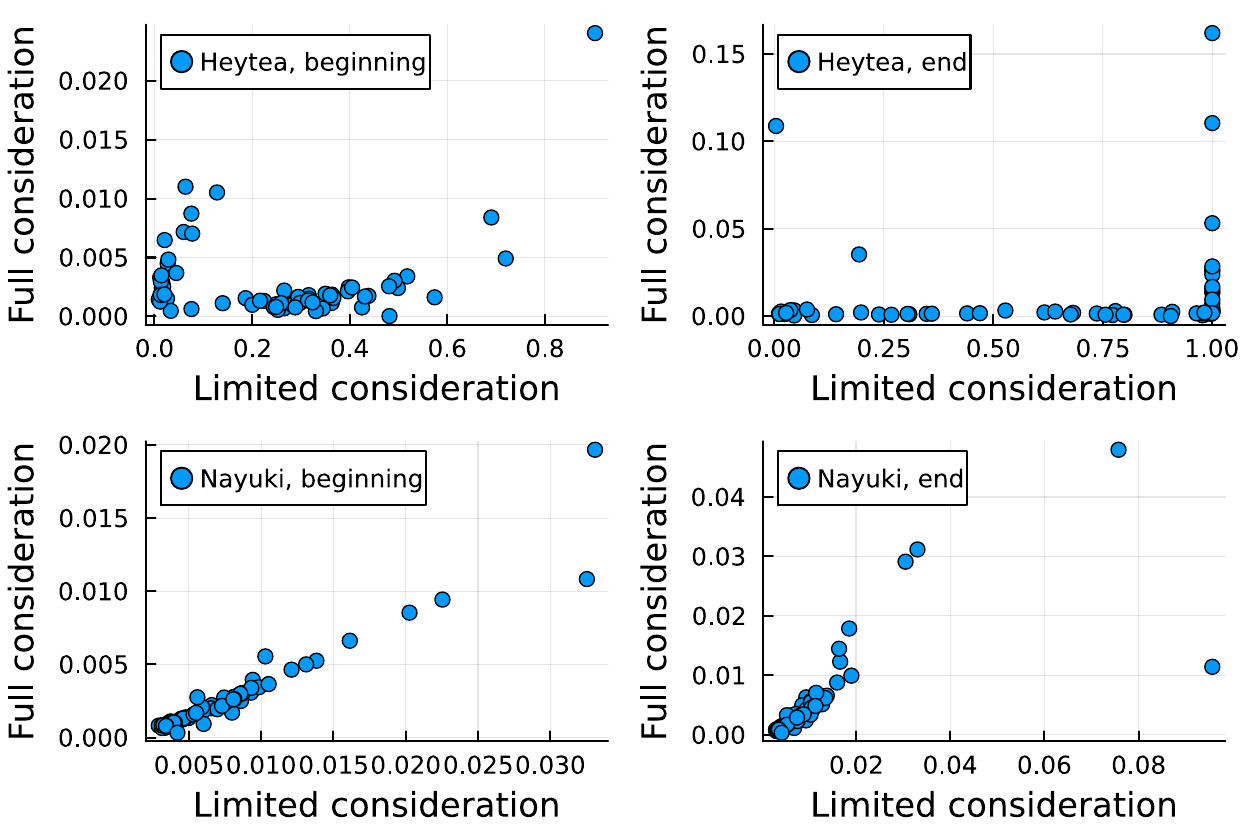}
\caption{\textbf{Limited consideration vs. full consideration expansion probabilities.} }
\label{fig: exp hist 12}
\end{figure}

\subsection{Counterfactuals}
\noindent We evaluate the effect of limited consideration on market structure by comparing the fraction of monopolistic, duopolistic, and markets that are not served across time between our limited consideration model and a situation in which firms were fully attentive, i.e., the firms consider all 71 markets but marginal profit functions are kept the same as in our model estimates. 

The simulation starts with zero stores and then creates expansion decisions for about $500$ days. Figure~\ref{fig: C1} depicts the fraction of monopolistic, duopolistic, and markets not served by any firm as a function of time in the estimated limited consideration setting and the one that imposes full consideration. Limited consideration has a large effect on the dynamics of market structure. With full consideration, almost all markets would be served rather faster. For instance, in less than a year both firms would be present in about half of the markets. 

\begin{figure}[h!]
\centering
\begin{subfigure}{.33\textwidth}
  \centering
  \includegraphics[width=0.9\textwidth]{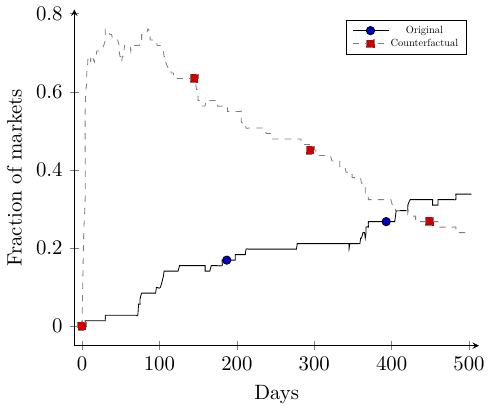}
    \caption{Monopolies.}
    \label{fig: C1 mon}
\end{subfigure}
\begin{subfigure}{.33\textwidth}
  \centering
  \includegraphics[width=0.9\textwidth]{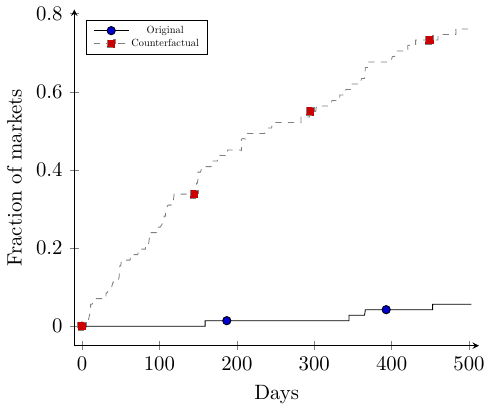}
    \caption{Duopolies.}
    \label{fig: C1 duo}
\end{subfigure}%
\begin{subfigure}{.33\textwidth}
  \centering
  \includegraphics[width=0.9\textwidth]{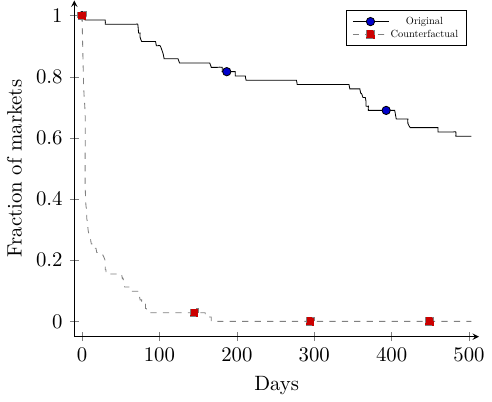}
    \caption{Not served.}
    \label{fig: C1 ns}
\end{subfigure}
\caption{\textbf{Counterfactual Scenario 1. Fraction of monopolistic, duopolistic, and markets that are not served over time.}}
\label{fig: C1}
\end{figure}

In the second counterfactual, we remove connections across markets in consideration. Specifically, we shut down the effect of own and opponent stores in the \emph{neighboring} markets on consideration.\footnote{We abstract away from any potential dependence between the network formation and the expansion decision processes and assume that the network is an exogenously given fixed parameter. This assumption is important for any form of counterfactual analyses that involves changes in the network structure.} The market penetration is not affected much. As part (b) of Figure~\ref{fig: C2} shows peer effects on consideration from the neighboring markets speed up competition.

\begin{figure}[h!]
\centering
\begin{subfigure}{.33\textwidth}
  \centering
  \includegraphics[width=0.9\textwidth]{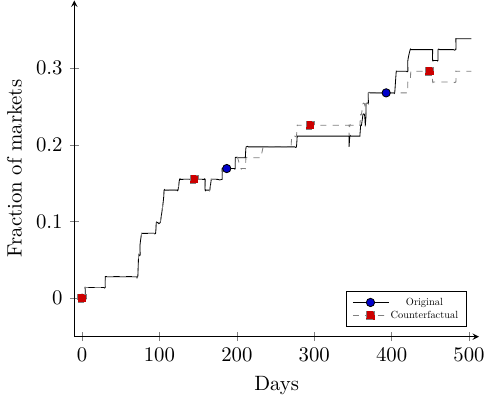}
    \caption{Monopolies.}
    \label{fig: C2 mon}
\end{subfigure}
\begin{subfigure}{.33\textwidth}
  \centering
  \includegraphics[width=0.9\textwidth]{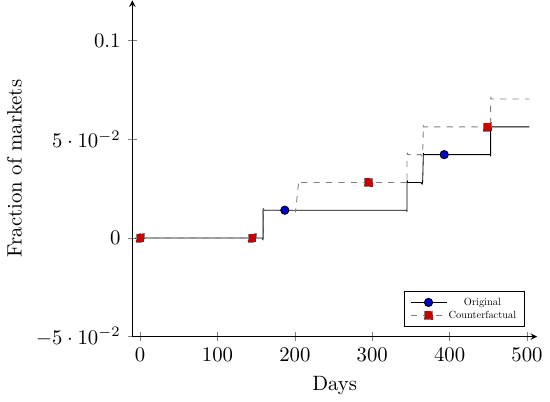}
    \caption{Duopolies.}
    \label{fig: C2 duo}
\end{subfigure}%
\begin{subfigure}{.33\textwidth}
  \centering
  \includegraphics[width=0.9\textwidth]{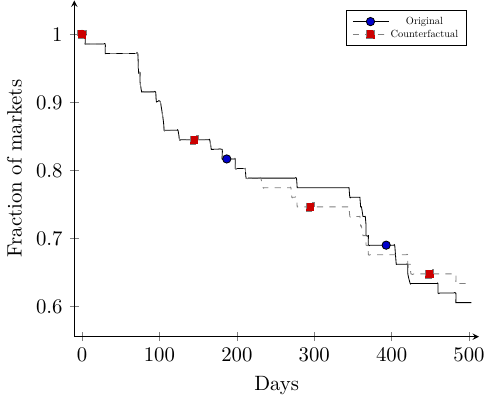}
    \caption{Not served.}
    \label{fig: C2 ns}
\end{subfigure}
\caption{\textbf{Counterfactual Scenario 2. Fraction of monopolistic, duopolistic, and markets that are not served over time.}}
\label{fig: C2}
\end{figure}

\section{Conclusion}\label{sec: conclusion}

\noindent This paper offers a model of interactions in which different types of peers affect the choices of a given agent via different mechanisms. We show that these peer effect mechanisms have different behavioral implications in the data. This allows us to recover the set of connections between the agents and the type of interactions between them. The choices of different types of peers act as dual exclusion restrictions and allow us to recover the consideration probabilities and the random preferences. We apply the model to data on tea chains expansions in China. The empirical application adds to the literature on boundedly rational firms. While studying a rather general model, we leave a few interesting variants for future research, such as forward-looking behavior or the  possibility that each agent ends up making decisions with the purpose of affecting the consideration set of others. We believe this setup could lead to a new model of endogenous social norms.

\section{Data Availability}
\noindent The data sets and replication codes underlying this article are available in Zenodo, at
\url{https://doi.org/10.5281/zenodo.18155977}.
\setstretch{0.9}
\bibliography{references}

\setstretch{1.5}
\newpage
\section*{Online Appendix for ``Peer Effects in Consideration and Preferences''}

\noindent The main paper contains detailed sketches of the proofs of the identification results. This Online Appendix further formalizes these arguments and presents the regularity condition. The Appendix also offers a simulation and an estimation analysis for the running example (Example 1) in the main paper. It finally adds additional details and results to the empirical application.
\bigskip \bigskip

\appendix

\section{Proofs}\label{sec: proof}

\subsection{Proof of Proposition~\ref{prop: unique equilibrium}}
\noindent For an irreducible, finite-state, continuous Markov chain, the equilibrium $\mu$ exists, is unique and has full support. Thus, we only need to prove that Assumptions~\ref{ass: A1}(i) and~\ref{ass: A2}(i) imply that the Markov chain induced by our model is irreducible. Note that
\[
\operatorname{P}_{a}\left( v \mid \mathbf{y}\right) =\sum\nolimits_{ 
\mathcal{C}\subseteq \mathcal{Y}}\operatorname{R}_{a}\left( v \mid \mathbf{y},\mathcal{NR}_{a},\mathcal{C}\right)\operatorname{C}_{a}\left( \mathcal{C}  \mid \mathbf{y},\mathcal{NC}_{a},\mathcal{Y}\right).
\]
Assumption~\ref{ass: A1}(i) implies that given any $\mathbf{y}$, any $v$ is always considered with a positive probability by any Agent $a$. Assumption~\ref{ass: A2}(i) then implies that any option is picked with a positive probability if considered. Thus, $0<\operatorname{P}_{a}\left( v \mid \mathbf{y}\right)<1$ for all $a$ and $\mathbf{y}$, and we can go from one configuration to the other one in less than $A$ steps with a positive probability.

\subsection{The Regularity Condition}

\noindent We formally state and discuss the regularity condition needed for the identification of the network in Section~\ref{sec: model identification} of the main text of the paper. The aim of this condition is to eliminate some ties that would only arise under very unlikely situations ---that we describe in the paper. 

For a given $a\in\mathcal{A}$, define the set of all possible values that $\operatorname{NR}^{\mathcal{Y}}_a(\mathbf{y})$ and $ \operatorname{NC}^{\mathcal{Y}}_a(\mathbf{y})$ can take:
\[
\operatorname{Nrc}_a=\left\{\left(\operatorname{NR}^{\mathcal{Y}}_a(\mathbf{y}), \operatorname{NC}^{\mathcal{Y}}_a(\mathbf{y})\right)\::\: \mathbf{y}\in\mathcal{Y}^A\right\}.
\]
In addition, define $\bar{\operatorname{P}}_a\left(v\mid \mathbf{nr},\mathbf{nc}\right)$ as the probability that Agent $a$ picks option $v\neq 0$ conditional on $(\mathbf{nr},\mathbf{nc})\in\operatorname{Nrc}_a$, where $nr^{v'}$ and $nc^{v'}$ denote the number of peers that affect preference and consideration, respectively, picking alternative $v'$, with $v'\in\mathcal{Y}\setminus\{0\}$. That is,
\[
\bar{\operatorname{P}}_a\left(v\mid \mathbf{nr},\mathbf{nc}\right)=\sum_{\mathcal{C}\subseteq\mathcal{Y}}\operatorname{R}_a\left(v\mid nr^{\mathcal{C}},\mathcal{C}\right)\operatorname{C}_a\left(\mathcal{C}\mid \mathbf{nc},\mathcal{Y}\right),
\]
where $nr^{\mathcal{C}}=\left(nr^{v'}\right)_{v'\in\mathcal{C}\setminus\{0\}}$ and
\[
\operatorname{C}_a\left(\mathcal{C}\mid \mathbf{nc},\mathcal{Y}\right)=\prod_{v'\in\mathcal{C}}\operatorname{Q}_a(v'\mid nc^{v'})\prod_{v'\in\mathcal{Y}\setminus\mathcal{C}}(1-\operatorname{Q}_a(v'\mid nc^{v'})).
\]

Let $\Delta_{v,v'} f(\mathbf{x},\mathbf{y})$ denote an operator that computes the increment of a given function when the $v$-th component of $\mathbf{x}$ and the $v'$-th component of $\mathbf{y}$ are increased by $1$, respectively. We use the convention that if $v=0$, then $\mathbf{x}$ remains unchanged. Similarly, if $v'=0$, then $\mathbf{y}$ remains unchanged.

The next assumption describes the "regularity condition".
\begin{assumption}[Regularity Condition]\label{ass: separability}
    For any $a\in\mathcal{A}$, 
    \begin{enumerate}
        \item there exist $v\in\mathcal{Y}\setminus\{0\}$ and a vector of aggregate peers' choices $(\mathbf{nr},\mathbf{nc})\in \operatorname{Nrc}_a$ such that
    \[
    \Delta_{v,v}\ln\bar{\operatorname{P}}_a(v\mid \mathbf{nr},\mathbf{nc})\neq 0;
    \]
    \item  there exist three sets of alternative pairs and aggregate peers' choices, i.e., $\{v_i, w_i,\mathbf{nr}_i,\mathbf{nc}_i\}$, with $v_i,w_i\in\mathcal{Y}\setminus\{0\}$,  $v_i\neq w_i$,  $(\mathbf{nr}_i,\mathbf{nc}_i)\in \operatorname{Nrc}_a$, and $i=1,2,3$, such that
    \begin{align*}
    \Delta_{w_1,0}\Delta_{v_1,0}\ln\bar{\operatorname{P}}_a(v_1\mid \mathbf{nr}_1,\mathbf{nc}_1)\neq 0,\\
        \Delta_{0,w_2}\Delta_{v_2,0}\ln\bar{\operatorname{P}}_a(v_2\mid \mathbf{nr}_2,\mathbf{nc}_2)\neq 0, and\\
        \Delta_{w_3,w_3}\Delta_{v_3,0}\ln\bar{\operatorname{P}}_a(v_3\mid \mathbf{nr}_3,\mathbf{nc}_3)\neq 0.
    \end{align*}
    \end{enumerate}    
\end{assumption}
Assumption~\ref{ass: separability}(i) ensures that peer effects in consideration and preferences do not cancel out. This guarantees that peers that affect both consideration and preferences can be distinguished from those who are not in the reference group of Agent $a$. Assumption~\ref{ass: separability}(ii) is needed to distinguish peers who affect consideration only from those who affect preference. Specifically, for consideration-only peers, the double shifts described above are always zero. The conditions in Assumption~\ref{ass: separability}(ii) ensure that the double shift in the observed CCPs is nonzero for peers who affect preference in any of the three key scenarios contemplated by Assumption~\ref{ass: separability}(ii). 

It is worth emphasizing that the inequality in Assumption~\ref{ass: separability}(i) is only required to hold for one configuration of actions and peers. Additionally, Assumption~\ref{ass: separability}(ii) allows $v_1=v_2=v_3$, $w_1=w_2=w_3$, and $(\mathbf{nr}_1,\mathbf{nc}_1)=(\mathbf{nr}_2,\mathbf{nc}_2)=(\mathbf{nr}_3,\mathbf{nc}_3)$. Furthermore, as the number of peers  and/or the size of the menu grow, it gets harder to violate Assumption~\ref{ass: separability}. Therefore, Assumption~\ref{ass: separability} is a mild functional form restriction that is usually generically satisfied. 

The following example clarifies the scope of Assumption~\ref{ass: separability}. 

\begin{example}
Suppose that
\[
\operatorname{R}_{a}\left( v \mid nr^{\mathcal{C}},\mathcal{C}\right)=\dfrac{u_v(nr_v)}{\sum_{v'\in\mathcal{C}}u_{v'}(nr_{v'})},
\]
where $u_0(nr_0)=1$ and $u_v(\cdot)$, $v\in\mathcal{Y}\setminus\{0\}$, are strictly monotone positive functions. That is, after the consideration set is formed, Agent $a$ picks alternatives according to a logit-type rule.

Then, for the binary choice case, i.e., $Y=1$ and $v=1$, we have that
\begin{align*}
\bar{\operatorname{P}}_{a}(v\mid nr^{v},nc^{v})&=\operatorname{Q}_{a}\left( v \mid nc^{v}\right)\dfrac{u_1(nr^{v})}{1+u_1(nr^{v})}.
\end{align*}
In this case, Assumption~\ref{ass: separability}(i) is violated if, for all admissible values of $nr^{v}$ and $nc^{v}$, the following equality holds:
\begin{align*}
&\dfrac{\operatorname{Q}_{a}\left( v \mid nc^{v}+1\right)}{\operatorname{Q}_{a}\left( v \mid nc^{v}\right)}
=\dfrac{u_1(nr^{v})}{1+u_1(nr^{v})}\dfrac{1+u_1(nr^{v}+1)}{u_1(nr^{v}+1)}.
\end{align*}
A larger peer group makes it harder for this equality to hold for all values. 

We also illustrate that the same conclusion holds with a larger menu of choices. Specifically, if we add one more alternative $v'=2$, then
\begin{align*}
\bar{\operatorname{P}}_{a}(v\mid nr^{v},nr^{v'},nc^{v},nc^{v'})&=\operatorname{Q}_{a}\left( v \mid nc^{v}\right)\operatorname{Q}_{a}\left( v' \mid nc^{v'}\right)\dfrac{u_1(nr^{v})}{1+u_1(nr^{v})+u_2(nr^{v'})}\\
&+\operatorname{Q}_{a}\left( v \mid nc^{v}\right)(1-\operatorname{Q}_{a}\left( v' \mid nc^{v'}\right))\dfrac{u_1(nr^{v})}{1+u_1(nr^{v})}\\
&=\operatorname{Q}_{a}\left( v \mid nc^{v}\right)\dfrac{u_1(nr^{v})}{1+u_1(nr^{v})}\dfrac{1+u_1(nr^{v})+\left[1-\operatorname{Q}_{a}\left( v' \mid nc^{v'}\right)\right]u_2(nr^{v'})}{1+u_1(nr^{v})+u_2(nr^{v'})}.
\end{align*}
So Assumption~\ref{ass: separability}(i) is violated only if, for $v\in\{1,2\}$,     $\Delta_{v,v}\ln\bar{\operatorname{P}}_a(v\mid \mathbf{nr},\mathbf{nc})= 0$ for all admissible values of $\mathbf{nr}$ and $\mathbf{nc}$. To illustrate how strong this condition is, note that a violation at $\mathbf{nr}=\mathbf{nc}=\mathbf{0}$ would indicate that 
\begin{align*}
&\ln \dfrac{\operatorname{Q}_{a}\left( v \mid 1\right)}{\operatorname{Q}_{a}\left( v \mid 0\right)}
+\ln\left[\dfrac{u_1(1)}{1+u_1(1)}\dfrac{1+u_1(0)}{u_1(0)}\right]\\
&+\ln\dfrac{1+u_1(0)+u_2(0)}{1+u_1(1)+u_2(0)}-
\ln\dfrac{1+u_1(0)+\left[1-\operatorname{Q}_{a}\left( v' \mid 0\right)\right]u_2(0)}{1+u_1(1)+\left[1-\operatorname{Q}_{a}\left( v' \mid 0\right)\right]u_2(0)}=0
\end{align*}
for $v\in\{1,2\}$. Also, if $\operatorname{Nrc}_a$ is rich enough so that it contains $\mathbf{nr}=(0,0)$, $\mathbf{nc}=(0,1)$, and $\mathbf{nc}'=(1,1)$, then we can switch from $\operatorname{Q}_{a}\left( v' \mid 0\right)$ to $\operatorname{Q}_{a}\left( v' \mid 1\right)$ without changing other parameters. Hence, we should also have that
\begin{align*}
&\ln \dfrac{\operatorname{Q}_{a}\left( v \mid 1\right)}{\operatorname{Q}_{a}\left( v \mid 0\right)}
+\ln\left[\dfrac{u_1(1)}{1+u_1(1)}\dfrac{1+u_1(0)}{u_1(0)}\right]\\
&+\ln\dfrac{1+u_1(0)+u_2(0)}{1+u_1(1)+u_2(0)}-
\ln\dfrac{1+u_1(0)+\left[1-\operatorname{Q}_{a}\left( v' \mid 1\right)\right]u_2(0)}{1+u_1(1)+\left[1-\operatorname{Q}_{a}\left( v' \mid 1\right)\right]u_2(0)}=0.
\end{align*}
The last two equalities imply that
\[
\dfrac{1+u_1(0)+\left[1-\operatorname{Q}_{a}\left( v' \mid 1\right)\right]u_2(0)}{1+u_1(1)+\left[1-\operatorname{Q}_{a}\left( v' \mid 1\right)\right]u_2(0)}=\dfrac{1+u_1(0)+\left[1-\operatorname{Q}_{a}\left( v' \mid 0\right)\right]u_2(0)}{1+u_1(1)+\left[1-\operatorname{Q}_{a}\left( v' \mid 0\right)\right]u_2(0)}.
\]
The latter is possible if and only if $\operatorname{Q}_{a}\left( v' \mid 0\right)=\operatorname{Q}_{a}\left( v' \mid 1\right)$, which violates Assumption~\ref{ass: A1}(iii).

Note that Assumption~\ref{ass: A1}(iii) also establishes that
\[
\Delta_{0,v'}\Delta_{v,0}\ln\bar{\operatorname{P}}_a(v\mid \mathbf{0},\mathbf{0})\neq 0,
\]
which guarantees Assumption~\ref{ass: separability}(ii) for $i=2$. Assumption~\ref{ass: separability}(ii) is a bit more general. Specifically, violations of Assumption~\ref{ass: separability}(ii) for case $i=1$ mean that the following equations hold 
\begin{eqnarray*}
\Delta_{v',0}\Delta_{v,0}\ln\bar{\operatorname{P}}_a(v\mid \mathbf{nr},\mathbf{nc})= 0,\\
\Delta_{v,0}\Delta_{v',0}\ln\bar{\operatorname{P}}_a(v'\mid \mathbf{nr},\mathbf{nc})= 0,
\end{eqnarray*}
for all combinations of the choice configuration $\mathbf{nr}$ and $\mathbf{nc}$, including $\mathbf{nr}=\mathbf{0}$ and $\mathbf{nc}=\mathbf{0}$. The set of equalities increases with the size of the peer group and/or the set of alternatives, making violations harder to arise. A similar logic carries to case $i=3$. 

\end{example}

\subsection{Proof of Proposition~\ref{prop: identification of peers}}
\noindent Fix some $a\in\mathcal{A}$. We will prove that
\[
a'\not\in\mathcal{N}_a \:\iff\: \dfrac{\operatorname{P}_a\left(v\mid \mathbf{y}\right)}{\operatorname{P}_a\left(v\mid \mathbf{y}'\right)}=1 \text{ for all } v,\text{ and }\mathbf{y},\mathbf{y'} \text{ that are different in the $a'$th component only}.
\]
The ``only if'' part is straightforward. To show the ``if'' part, assume, towards a contradiction, that 
\[    
\dfrac{\operatorname{P}_a\left(v\mid \mathbf{y}\right)}{\operatorname{P}_a\left(v\mid \mathbf{y}'\right)}=1 \text{ for all } \mathbf{y},\mathbf{y'} \text{ that are different in the $a'$th component only},
\]
and $a'\in\mathcal{N}_a$. Let $\mathbf{y}^{v}_{z}$ denote the vector in which the $z$-th component of $\mathbf{y}$ is replaced by $v$. 

Note that the observed CCP can be expressed as
\[
{\tiny\operatorname{P}_{a}\left( v \mid \mathbf{y}\right) =
\operatorname{Q}_{a}\left(v \mid \operatorname{NC}_{a}^{v}\left(\mathbf{y}\right)\right)\times \sum\nolimits_{ 
\mathcal{C}\subseteq\mathcal{Y}\setminus\{v\}}\operatorname{R}_{a}\left( v \mid \operatorname{NR}^{\mathcal{C}\cup\{v\}}_{a}\left(\mathbf{y}\right),\mathcal{C}\cup\{v\}\right)\operatorname{C}_a\left(\mathcal{C}\mid \operatorname{NC}_a^{\mathcal{Y}\setminus\{v\}}(\mathbf{y}),\mathcal{Y}\setminus\{v\}\right),}
\]
where the first component only depends on the number of consideration peers selecting alternative $v$, and the second component depends on the whole vector of the number of preference peers' choices. 

If $a'\in \mathcal{NC}_a\setminus\mathcal{NR}_a$, then 
\[
\dfrac{\operatorname{P}_a\left(v\mid \mathbf{0}^{v}_{a'}\right)}{\operatorname{P}_a\left(v\mid \mathbf{0}\right)}=\dfrac{\operatorname{Q}_a\left(v\mid 1\right)}{\operatorname{Q}_a\left(v\mid 0\right)}\neq 1,
\]
where the first equality holds by Assumption~\ref{ass: A1}(ii) and the fact that $\operatorname{NC}_{a}^{v}\left(\mathbf{0}\right)=0$ and $\operatorname{NC}_{a}^{v}\left(\mathbf{0}^{v}_{a'}\right)=1$. (It also follows as the change of Agent $a'$'s choice does not affect Agent $a$'s preference towards any alternative.) The last inequality follows from Assumption~\ref{ass: A1}(iii). 

Similarly, if $a'\in \mathcal{NR}_a\setminus\mathcal{NC}_a$, then
\begin{eqnarray*}
\dfrac{\operatorname{P}_a\left(v\mid \mathbf{0}^{v}_{a'}\right)}{\operatorname{P}_a\left(v\mid \mathbf{0}\right)} &=&\dfrac{\sum\nolimits_{ 
\mathcal{C}\subseteq\mathcal{Y}\setminus\{v\}}\operatorname{R}_{a}\left( v \mid \operatorname{NR}^{\mathcal{C}\cup\{v\}}_{a}\left( \mathbf{0}^{v}_{a'}\right),\mathcal{C}\cup\{v\}\right)\operatorname{C}_a\left(\mathcal{C}\mid \operatorname{NC}_a^{\mathcal{Y}\setminus\{v\}}( \mathbf{0}^{v}_{a'}),\mathcal{Y}\setminus\{v\}\right)} {\sum\nolimits_{ 
\mathcal{C}\subseteq\mathcal{Y}\setminus\{v\}}\operatorname{R}_{a}\left( v \mid \operatorname{NR}^{\mathcal{C}\cup\{v\}}_{a}\left(\mathbf{0}\right),\mathcal{C}\cup\{v\}\right)\operatorname{C}_a\left(\mathcal{C}\mid \operatorname{NC}_a^{\mathcal{Y}\setminus\{v\}}(\mathbf{0}),\mathcal{Y}\setminus\{v\}\right)} \\
&=& \dfrac{\sum\nolimits_{ 
\mathcal{C}\subseteq\mathcal{Y}\setminus\{v\}}\operatorname{R}_{a}\left( v \mid \mathbf{0}_{v}^1,\mathcal{C}\cup\{v\}\right)\operatorname{C}_a\left(\mathcal{C}\mid \mathbf{0},\mathcal{Y}\setminus\{v\}\right)} {\sum\nolimits_{ 
\mathcal{C}\subseteq\mathcal{Y}\setminus\{v\}}\operatorname{R}_{a}\left( v \mid \mathbf{0},\mathcal{C}\cup\{v\}\right)\operatorname{C}_a\left(\mathcal{C}\mid \mathbf{0},\mathcal{Y}\setminus\{v\}\right)}\\
&=& \dfrac{\sum\nolimits_{ 
\mathcal{C}\subseteq\mathcal{Y}\setminus\{v\}}\Big[\operatorname{R}_{a}\left( v \mid \mathbf{0}_{v}^1,\mathcal{C}\cup\{v\}\right)-\operatorname{R}_{a}\left( v \mid \mathbf{0},\mathcal{C}\cup\{v\}\right)\Big]\operatorname{C}_a\left(\mathcal{C}\mid \mathbf{0},\mathcal{Y}\setminus\{v\}\right)} {\sum\nolimits_{ 
\mathcal{C}\subseteq\mathcal{Y}\setminus\{v\}}\operatorname{R}_{a}\left( v \mid \mathbf{0},\mathcal{C}\cup\{v\}\right)\operatorname{C}_a\left(\mathcal{C}\mid \mathbf{0},\mathcal{Y}\setminus\{v\}\right)}+1\\
&\neq& 1,
\end{eqnarray*}
where the first equality holds because the probability of considering $v$ does not change when switching Agent $a'$'s choice from $0$ to $v$. The second equality holds by Assumption~\ref{ass: A1}(ii) and~\ref{ass: A2}(ii). The last inequality holds since  
\[
\sum\nolimits_{ 
\mathcal{C}\subseteq\mathcal{Y}\setminus\{v\}}\Big[\operatorname{R}_{a}\left( v \mid \mathbf{0}_{v}^1,\mathcal{C}\cup\{v\}\right)-\operatorname{R}_{a}\left( v \mid \mathbf{0},\mathcal{C}\cup\{v\}\right)\Big]\operatorname{C}_a\left(\mathcal{C}\mid \mathbf{0},\mathcal{Y}\setminus\{v\}\right)\neq 0
\]
by Assumption~\ref{ass: A2}(iii).
Hence, the only remaining possibility is $a'\in\mathcal{NCR}_a$. But the latter contradicts Assumption~\ref{ass: separability}(i), since $a'\in\mathcal{NCR}_a$ would imply that the consideration peer effect offsets the preference peer effect \emph{everywhere} over the support. The contradiction completes the proof. 

\subsection{Proof of Proposition~\ref{prop: identification of pref group}}
\noindent Note that $\mathcal{N}_a$ is identified by Proposition~\ref{prop: identification of peers}. Take any two distinct agents $a',a''\in \mathcal{N}_a$. We will show that $a'\in\mathcal{NC}_a\setminus\mathcal{NR}_a$ if and only if
\begin{equation}\label{eq: proof NC-NR 1}
\dfrac{\operatorname{P}_a\left(v\mid \mathbf{y}^{w}_{a''}\right)}{\operatorname{P}_a\left(v\mid \mathbf{y}\right)}=\dfrac{\operatorname{P}_a\left(v\mid \left(\mathbf{y}^{v}_{a'}\right)_{a''}^{w}\right)}{\operatorname{P}_a\left(v\mid \mathbf{y}^{v}_{a'}\right)},
\end{equation}
for all $v \in \mathcal{Y}\setminus\{0\}$,  all $w \not\in\{0,v\}$, and all $\mathbf{y}$ with $y_{a'}=y_{a''}=0$. Thus, $\mathcal{NC}_a\setminus\mathcal{NR}_a$ is identified from $\operatorname{P}_a$.

To prove the ``only if'' part note that if $a'\in\mathcal{NC}_a\setminus\mathcal{NR}_a$ and $\mathbf{y}$ is such that $y_{a'}=y_{a''}=0$, then
\[
\dfrac{\operatorname{Q}_a\left(v\mid\operatorname{NC}^v_{a}\left(\mathbf{y}^{w}_{a''}\right) + 1\right)}{\operatorname{Q}_a\left(v\mid\operatorname{NC}^v_{a}\left(\mathbf{y}\right) + 1\right)}=\dfrac{\operatorname{Q}_a\left(v\mid\operatorname{NC}^v_{a}\left(\mathbf{y}\right) + 1\right)}{\operatorname{Q}_a\left(v\mid\operatorname{NC}^v_{a}\left(\mathbf{y}\right) + 1\right)}=1=\dfrac{\operatorname{Q}_a\left(v\mid\operatorname{NC}^v_{a}\left(\mathbf{y}\right)\right)}{\operatorname{Q}_a\left(v\mid\operatorname{NC}^v_{a}\left(\mathbf{y}\right)\right)}=\dfrac{\operatorname{Q}_a\left(v\mid\operatorname{NC}^v_{a}\left(\mathbf{y}^{w}_{a''}\right)\right)}{\operatorname{Q}_a\left(v\mid\operatorname{NC}^v_{a}\left(\mathbf{y}\right)\right)},
\]
where the first and the last equalities follow from the fact that $w\neq v$. Hence, since $\left(\mathbf{y}^{v}_{a'}\right)_{a''}^{w}=\left(\mathbf{y}^{w}_{a''}\right)_{a'}^{v}$ we have that 
\begin{align*}
\dfrac{\operatorname{P}_a\left(v\mid \left(\mathbf{y}^{v}_{a'}\right)_{a''}^{w}\right)}{\operatorname{P}_a\left(v\mid \mathbf{y}^{v}_{a'}\right)} &= \dfrac{\operatorname{P}_a\left(v\mid \left(\mathbf{y}^{w}_{a''}\right)_{a'}^{v}\right)}{\operatorname{P}_a\left(v\mid \mathbf{y}^{v}_{a'}\right)}=\dfrac{\operatorname{Q}_a\left(v\mid\operatorname{NC}^v_{a}\left(\mathbf{y}^{w}_{a''}\right) + 1\right)}{\operatorname{Q}_a\left(v\mid\operatorname{NC}^v_{a}\left(\mathbf{y}\right) + 1\right)}\\
&\times \dfrac{\sum\nolimits_{ 
\mathcal{C}\subseteq\mathcal{Y}\setminus\{v\}}\operatorname{R}_{a}\left( v \mid \operatorname{NR}^{\mathcal{C}\cup\{v\}}_{a}\left(\mathbf{y}^{w}_{a''}\right),\mathcal{C}\cup\{v\}\right)\operatorname{C}_a\left(\mathcal{C}\mid \operatorname{NC}_a^{\mathcal{Y}\setminus\{v\}}(\mathbf{y}^{w}_{a''}),\mathcal{Y}\setminus\{v\}\right)}{\sum\nolimits_{ 
\mathcal{C}\subseteq\mathcal{Y}\setminus\{v\}}\operatorname{R}_{a}\left( v \mid \operatorname{NR}^{\mathcal{C}\cup\{v\}}_{a}\left(\mathbf{y}\right),\mathcal{C}\cup\{v\}\right)\operatorname{C}_a\left(\mathcal{C}\mid \operatorname{NC}_a^{\mathcal{Y}\setminus\{v\}}(\mathbf{y}),\mathcal{Y}\setminus\{v\}\right)} \\
&=\dfrac{\operatorname{Q}_a\left(v\mid\operatorname{NC}^v_{a}\left(\mathbf{y}^{w}_{a''}\right) \right)}{\operatorname{Q}_a\left(v\mid\operatorname{NC}^v_{a}\left(\mathbf{y}\right)\right)}\\
&\times \dfrac{\sum\nolimits_{ 
\mathcal{C}\subseteq\mathcal{Y}\setminus\{v\}}\operatorname{R}_{a}\left( v \mid \operatorname{NR}^{\mathcal{C}\cup\{v\}}_{a}\left(\mathbf{y}^{w}_{a''}\right),\mathcal{C}\cup\{v\}\right)\operatorname{C}_a\left(\mathcal{C}\mid \operatorname{NC}_a^{\mathcal{Y}\setminus\{v\}}(\mathbf{y}^{w}_{a''}),\mathcal{Y}\setminus\{v\}\right)}{\sum\nolimits_{ 
\mathcal{C}\subseteq\mathcal{Y}\setminus\{v\}}\operatorname{R}_{a}\left( v \mid \operatorname{NR}^{\mathcal{C}\cup\{v\}}_{a}\left(\mathbf{y}\right),\mathcal{C}\cup\{v\}\right)\operatorname{C}_a\left(\mathcal{C}\mid \operatorname{NC}_a^{\mathcal{Y}\setminus\{v\}}(\mathbf{y}),\mathcal{Y}\setminus\{v\}\right)} \\
&= \dfrac{\operatorname{P}_a\left(v\mid \mathbf{y}^{w}_{a''}\right)}{\operatorname{P}_a\left(v\mid \mathbf{y}\right)}.
\end{align*}

To prove the ``if'' part, note that it is equivalent to the statement that if $a'\in\mathcal{NR}_a$, then there exist $a'' \in \mathcal{N}_a$, $v$, $w$, and $\mathbf{y}$ with $y_{a'}\neq v$ and $y_{a''}\not\in\{v,w\}$ such that 
\[
\dfrac{\operatorname{P}_a\left(v\mid \mathbf{y}^{w}_{a''}\right)}{\operatorname{P}_a\left(v\mid \mathbf{y}\right)}\neq\dfrac{\operatorname{P}_a\left(v\mid \left(\mathbf{y}^{v}_{a'}\right)_{a''}^{w}\right)}{\operatorname{P}_a\left(v\mid \mathbf{y}^{v}_{a'}\right)}.
\]
or equivalently
\[
\Delta_{a'}^{v}\Delta_{a''}^{w}\ln \operatorname{P}_a\left(v\mid \mathbf{y}\right)\neq 0.
\]
If $a''\in\mathcal{NR}_a\setminus\mathcal{NC}_a$, then let $i=1$. If $a''\in\mathcal{NC}_a\setminus\mathcal{NR}_a$, then let $i=2$. Finally, if $a''\in\mathcal{NCR}_a$, then let $i=3$. Take $v=v_i$, $w=w_i$, and $\mathbf{y}$ such that $y_{a'}=y_{a''}=0$, $\operatorname{NR}^{\mathcal{Y}}_{a}(\mathbf{y})=\mathbf{nr}_i$, and $\operatorname{NC}^{\mathcal{Y}}_{a}(\mathbf{y})=\mathbf{nc}_i$ from Assumption~\ref{ass: separability}(ii). Then
\[
\Delta_{a'}^{v}\Delta_{a''}^{w}\ln \operatorname{P}_a\left(v\mid \mathbf{y}\right)=\Delta_{a''}^{w}\Delta_{a'}^{v}\ln \operatorname{P}_a\left(v\mid \mathbf{y}\right)=\begin{cases}
    \Delta_{w_1,0}\Delta_{v_1,0}\ln\bar{\operatorname{P}}_a(v_1\mid \mathbf{nr}_1,\mathbf{nc}_1)\neq 0 \text{ if } i=1,\\
    \Delta_{0,w_2}\Delta_{v_2,0}\ln\bar{\operatorname{P}}_a(v_2\mid \mathbf{nr}_2,\mathbf{nc}_2)\neq 0 \text{ if } i=2,\\
    \Delta_{w_3,w_3}\Delta_{v_3,0}\ln\bar{\operatorname{P}}_a(v_3\mid \mathbf{nr}_3,\mathbf{nc}_3)\neq 0 \text{ if } i=3,
\end{cases}
\]
where the first equality follows from the exchangeability of the difference operator, the second equality follows from the definition of $\bar{\operatorname{P}}_a$, and the last inequality follows from Assumption~\ref{ass: separability}(ii). So in all possible cases, Assumption~\ref{ass: separability}(ii) implies that if $a'\in\mathcal{NR}_a$, then there exist $v$, $w$, and $\mathbf{y}$ with $y_{a'}=y_{a''}=0$ such that 
\[
\dfrac{\operatorname{P}_a\left(v\mid \mathbf{y}^{w}_{a''}\right)}{\operatorname{P}_a\left(v\mid \mathbf{y}\right)}\neq\dfrac{\operatorname{P}_a\left(v\mid \left(\mathbf{y}^{v}_{a'}\right)_{a''}^{w}\right)}{\operatorname{P}_a\left(v\mid \mathbf{y}^{v}_{a'}\right)}.
\]

\subsection{Proof of Proposition~\ref{prop: identification of groups}}

\noindent Note that we know $\mathcal{N}_a$ and $\mathcal{NR}_a$ (or $\mathcal{NC}_a\setminus\mathcal{NR}_a$). To identify the rest of the network structure ($\mathcal{NR}_a\setminus\mathcal{NC}_a$ and $\mathcal{NCR}_a$), suppose that $\mathcal{NC}_a\setminus\mathcal{NR}_a\neq\emptyset$. Take $a'\in\mathcal{NC}_a\setminus\mathcal{NR}_a$. First, note that
\[
\Delta_{a'}^v \ln \operatorname{P}_a(v\mid \mathbf{0})
= \ln\operatorname{Q}_a(v\mid 1) - \ln\operatorname{Q}_a(v\mid 0).
\]

Thus, for any $a''\in\mathcal{NR}_a$, by Assumption~\ref{ass: A1},
\[
\Delta_{a''}^v\Delta_{a'}^v\ln \operatorname{P}_a(v\mid\mathbf{0})\neq0\:\iff\: a''\in\mathcal{NCR}_a.
\]
Hence, $\mathcal{NCR}_a$ is identified from $\operatorname{P}_a$. 

Next, suppose that $\mathcal{NC}_a\setminus\mathcal{NR}_a=\emptyset$. Then, by Assumption~\ref{ass: A3}, either $\mathcal{N}_a=\mathcal{NR}_a\setminus\mathcal{NC}_a$ or both $\mathcal{NR}_a\setminus\mathcal{NC}_a$ and $\mathcal{NCR}_a$ are nonempty. Since the consideration effect is nonzero, the effects of preference-only peers and consideration-preference peers have to be different. As a result, we can identify the partition of $\mathcal{NR}_a$, $\mathcal{M}'$ and $\mathcal{M}''$, such that one of its elements is $\mathcal{NCR}_a$. Since $\abs{\mathcal{N}_a}\geq 3-\abs{\mathcal{NC}_a\setminus\mathcal{NR}_a}=3$, we can take $a'\in\mathcal{M}'$ and $a''\in\mathcal{M}''$. Next, take $\mathbf{y}$ such that $y_a=0$ for all $a\neq a'$ and $y_{a'}=v$. Next note that
\begin{align*}
\ln{\operatorname{P}_{a}\left( v \mid \mathbf{y}\right)} - \ln{\operatorname{P}_{a}\left( v \mid \left(\mathbf{y}^{0}_{a'}\right)^{v}_{a''}\right)}=(-1)^{\Char{a'\not\in\mathcal{NCR}_a}} (\ln{\operatorname{Q}_{a}\left(v\mid  1\right)} - \ln{\operatorname{Q}_{a}\left(v\mid  0\right)}).
\end{align*}
Finally, take another $a'''\not\in\{a',a''\}$ in either $\mathcal{M}'$ or $\mathcal{M}''$. Without loss of generality, assume that $a'''\in\mathcal{M}'$. Note that, by Assumption~\ref{ass: A1},
\begin{align*}
\Delta_{a'''}^v\ln{\operatorname{P}_{a}\left( v \mid \mathbf{y}\right)} - \Delta_{a'''}^v\ln{\operatorname{P}_{a}\left( v \mid \left(\mathbf{y}^{0}_{a'}\right)^{v}_{a''}\right)}=0 \iff a'''\in \mathcal{NR}_a\setminus\mathcal{NC}_a.
\end{align*}
Thus, we identify $\mathcal{NR}_a\setminus\mathcal{NC}_a$ and $\mathcal{NCR}_a$.

\subsection{Proof of Proposition~\ref{prop: identification of ratios of Q}}
\noindent Fix $a\in\mathcal{A}$ and $v\in\mathcal{Y}\setminus\{0\}$. Assume first that $\abs{\mathcal{NC}_a\setminus\mathcal{NR}_a}\geq 1$. Under this situation, the relative consideration probability is identified via switching the choice of just one consideration-only peer from alternative $v$ to the default while keeping the configuration of others fixed. Specifically, take $a'\in\mathcal{NC}_a\setminus\mathcal{NR}_a$ and $\mathbf{y}$ such that every peer in $\mathcal{NC}_a$ picks $v$. Then
\[
    \dfrac{\operatorname{P}_a(v|\mathbf{y})}{\operatorname{P}_a(v|\mathbf{y}^{0}_{a'})}=\dfrac{\operatorname{Q}_a(v|\abs{\mathcal{NC}_a})}{\operatorname{Q}_a(v|\abs{\mathcal{NC}_a}-1)}.
\]
Next, redefine $\mathbf{y}$ as before except that we let one of the peers from $\mathcal{NCR}_a$ to pick $0$. As a result,
    \[
    \dfrac{\operatorname{P}_a(v|\mathbf{y})}{\operatorname{P}_a(v|\mathbf{y}^{0}_{a'})}=\dfrac{\operatorname{Q}_a(v|\abs{\mathcal{NC}_a}-1)}{\operatorname{Q}_a(v|\abs{\mathcal{NC}_a}-2)}.
    \]
Repeating this procedure, we identify 
\[
\operatorname{Q}_a(v\mid n_1)/\operatorname{Q}_a(v\mid n_1-1) \text{ for all } n_1\in\{\abs{\mathcal{NC}_a}-\abs{\mathcal{NCR}_a},\dots,\abs{\mathcal{NC}_a}\}.
\]

Next, we take $\mathbf{y}$ such that all peers in $\mathcal{NCR}_a$ and one of the peers in $\mathcal{NC}_a\setminus\mathcal{NR}_a$ different from $a'$ are picking $0$ and the rest of peers in $\mathcal{NC}_a\setminus\mathcal{NR}_a$ are picking $v$. Switching one by one all peers in $\mathcal{NC}_a\setminus\mathcal{NR}_a$ we identify $\operatorname{Q}_a(v\mid n_1)/\operatorname{Q}_a(v\mid n_1-1)$ for all $n_1$.

We next show that the relative consideration probability can be identified even if the consideration-only group is empty. Specifically, assume that $\abs{\mathcal{NC}_a\setminus\mathcal{NR}_a}=0$, so we have $\abs{\mathcal{NR}_a\setminus\mathcal{NC}_a}\geq 1$ by Assumption \ref{ass: A3}. Then the relative consideration probability can be identified by switching one preference-only peer from $v$ to the default and one consideration-preference peer from the default to alternative $v$. Specifically, take $a'\in\mathcal{NR}_a\setminus\mathcal{NC}_a$, $a''\in\mathcal{NCR}_a$, and $\mathbf{y}$ such that every peer in $\mathcal{NCR}_a$ picks $v$ and $a'$ picks $0$. Then, comparing Agent $a$'s probability of choosing alternative $v$ between configuration $\mathbf{y}$ and a configuration of switching Agent $a'$ from 0 to alternative $v$ and Agent $a''$ from alternative $v$ to $0$, which does not change the choice probability given consideration because the number of peers affecting preference is the same in both scenario, we have
\[
    \dfrac{\operatorname{P}_a(v|\mathbf{y})}{\operatorname{P}_a\left(v|\left(\mathbf{y}^{v}_{a'}\right)_{a''}^{0}\right)}=\dfrac{\operatorname{Q}_a(v|\abs{\mathcal{NC}_a})}{\operatorname{Q}_a(v|\abs{\mathcal{NC}_a}-1)}.
\]
Next, redefine $\mathbf{y}$ as before except that we let one of the peers from $\mathcal{NCR}_a$ different from $a''$ to pick $0$. As a result,
    \[
    \dfrac{\operatorname{P}_a(v|\mathbf{y})}{\operatorname{P}_a\left(v|\left(\mathbf{y}^{v}_{a'}\right)_{a''}^{0}\right)}=\dfrac{\operatorname{Q}_a(v|\abs{\mathcal{NC}_a}-1)}{\operatorname{Q}_a(v|\abs{\mathcal{NC}_a}-2)}.
    \]
Repeating this procedure finitely many times we identify $\operatorname{Q}_a(v\mid n_1)/\operatorname{Q}_a(v\mid n_1-1)$ for all $n_1\in\{1,\dots,\abs{\mathcal{NC}_a}\}$. 

\subsection{Proof of Proposition~\ref{prop: counterfactual CCP}}
\noindent Fix some $a\in\mathcal{A}$ and $a'\in\mathcal{NC}_a\setminus\mathcal{NR}_a$. Moreover, take any distinct $v,v'\in\mathcal{Y}\setminus\{0\}$. Take any $\mathbf{y}$ such that no one picks $v'$. Since we will only use the variation in choices of Agent $a'$, we drop the choices of everyone else from the notation. For example, $\operatorname{P}_a(v|v')$ is equal to $\operatorname{P}_a(v|\mathbf{y})$, where $y_{a'}=v'$. We use $t_{v'}$ to denote the ratio between the probability that Agent $a$ picks $v'$ conditional on Agent $a'$ choosing $v'$ and the default $0$:
\[
t_{v'}\equiv \dfrac{\operatorname{P}_a(v'|v')}{\operatorname{P}_a(v'|0)}=\dfrac{\operatorname{Q}_a(v'|1)}{\operatorname{Q}_a(v'|0)}\neq 1,
\]
where the second equality holds because we can cancel out the choice probability conditional on considering $v'$, and the inequality follows by Assumption \ref{ass: A1}(iii). Note that $t_{v'}$ is identified from the data. 

Moreover,
\begin{align*}
\operatorname{P}_a(v|0)&=\operatorname{Q}_a(v'|0)\left\{\operatorname{R}_a^*(v|v')-\operatorname{P}_a^*\left(v\mid\mathcal{Y}\setminus\{v'\}\right) \right\}+\operatorname{P}_a^*(v\mid\mathcal{Y}\setminus\{v'\}),\\
\operatorname{P}_a(v|v')&=\operatorname{Q}_a(v'|1)\left\{\operatorname{R}_a^*(v|v')-\operatorname{P}_a^*(v\mid\mathcal{Y}\setminus\{v'\}) \right\}+\operatorname{P}_a^*(v\mid\mathcal{Y}\setminus\{v'\}),
\end{align*}
where $\operatorname{R}_a^*(v|v')$ denotes the probabilities that Agent $a$ picks $v$ conditional on considering $v'$. Since, $\operatorname{Q}_a(v'|0)t_{v'}=\operatorname{Q}_a(v'|1)$, we obtain from the above two equations that 
\begin{align*}
\operatorname{P}_a^*(v\mid\mathcal{Y}\setminus\{v'\})&=\dfrac{\operatorname{P}_a(v|v')-t_{v'}\operatorname{P}_a(v|0)}{1-t_{v'}}.
\end{align*}
Since the choice of $v$, $v'$, $a$, $a'$, and choices of everyone else was arbitrary, we can identify $\operatorname{P}_a^*(v\mid\mathbf{y},\mathcal{Y}\setminus\{v'\})$ for all $a\in\mathcal{A}$, $v'\neq v$, $v'\neq 0$, and $\mathbf{y}$ such that (i) $y_{a'}\neq v'$ for all $a'\in\mathcal{N}_a$ and (ii) $y_{a'}=0$ for some $a'\in\mathcal{NC}_a\setminus\mathcal{NR}_a$. 

Applying the above argument to $\operatorname{P}_a^*(\cdot|\cdot,\mathcal{Y}\setminus\{v'\})$, we can identify $\operatorname{P}_a^*(v\mid\mathbf{y},\mathcal{Y}\setminus\{v',v''\})$ for all $a\in\mathcal{A}$, $v''\neq v$, $v'' \neq v'$, $v''\neq 0$, and $\mathbf{y}$ such that (i) $y_{a'}\not\in\{v',v''\}$ for all $a'\in\mathcal{N}_a$ and (ii) $y_{a'}=y_{a''}=0$ for some $a',a''\in\mathcal{NC}_a\setminus\mathcal{NR}_a$, $a'\neq a''$. 

Repeating the above argument $\abs{\mathcal{NC}_a\setminus\mathcal{NR}_a}$ times, we can identify $\operatorname{P}_a^*(\cdot\mid\mathbf{y},\mathcal{Y}\setminus\mathcal{Z})$ for all $\mathcal{Z}\subseteq \mathcal{Y}\setminus\{0\}$ and $\mathbf{y}$ with the following two properties.  First, $y_{a'}\not\in\mathcal{Z}$ for all $a'\in\mathcal{N}_a$. Second, if we take any different $\abs{\mathcal{Z}}$ components of $\mathbf{y}$ that correspond to peers from $\mathcal{NC}_a\setminus\mathcal{NR}_a$, then these components have to be equal to $0$ since we switched these $\abs{\mathcal{Z}}$ peers to $0$.  

\subsection{Proof of Proposition~\ref{prop: identification of all}}
\noindent Fix some $v\neq 0$. If $\operatorname{Q}_a\left(v\mid n_1\right)$ is known for some $n_1$ in the support, by Proposition~\ref{prop: identification of ratios of Q}, we identify $\operatorname{Q}_a(v\mid\cdot)$. If, instead, we know $\operatorname{R}_a\left(v\mid n_2,\{0,v\}\right)$, then, since $\abs{\mathcal{NC}_a\setminus\mathcal{NR}_a}\geq Y$, by Proposition~\ref{prop: counterfactual CCP}, we identify 
\[
\operatorname{P}^{*}_a(v \mid \mathbf{y}, \{0,v\})=\operatorname{Q}_a\left(v\mid\mathcal{NC}^v_{a}\left(\mathbf{y}\right)\right)\operatorname{R}_a\left(v\mid \operatorname{NR}_a^v\left(\mathbf{y}\right),\{0,v\}\right)
\]
for some $\mathbf{y}$ such that $\operatorname{NR}_a^v\left(\mathbf{y}\right)=n_2$. Hence, we identify $\operatorname{Q}_a\left(v\mid\mathcal{NC}^v_{a}\left(\mathbf{y}\right)\right)$ and, by Proposition~\ref{prop: identification of ratios of Q}, we also identify $\operatorname{Q}_a(v\mid\cdot)$. Since, the choice of $v$ was arbitrary, we identify $\operatorname{Q}_a$.

By Proposition~\ref{prop: counterfactual CCP}, we now can identify 
$\operatorname{R}_a\left(v\mid n_2,\{0,v\}\right)$ for all $v\neq0$ and $n_2$ in the support. Next, consider
\begin{align*}
\operatorname{P}^{*}_a(v \mid \mathbf{y}, \{0,v,v'\})=&\operatorname{Q}_a\left(v\mid\mathcal{NC}^v_{a}\left(\mathbf{y}\right)\right)\operatorname{Q}_a\left(v'\mid\mathcal{NC}^{v'}_{a}\left(\mathbf{y}\right)\right)\operatorname{R}_a\left(v\mid \operatorname{NR}_a^v\left(\mathbf{y}\right),\{0,v\}\right)+\\
+&\operatorname{Q}_a\left(v\mid\mathcal{NC}^v_{a}\left(\mathbf{y}\right)\right)(1-\operatorname{Q}_a\left(v'\mid\mathcal{NC}^{v'}_{a}\left(\mathbf{y}\right)\right))\operatorname{R}_a\left(v\mid \operatorname{NR}_a^v\left(\mathbf{y}\right),\operatorname{NR}_a^{v'}\left(\mathbf{y}\right),\{0,v,v'\}\right).
\end{align*}
Since $\operatorname{Q}_a$ and $\operatorname{R}_a$ for binary consideration sets are identified, we identify $\operatorname{R}_a$ for all possible sets of size $3$. Repeating the above argument, we identify $\operatorname{R}_a$ for all possible sets of size $4$. Applying this argument finitely many times, we can identify $\operatorname{R}_a$ for all possible sets.

\subsection{Proof of Proposition~\ref{ID2}}
\noindent Since $\lim\nolimits_{\Delta \rightarrow 0}\mathcal{P}\left( \Delta \right) =\mathcal{W}$, we can recover the transition rate matrix $\mathcal{W}$ from the data. Recall that each element in the transition rate matrix is defined as
\begin{equation*}
\operatorname{w}\left( \mathbf{y}'\mid \mathbf{y}\right) =\left\{ 
\begin{array}{lcc}
0 & \text{if} & \sum\nolimits_{a\in \mathcal{A}}\mathds{1}\left( y_{a}'\neq
y_{a}\right) >1 \\ 
\sum\nolimits_{a\in \mathcal{A}}\lambda_{a}\operatorname{P}_{a}\left( y_{a}' \mid \mathbf{y}\right) \mathds{1}\left( y_{a}'\neq y_{a}\right) & 
\text{if} & \sum\nolimits_{a\in \mathcal{A}}\mathds{1}\left( y_{a}'\neq
y_{a}\right) =1
\end{array}
\right. .
\end{equation*}
Thus, $\lambda_{a}\operatorname{P}_{a}\left( y_{a}' \mid \mathbf{y}\right) =\operatorname{w}\left( y_{a}',\mathbf{y}_{-a}\mid \mathbf{y}\right)$. It follows that we can recover $\lambda_{a}\operatorname{P}_{a}\left( v \mid \mathbf{y}\right)$ for each $v\in \mathcal{Y}$, $\mathbf{y}\in \mathcal{Y}^{A}$, and $a\in \mathcal{A}$. Note that, for each $\mathbf{y}\in \mathcal{Y}^{A}$, 
\begin{equation*}
\sum\nolimits_{v\in \mathcal{Y}}\lambda_{a}\operatorname{P}_{a}\left( v \mid \mathbf{y}\right) =\lambda_{a}\sum\nolimits_{v\in \mathcal{Y}}\operatorname{P}_{a}\left( v \mid \mathbf{y}\right) =\lambda_{a}.
\end{equation*}
Then we can also recover $\lambda_{a}$ for each $a\in \mathcal{A}$.

\subsection{Proof of Proposition~\ref{ID3}}
\noindent This proof builds on Theorem 1 of \citet{blevins2017identifying} and Theorem 3 of \citet{blevins2018identification}. For the present case, it follows from these two theorems, that the transition rate matrix  $\mathcal{W}$ is generically identified if, in addition to the conditions in Proposition~\ref{ID3}, we have that 
\begin{equation*}
\left( Y+1\right) ^{A}-AY-1\geq \frac{1}{2}.
\end{equation*}
This condition is always satisfied if $A>1$. Thus, the identification of $\mathcal{W}$ follows because $A\geq 2$. We can then uniquely recover $\left(\operatorname{P}_{a}\right)_{a\in \mathcal{A}}$ and $\left(\lambda_{a}\right)_{a\in \mathcal{A}}$ from $\mathcal{W}$ as in the proof of Proposition~\ref{ID2}.

\section{Simulation and Estimation of the Running Example}\label{app: simul}
\noindent This section offers Monte Carlo simulation and estimation results for Example 1 in the paper. This exercise aims to show that with a rich dataset the primitives of the model can be estimated by following our main identification strategy step-by-step. Specifically, if we have enough observations, we can first reliably and nonparametrically estimate the CCPs for each agent using a frequency estimator. Then, we can use these estimates to recover which peers are in the consideration and preference group of each agent ---by following Propositions~\ref{prop: identification of peers}-~\ref{prop: identification of groups}. With the network structure being estimated for each agent separately, we can then estimate the choice probability and the consideration mechanism following Propositions~\ref{prop: identification of ratios of Q}-~\ref{prop: identification of all}. We show later that when the dataset is not long enough one could implement a parametric version of this approach.

\noindent\textbf{Simulation Design} Recall that there are four agents and three alternatives (i.e., $\mathcal{A}=\left\{ 1,2,3,4\right\}$ and $\mathcal{Y}=\left\{ 0,1,2\right\}$). The reference groups for consideration and preferences are as follows
\[
\mathcal{NC}_{1}=\left\{2, 3\right\},\quad \mathcal{NC}_{2}=\left\{1\right\},\quad \mathcal{NC}_{3}=\left\{2\right\}, \quad\mathcal{NC}_{4}=\emptyset
\]
\[
\mathcal{NR}_{1}=\left\{3\right\},\quad \mathcal{NR}_{2}=\emptyset,\quad \mathcal{NR}_{3}=\left\{1\right\}, \quad \mathcal{NR}_{4}=\emptyset.
\]

We specify the preferences of the agents and their consideration mechanisms as follows. For all $a$ and $v\in\{1,2\}$,
\[
\operatorname{Q}_a(v\mid n)=\begin{cases}
    1/4 &\text{ if } n=0\\
    3/4 &\text{ if } n=1\\
    1 &\text{ if } n=2.
\end{cases}
\]
As in the paper, the default $0$ is always considered.

The mean utility for all $a$, $v \in \{1,2\}$, and $\mathcal{C}$ is 
\[
\bar{u}_{a,v,\mathcal{C}}(n)=\begin{cases}
    3 &\text{ if } n=0\\
    9/2 &\text{ if } n=1\\
    5 &\text{ if } n=2.
\end{cases}
\]
The mean utility from the default is normalized to be $0$ regardless of how many peers choose the default. We assume the payoff shocks are generated by a Type I extreme value distribution, so the choice probability has the logit form.  Note that the agent's previous choice does not affect her consideration probabilities or preferences.

We can calculate the implied (population) CCPs of Agent 1 as a function of choices of Agents 2 and 3. To simplify the notation, we ignore the previous choice of Agent 4 in the choice configuration $\mathbf{y}$ as it does not affect the CCPs of Agent 1. For example, when the previous choices of Agents 2 and 3 are the default, the probability that Agent 1 selects option 1 is
\begin{eqnarray*}
\operatorname{P}_{1}\left(v=1 \mid (0,0)\right)&=& \underbrace{\operatorname{Q}_1(1\mid (0,0))(1-\operatorname{Q}_2(1\mid (0,0)))}_{\text{prob of considering} \{0,1\}} \underbrace{\frac{\exp{(3)}}{1+\exp{(3)}}}_{\text{prob of choosing 1 when considers }  \{0,1\}} \\
&+& \underbrace{ \operatorname{Q}_1(1\mid (0,0))\operatorname{Q}_2(1\mid (0,0))}_{\text{prob of considering} \{0,1,2\}}\underbrace{\frac{\exp{(3)}}{1+\exp{(3)}+\exp{(3)}}}_{\text{prob of choosing 1 when considers }  \{0,1,2\}} \\
&\approx&1/4 \times (1-1/4) \times 0.9526+1/4 \times 1/4 \times  0.4879=0.209.
\end{eqnarray*}
Similarly, we can calculate the CCPs for other configurations of $\mathbf{y}$ and alternative $2$. 

\noindent\textbf{Simulation Procedure} We simulate the data by the following procedure. First, we simulate four different Poisson alarms with an arrival rate of 1 and record the specific time at which the alarm goes off for each of the four agents. We start by assuming that all of them have selected the default. We then simulate the choices of each agent based on the order of the Poisson alarms. We start the simulation for $t=0$ and continue until the time reaches $T$ from which we can collect the profile of actions at different times that we indicate by $\{y_{1t}, y_{2t}, y_{3t}, y_{4t} \}_{t\le T}$. We assume that we can observe when the agents select each alternative, including the default ---Dataset 1 in the paper. 

\noindent\textbf{Estimation of CCPs} With the simulated data, we can first estimate the CCPs of each agent by using a simple frequency estimator. We use Agent 1 to illustrate the ideas. With a slight abuse of notation, we indicate by $t$ the time at which the alarm of Agent 1 is off and reserve $t'$, $t''$, and $t''''$ to denote the previous time at which the alarms of Agents 2, 3, and 4 were off, respectively. Then the estimator of $\operatorname{P}_1$ is
\begin{align*}
\hat{\operatorname{P}}_1(v|\mathbf{y}) = \frac{\#_{ \{t: t'<t, t''<t, t'''<t\}}\{ a_{1t}=v, (y_{2t'},y_{3t''},y_{4t'''}) = \mathbf{y} \}}{\#_{ \{t: t'<t, t''<t, t'''<t\}}\{(y_{2t'},y_{3t''},y_{4t'''}) = \mathbf{y}\}}.
\end{align*}

Table~\ref{tab1:estimates} displays the average of the estimated CCPs of Agent 1 for 1000 replications and $T=800$ (about 800 choices per agent). We also present the true (population) CCPs calculated by using the primitives of the model.

\begin{table}[h]
\centering
\caption{Agent 1's CCPs}
\begin{tabular}{c|ccccc}
\hline
 & & & \multicolumn{3}{c}{estimates for different $y_4$} \\
 & population & estimates & $y_4=0$ & $y_4=1$ & $y_4=2$ \\
\hline
$\operatorname{P}_1(1|00)$ & 0.209 & 0.208 & 0.207 & 0.207 & 0.212 \\
$\operatorname{P}_1(1|01)$ & 0.708 & 0.710 & 0.710 & 0.712 & 0.713 \\
$\operatorname{P}_1(1|02)$ & 0.093 & 0.095 & 0.094 & 0.096 & 0.095 \\
$\operatorname{P}_1(1|10)$ & 0.627 & 0.629 & 0.628 & 0.629 & 0.630 \\
$\operatorname{P}_1(1|11)$ & 0.944 & 0.944 & 0.944 & 0.945 & 0.943 \\
$\operatorname{P}_1(1|12)$ & 0.280 & 0.280 & 0.279 & 0.280 & 0.285 \\
$\operatorname{P}_1(1|20)$ & 0.151 & 0.151 & 0.153 & 0.151 & 0.148 \\
$\operatorname{P}_1(1|21)$ & 0.641 & 0.643 & 0.641 & 0.646 & 0.641 \\
$\operatorname{P}_1(1|22)$ & 0.045 & 0.045 & 0.046 & 0.044 & 0.045 \\
$\operatorname{P}_1(2|00)$ & 0.209 & 0.209 & 0.208 & 0.210 & 0.211 \\
$\operatorname{P}_1(2|01)$ & 0.093 & 0.094 & 0.094 & 0.099 & 0.087 \\
$\operatorname{P}_1(2|02)$ & 0.708 & 0.709 & 0.709 & 0.704 & 0.711 \\
$\operatorname{P}_1(2|10)$ & 0.151 & 0.149 & 0.152 & 0.145 & 0.146 \\
$\operatorname{P}_1(2|11)$ & 0.045 & 0.045 & 0.046 & 0.044 & 0.046 \\
$\operatorname{P}_1(2|12)$ & 0.641 & 0.642 & 0.642 & 0.648 & 0.636 \\
$\operatorname{P}_1(2|20)$ & 0.627 & 0.627 & 0.626 & 0.624 & 0.628 \\
$\operatorname{P}_1(2|21)$ & 0.280 & 0.279 & 0.282 & 0.279 & 0.278 \\
$\operatorname{P}_1(2|22)$ & 0.944 & 0.944 & 0.943 & 0.945 & 0.944 \\
\hline
\end{tabular}
\label{tab1:estimates}
\end{table}

The nonparametric estimator performs reasonably well for $T=800$ observations per agent. This simple framework has four agents and three alternatives. For each agent, the configuration of the choices of the other agents can take $3^3=27$ values. Thus, one needs to estimate $27$ probabilities per agent. This requirement increases exponentially with the number of agents in the model. 

\noindent\textbf{Estimation of the Network} We estimate the network using the nonparametric estimators of CCPs we just described. We first recover the reference group of Agent 1 by following the identification strategy in Proposition~\ref{prop: identification of peers}. Specifically, we compute the difference of the ln of the probability that Agent 1 selects alternative 1 when all other agents initially select the default and the ones we obtain when we change, one by one, each other agent to alternative 1: $\Delta_{a}^1\ln\operatorname{P}_1(1\mid\mathbf{0})$, $a=2,3,4$. 
 
Note that because we work with a finite sample, we cannot directly conclude that there is a link when $\Delta_{a}^1\ln\hat{\operatorname{P}}_1(1\mid\mathbf{0})\neq 0$. To decide whether there is a link or not, we implement a t-statistic test with a sample-size-driven critical value. Specifically, we say that $a$ is a peer of Agent 1 if and only if $\abs{\Delta_{a}^1\ln\hat{\operatorname{P}}_1(1\mid\mathbf{0})/\textrm{std}(\Delta_{a}^1\ln\hat{\operatorname{P}}_1(1\mid\mathbf{0}))}> \kappa_T$, where $\kappa_T=0.15\ln{T}$ and $\textrm{std}(A)$ is the estimated standard error of $A$.\footnote{We use the estimated CCPs to bootstrap standard errors \citep{kline2016bayesian}.} \footnote{Any sequence that converges to $0$ slower than $1/\sqrt{T}$ will work asymptotically.}  By following this criterion we calculate the frequency of correct estimates of connections for for each agent in 1000 replications. We present the results in Table~\ref{tab2:peer_estimates} for different sample sizes. The results improve with the sample size, and the estimation performs reasonably well when the sample size is about 2000. 

\begin{table}[h]
\centering
\caption{Percentage of correctly estimated peers of Agent 1}
\begin{tabular}{c|c|c|c|c|c|c|c}
\hline
 Agent & Peer &  T=800 & T=2000 & T=4000 & T=8000  & T=10000 & T=100000\\
\hline
2 & Yes & 100.0  & 100.0 & 100.0 & 100.0  & 100.0 & 100.0\\
3 & Yes & 100.0  & 100.0 & 100.0 & 100.0  & 100.0 & 100.0\\
4 & No & 66.2  & 75.1 & 79.5 & 81.3  & 82.8 & 91.6\\
\hline
\end{tabular}
\label{tab2:peer_estimates}
\end{table}

To state whether Agents 2 or 3 are consideration-only peers we implement double differences --- as we do in the proof of Proposition~\ref{prop: identification of pref group}. We use the population CCPs to illustrate the identification strategy in Proposition~\ref{prop: identification of pref group}. Specifically, to determine whether Agent 2 is a consideration-only peer of Agent 1, we check the following double difference:
\begin{align*}
\left[\ln{\operatorname{P}_1(1|12)}-\ln{\operatorname{P}_1(1|10)}\right]-\left[\ln{\operatorname{P}_1(1|02)}-\ln{\operatorname{P}_1(1|00)}\right]\\
=\ln{(0.280)}-\ln{(0.627)}- \ln{(0.093)}+\ln{(0.209)} = 0.
\end{align*}
Since the difference is $0$ we could conclude that Agent 2 is a consideration-only peer of Agent 1. We conduct the same analysis for Agent 3, 
\begin{align*}
\left[\ln{\operatorname{P}_1(1|21)}-\ln{\operatorname{P}_1(1|01)}\right]-\left[\ln{\operatorname{P}_1(1|20)}-\ln{\operatorname{P}_1(1|00)}\right]\\
=\ln{(0.641)}-\ln{(0.708)}- \ln{(0.151)}+\ln{(0.209)}=0.2256 \neq 0,
\end{align*}
and conclude that Agent 3 affects the preferences and maybe consideration of Agent 1. 

We apply the above double differences to the estimated CCPs with the same threshold rule to address whether the peer is a consideration-only peer or not. The results of the simulation are presented in Table~\ref{tab3:peer_estimates}. We can see that the percentage of correct estimates of the network increases with the sample size. 

\begin{table}[h]
\centering
\caption{Percentage of correctly estimated consideration-only peers of Agent 1}
\begin{tabular}{c|c|c|c|c|c|c|c}
\hline
 Agent & Consideration-only Peers &  T=800 & T=2000 & T=4000 & T=8000  & T=10000 & T=100000\\
\hline
2 & Yes & 67.1  & 72.3 & 78.7& 79.5  & 81.3 & 92.1 \\
3 & No & 37.4  & 49.4 & 56.9& 74.5  & 78.7 & 100.0 \\
\hline
\end{tabular}
\label{tab3:peer_estimates}
\end{table}

We finally identify whether Agent 3 affects both preferences and consideration. We follow the identification strategy in Proposition~\ref{prop: identification of groups}. Instead of switching Agent 2's choice from default to alternative 2, we check the changes in Agent 1's choice when switching the choice of Agent 2, the consideration-only peer, from the default to alternative 1 in the following double difference: 
\begin{align*}
\left[\ln{\operatorname{P}_1(1|11)}-\ln{\operatorname{P}_1(1|01)}\right]-\left[\ln{\operatorname{P}_1(1|10)}-\ln{\operatorname{P}_1(1|00)}\right]\\
=\ln{(0.944)}-\ln{(0.627)}- \ln{(0.708)}+\ln{(0.209)}=-0.8109 \neq 0.
\end{align*}
Since the result differs from 0, we can correctly conclude that Agent 3 is a consideration-preference peer of Agent 1. 

The results of applying the same logic to estimated CCPs (with the same threshold rule) are presented in Table~\ref{tab4:peer_estimates}. We can see that the percentage of correct network estimates increases with the sample size. Moreover, even with a sample size of 800, we can correctly estimate the link in 99.8\% of 1000 replicated samples.

\begin{table}[h]
\centering
\caption{Percentage of correctly estimated preference-only peers of Agent 1}
\begin{tabular}{c|c|c|c|c|c|c|c}
\hline
 Agent & Preference-only Peer &  T=800 & T=2000 & T=4000 & T=8000  & T=10000 & T=100000\\
\hline
3 & No & 99.8  & 100.0 & 100.0 & 100.0  & 100.0 & 100.0 \\
\hline
\end{tabular}
\label{tab4:peer_estimates}
\end{table}

\noindent\textbf{Estimation of Consideration Mechanism} Once the network structure is recovered, we can proceed to identify and estimate the consideration mechanism. First, given that Agent 2 is a consideration-only peer of Agent 1, we can switch Agent 2 first and then Agent 3 from default to alternative $1$ to identify ratios of consideration probabilities of alternative 1 for different numbers of peers selecting it ---following the ideas in the proof of Proposition~\ref{prop: identification of ratios of Q}. Specifically, we get that
\begin{align*}
&\dfrac{\operatorname{Q}_1(1|1)}{\operatorname{Q}_1(1|0)} = \frac{\operatorname{P}_1(1|10)}{\operatorname{P}_1(1|00)}=\frac{0.627}{0.209}=3= \underbrace{\frac{3/4}{1/4}}_{\text{the ratio from the model}}\\
&\dfrac{\operatorname{Q}_1(1|2)}{\operatorname{Q}_1(1|1)} = \frac{\operatorname{P}_1(1|11)}{\operatorname{P}_1(1|01)}=\frac{0.944}{0.708}= 1.3333 = \underbrace{\frac{1}{3/4}}_{\text{the ratio from the model}}.
\end{align*}
Assuming that when all consideration peers are picking the alternative, the alternative is considered with probability 1, i.e., $\operatorname{Q}_1(1|2)=1$, we can fully identify and estimate the rest of the consideration probabilities. In particular,
\begin{align*}
\operatorname{Q}_1(1|1)=&\dfrac{\operatorname{P}_1(1|01)}{\operatorname{P}_1(1|11)},\\
\operatorname{Q}_1(1|0)=&\dfrac{\operatorname{P}_1(1|01)}{\operatorname{P}_1(1|11)} \dfrac{\operatorname{P}_1(1|00)}{\operatorname{P}_1(1|10)}.
\end{align*}
We estimate those consideration probabilities and present the bias and the root mean squared error (RMSE) in Table~\ref{tab5:peer_estimates} for different sample sizes with 1000 replications.

\begin{table}[h]
\centering
\caption{Consideration Mechanism}
\begin{tabular}{c|c|c|c|c}
\hline
 Consideration Probability & Performance &  T=800 & T=2000 & T=4000 \\
\hline
$\operatorname{Q}_1(1|1)$ & Bias & 0.0022  & -0.0001 & 0.0002 \\
 & RMSE & 0.0723  & 0.0292 & 0.0151 \\ \hline
$\operatorname{Q}_1(1|0)$ & Bias & 0.0006  & 0.0037 & 0.0007 \\
 & RMSE & 0.0674  & 0.0255 & 0.0130 \\
\hline
\end{tabular}
\label{tab5:peer_estimates}
\end{table}

\noindent\textbf{Estimation of Counterfactual CCPs and Choice Rule} We now follow the proof of Proposition~\ref{prop: counterfactual CCP} to identify and estimate the counterfactual CCPs. For instance, we can identify and estimate the counterfactual CCPs if we shrink the menu from $\{0,1,2\}$ to $\{0,1\}$. We denote these counterfactual CCPs as $\operatorname{P}^{*}_1(1|00,\{0,1\})$, where $00$ denotes the previous choices of Agents 2 and 3, respectively, and $\{0,1\}$ indicates the counterfactual menu. Given the model primitives, this counterfactual choice probability is given by
\begin{align*}
\operatorname{P}^{*}_1(1|00,\{0,1\})=\operatorname{Q}_1(1|0) \operatorname{R}_1(1|0,\{0,1\})=1/4\times \underbrace{\frac{\exp{(3)}}{1+\exp{(3)}}}_{\text{choice prob}}\approx 0.2381.
\end{align*}
The proof of Proposition~\ref{prop: counterfactual CCP} implies that 
\begin{align*}
\operatorname{P}^{*}_1(1|00,\{0,1\})= \frac{\operatorname{P}_1(1|20)- \frac{\operatorname{Q}_1(2|1)}{\operatorname{Q}_1(2|0)} \operatorname{P}_1(1|00) }{ 1-\frac{\operatorname{Q}_1(2|1)}{\operatorname{Q}_1(2|0)} } \approx \frac{0.151-3 \times 0.209}{1-3}= 0.2380.
\end{align*}
Hence, up to a numerical error, the results agree. Given the identified and estimated $\operatorname{Q}_1(1|0)$ we can recover
\begin{align*}
    \operatorname{R}_1(1|0,\{0,1\})= \frac{\operatorname{P}_1(1|20)- \frac{\operatorname{P}_1(2|20)}{\operatorname{P}_1(2|00)} \operatorname{P}_1(1|00) }{ \left(1-\frac{\operatorname{P}_1(2|20)}{\operatorname{P}_1(2|00)} \right)\dfrac{\operatorname{P}_1(1|00)}{\operatorname{P}_1(1|11)}}.
\end{align*}
A similar formula can be used to identify the rest of the counterfactual CCPs and the values of the choice rules for each consideration set and agent. We estimate the counterfactual CCP and the choice rules and present the bias and RMSE in Table~\ref{tab6:peer_estimates} for different sample sizes in 1000 replications.

\begin{table}[h]
\centering
\caption{Counterfactual CCP and Choice Rule}
\begin{tabular}{c|c|c|c|c}
\hline
 Probability & Performance &  T=800 & T=2000 & T=4000 \\
\hline
$\operatorname{P}^{*}_1(1|00,\{0,1\})$ & Bias & -0.0041  & 0.0029 & 0.0007 \\
 & RMSE & 0.0713  & 0.0278 & 0.0134 \\ \hline
$\operatorname{R}_1(1|0,\{0,1\})$ & Bias & -0.0153  & -0.0004 & 0.0022 \\
 & RMSE & 0.1949  & 0.0732 & 0.0368 \\
\hline
\end{tabular}
\label{tab6:peer_estimates}
\end{table}

\section{The Empirical Application}
\noindent This appendix offers supplemental material for the empirical application. First, we provide details for the data we use for estimation. Second, we provide extra details about the network estimation we use in the main text. Lastly, we conduct a robustness check of our empirical findings where we allow own stores in nearby markets to affect payoffs. 

\subsection{Data}
\noindent We purchased the tea chain expansion data from CnOpenData, a data marketing company that scraps all the registration data from the National Enterprise Credit Information Publicity system. This system, which is an information-query platform for all types of enterprises (market entities) in the People's Republic of China, was launched by the State Administration for Industry and Commerce of the People's Republic of China in February 2014. Users can employ it to search for enterprise registration and filing details, license approvals, administrative penalties, records of abnormal business operations, and other related information.

Enterprises are required to register at the local (city/district) Administration for Industry and Commerce. For each new store, the enterprise has to register and provide the required information to obtain the approval for operation ---the required information includes the specific street location of the store. The entry dates are the registration dates. If a store is closed, the enterprise is required by law to update this information and the date of the change of status is recorded. The overall framework, document standards, and processing time frames are unified nationwide by law, but individual service windows may apply slightly different procedures.

The market characteristics, including population, Gross Regional Product (GDP) of the city and area, were mostly collected from the China City Statistics Yearbook with two exceptions. First, the yearbook of 2016 - 2020 reported the registered population in the city in 2015-2019, but China conducted its seventh national population census in 2020, so the population reported in the Yearbook of 2021 is the resident population instead of the registered population in 2020. Second, the Yearbook of 2018 only reported GDP in the Districts under City in 2017, excluding the suburban area. We supplemented the missing of the total city GDP in 2017 and the registered population in 2020 through the China Economic and Social Big Data Research Platform, which is a large-scale, integrated statistical database that aggregates China's official economic and social development data from 1949 to the present. It brings together all central-, provincial-, and major municipal-level statistical yearbooks, as well as census reports, survey results and historical statistical compilations, covering 32 sectors and industries of the Chinese economy and society (\url{https://data.oversea.cnki.net/}). We also used the Consumer Price Index from the National Bureau of Statistics of China to convert the GDP into real terms (\url{https://data.stats.gov.cn/english/easyquery.htm?cn=C01}).

\subsection{Estimation Details for the Network Structure}
\noindent The vector of parameters $\theta$ consists of two parts: the parameters of the network structure (i.e., $\mathcal{NC}_a$, $a\in\mathcal{A}$), and the parameters of the attention index and marginal profits. To maximize the likelihood value by searching $\theta$ in its parameter space, we can proceed as follows: In the inner loop, fixing the network structure, we maximize the likelihood function over the consideration and payoff parameters by using the profiled likelihood estimation. The outer loop then searches for the network structure that leads to the highest likelihood. 

The set of parameters in the inner loop ---attention and payoff set of parameters--- is rather standard and does not pose any particular challenge. Unfortunately, checking all possible network structures is often computationally prohibitive without some extra restrictions. In our application, the parameter space for $(\mathcal{NC}_a)_{a\in\mathcal{A}}$, consists of $2^{2 \times 71\times(71-1)}=2^{9940}>10^{2400}$ possible network structures (i.e., there are $2\times71\times(71-1)$ binary variables). Even if we impose full symmetry, then the size of the parameter space drops to $2^{71\times(71-1)/2}=2^{2485}$. 
Instead, to simplify the estimation, we use spatial information about markets. In particular, we assume that if market $m'$ is in the neighborhood of market $m$, then at least one of the following three conditions holds: $m'$ and $m$ are in the same province; the prefectures where $m'$ and $m$ are located share a border; and/or $m'$ is at least the $5$-th closest (in terms of geographical distance) market to market $m$.

With these additional constraints, the number of binary parameters describing the network structure is $563$. Searching through all possible network structures still involves $2^{563}$ possibilities. To further facilitate the empirical analysis, instead of searching every possible network, we start the search from the initial/largest possible network and then shut down one link at a time to find the best improvement of the likelihood. We repeat this procedure until no link shutdown leads to any improvement.\footnote{This heuristic algorithm is a variation of a greedy optimization algorithm. See \citet{kitagawa2023individualized} for a recent application in the context of treatment allocation in sequential network games.} Though this method is only guaranteed to converge to a local optimum, we believe that, in our application, it provides an informative approximation of the solution.

\subsection{Robustness Checks: Own Stores in Nearby Markets Affecting Profits}
\noindent In our benchmark analysis, we assume that only own and rival stores in the focal market can affect the payoff of the firms in the focal market. This rules out spillover effects across different markets due to various channels. A potential channel could be transportation cost saving, though, as we argue in the paper, transportation costs seem minimal in the tea industry. Another potential spillover channel could be due to information aggregation. To incorporate these potential spillover effects, we relax our benchmark assumption and allow own stores in the nearby markets to affect the firm's profitability in the focal market. We still sustain that rival's stores in the nearby markets only affect consideration. That is, the number of rival stores in nearby markets still acts as the excluded variable. We, moreover, assume that the consideration and preference networks for each market coincide.  As a result, the marginal profit and the attention index, respectively, are represented by
\begin{align*}
    \bar{\pi}_{at}(S_{t}, N_{t};\theta)=&S_{mt}\tr \beta_{f}+\sum_{f'}\Big[N_{(f',m)t}\alpha_{f,f'}+N^2_{(f',m)t}\gamma_{f,f'}\Big], \\        +&\left[\sum_{a''\in\mathcal{N}_a:f''=f}N_{a''t}\delta_{f,f}+    \left(\sum_{a''\in\mathcal{N}_a:f''=f}N_{a't}\right)^2\eta_{f,f}\right]  \\
    \bar{\tilde{\pi}}_{at}(S_{t}, N_{t};\theta)=&S_{mt}\tr \tilde\beta_{f} +\sum_{f'}\Big[N_{(f',m)t}\tilde\alpha_{f,f'}+N^2_{(f',m)t}\tilde\gamma_{f,f'}\Big]+\\
    +&\sum_{f'}\left[\sum_{a''\in\mathcal{N}_a:f''=f'}N_{a''t}\tilde\delta_{f,f'}+    \left(\sum_{a''\in\mathcal{N}_a:f''=f'}N_{a't}\right)^2\tilde\eta_{f,f'}\right].  
\end{align*}
We present the estimated consideration and expansion probabilities in Figures~\ref{fig: con hist 12 app} and ~\ref{fig: exp hist 12 app}. We find that the directions of the effects are similar to the ones in the main paper. Both firms displayed limited consideration initially but became almost full consideration at the end of the measurement period because of the increased number of stores in different markets. As in the main paper, the expansion probabilities estimated under full consideration substantially underestimate the profitability of markets. 

\begin{figure}[h!]
\centering
  \includegraphics[width=0.7\textwidth]{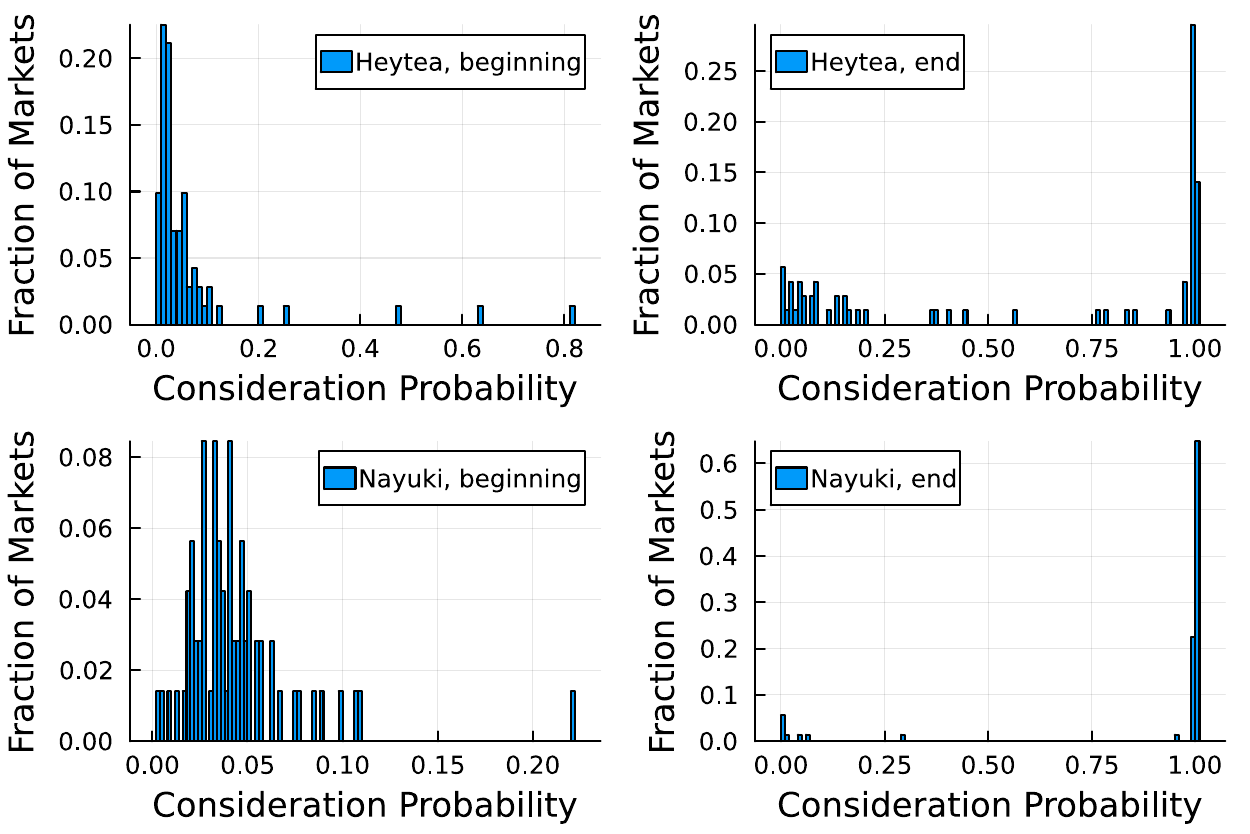}
\caption{\textbf{Normalized histogram of consideration probabilities for both firms at the data's beginning and end.} }
\label{fig: con hist 12 app}
\end{figure}

\begin{figure}[h!]
\centering
  \includegraphics[width=0.7\textwidth]{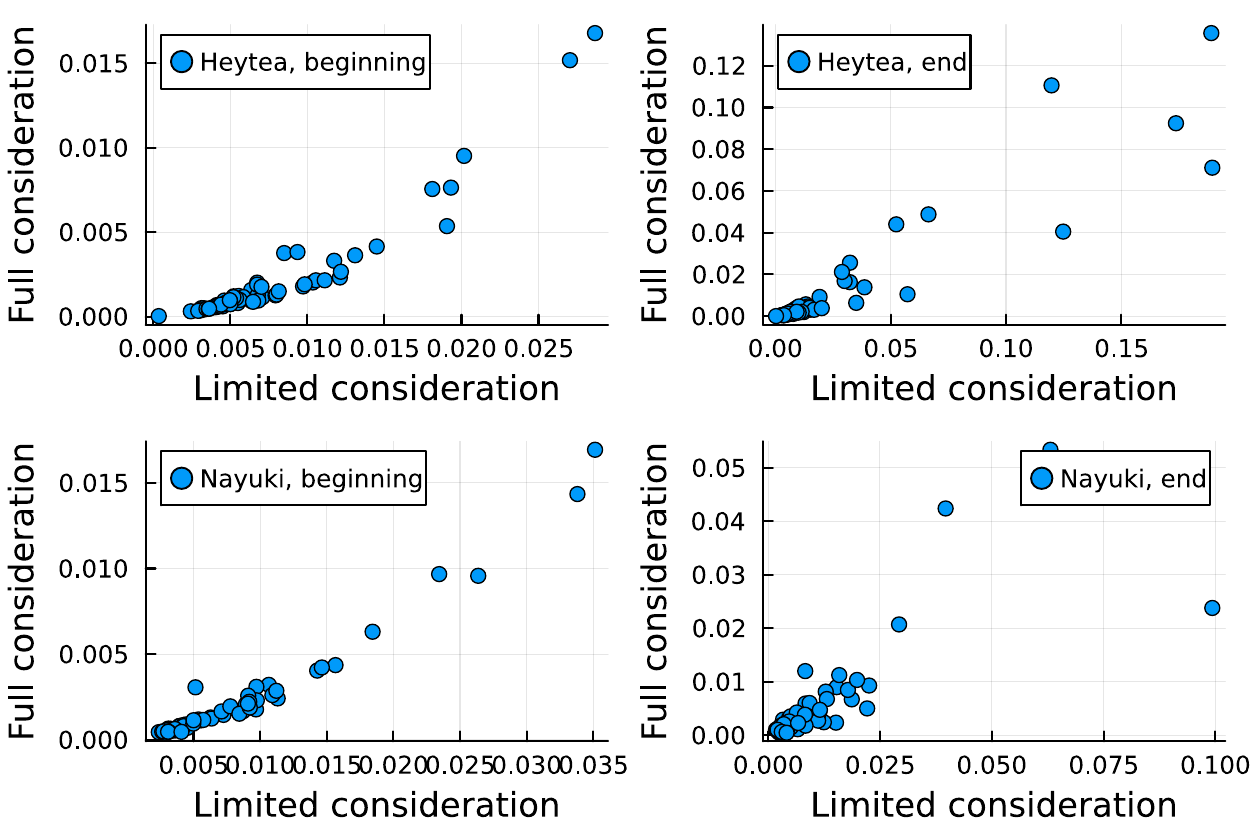}
\caption{\textbf{Limited consideration vs. full consideration expansion probabilities.} }
\label{fig: exp hist 12 app}
\end{figure}

\newpage

\bibliography{references}

\end{document}